\documentclass{article}
\usepackage[square]{natbib}

\usepackage{graphicx}
\usepackage{multirow}
\usepackage{amsmath,amssymb,amsfonts}
\usepackage{amsthm}
\usepackage{mathrsfs}
\usepackage[title]{appendix}
\usepackage{xcolor}
\usepackage{textcomp}
\usepackage{manyfoot}
\usepackage{booktabs}
\usepackage{algorithm}
\usepackage{algorithmicx}
\usepackage{algpseudocode}
\usepackage{listings}

\usepackage{bm}
\usepackage{tikz}
\usepackage{dsfont}
\usepackage{proof}
\usepackage{enumerate}
\usepackage{comment}

\usepackage[colorlinks=true]{hyperref}
\hypersetup{linktocpage,
colorlinks, linkcolor={blue},
citecolor={blue}, urlcolor={blue}
}

\usetikzlibrary{calc, intersections, through, backgrounds, patterns}

 \usepackage{relsize}
 \usepackage[margin=1in]{geometry}
 \newcommand{\nc}{\,\mid\!\sim}
 
\newcommand{\ind}[1]{\mathds{1}_{#1}}
\newcommand {\indT}[1] { \mathds{1}^{T}_{#1} }

\newcommand{\e}[1]{ [\![{#1}]\!]}

\newcommand{\bel}{ \mathsf{B} }
\newcommand{\sel}{ \mathsf{S} }
\newcommand{\kn}{ \mathsf{K} }

\newcommand{\sdom}{\rotatebox[origin=c]{-90}{$\succ$}}
\newcommand{\sdomeq}{\rotatebox[origin=c]{-90}{$\succeq$}}

\newcommand{\RNum}[1]{\uppercase\expandafter{\romannumeral #1\relax}}

\newcounter{Theorem}
\newtheorem{theorem}{Theorem}[section]
\newtheorem{cor}[Theorem]{Corollary}
\newtheorem{lem}[Theorem]{Lemma}

\newtheorem{prop}[Theorem]{Proposition}
\newtheorem{mydef}[Theorem]{Definition}

\newtheorem{observation}[Theorem]{Observation}

\newtheorem{ex}[Theorem]{Example}

\begin{document}

\title{Probabilistically stable revision and comparative probability: \\ a representation theorem and applications }
\author{Krzysztof Mierzewski \\ Carnegie Mellon University }
\date{August 2025}

\maketitle

\begin{abstract}
The stability rule for belief, advocated by \cite{LEI}, is a rule for rational acceptance that captures categorical belief in terms of \emph{probabilistically stable propositions}: propositions to which the agent assigns resiliently high credence. The stability rule generates a class of \emph{probabilistically stable belief revision} operators, which take as input (1) a belief set $B_{\mu}$, obtained through the stability rule from a prior probability measure $\mu$, and (2) new evidence $E$, and output a new all-or-nothing belief set $B_{\mu(\cdot\,|\,E)}$ induced by Bayesian conditioning. These operators capture the dynamics of belief that result from an agent updating their credences through Bayesian conditioning while complying with the stability rule for their all-or-nothing beliefs. In this paper, we prove a representation theorem that yields a complete characterisation of such probabilistically stable revision operators and provides a `qualitative' selection function semantics for the (non-monotonic) logic of probabilistically stable belief revision.
The theorem identifies exactly the selection functions that assign, to each proposition $E$, the logically strongest stable proposition after Bayesian conditioning on $E$. Drawing on the theory of comparative probability orders, this result gives necessary and sufficient conditions for a selection function to be representable as a strongest-stable-set operator on a finite probability space. We exhibit several unusual features of the resulting non-monotonic logic. The stability-induced belief revision operators do not satisfy the AGM belief revision postulates. The logic of probabilistically stable belief revision exhibits strong monotonicity properties while satisfying only very weak forms of case reasoning. In showing the representation theorem, we prove two results of independent interest to the theory of comparative probability: the first provides necessary and sufficient conditions for the joint representation of a pair of (respectively, strict and non-strict) comparative probability orders. The second result provides a method for axiomatising the logic of ratio comparisons of the form ``event $A$ is at least $k$ times more likely than event $B$''. In addition to these measurement-theoretic applications, we point out two applications of our main result to the theory of simple voting games and to revealed preference theory. \\

\noindent \emph{Key words and phrases:} probabilistic stability, acceptance rules, comparative probability, non-monotonic logic, belief revision, probability logics, measurement theory, selection function semantics, choice rules, voting games

\end{abstract}

\newpage 

\tableofcontents

\newpage 
\section{Introduction}

\cite{LEI} offers an acceptance rule based on the notion of \emph{probabilistically stable} hypotheses. Given a prior probability measure, a hypothesis is probabilistically stable if it retains high probability under conditioning on any information that is consistent with it. According to the stability rule, a Bayesian learner \emph{believes} a proposition simpliciter (or \emph{accepts} the proposition) whenever that proposition is logically entailed by the logically strongest stable hypothesis. This stability \emph{rule} for belief (to be distinguished from Leitgeb's later, `non-reductionist' \emph{Humean Thesis on Belief} \citep{LEIBOOK}) associates to each prior $\mu$ a unique consistent and logically closed belief set $B_{\mu}$.

What are the dynamics of categorical belief revision induced by the stability rule? Upon learning new evidence $E$, a Bayesian agent with prior $\mu$ and categorical belief set $B_{\mu}$ will update their prior to a posterior probability measure $\mu(\cdot\,|\, E)$ and, following the stability rule, adopt $B_{\mu(\cdot\,|\,E)}$ as their updated belief set. We thus get a belief revision operator mapping a belief set $B_{\mu}$ and a learned proposition $E$ to a revised belief set $B_{\mu} \ast E:=B_{\mu(\cdot\,|\,E)}$. The resulting class of \emph{probabilistically stable belief revision operators} is of particular interest because, by design, these operators constitute a belief revision policy that \emph{tracks} Bayesian conditioning in the sense of \cite{LIK}: for any probability measure $\mu$ representing the agent's credences, first applying the acceptance rule to obtain a belief set $B_{\mu}$, and then revising this belief set by new evidence $E$, gives the same result as first updating $\mu$ by Bayesian conditioning on $E$ and then applying the acceptance rule to obtain an updated belief set $B_{\mu(\cdot\,|\,E)}$. In a nutshell, tracking means that acceptance followed by qualitative revision agrees with Bayesian conditioning followed by acceptance. 

\cite{LIK} have shown that, under plausible assumptions, well-known belief revision operations, including all AGM belief revision policies, cannot track Bayesian conditioning. This leads to the \emph{tracking problem}, which asks for a minimally well-behaved acceptance rule and categorical belief revision policy which tracks Bayesian conditioning, allowing for a harmonious co-existence of credences with categorical beliefs. Given that probabilistically stable belief revision tracks Bayesian conditioning, an independent axiomatic characterization of these belief revision operators would naturally constitute one possible solution to the tracking problem. 

In this paper, we investigate the logic of probabilistically stable belief revision and provide such a solution. Probabilistically stable revision operators can be captured using selection function models, with corresponding non-monotonic consequence relations obtained through a probabilistic Ramsey test. The problem of characterising the class of probabilistically stable revision operators, and thus the appropriate selection function models, amounts to finding an axiomatic description of \emph{strongest-stable-set operators}, which send every event $E$ to the strongest stable set conditional on $E$: that is, they map the update input $E$ to the (unique) logically strongest stable event after updating the prior by $E$.  

Here we solve the characterisation problem by proving a probabilistic representation theorem for strongest-stable set operators. This result, which draws on the theory of comparative probability orders, gives necessary and sufficient conditions for a selection function to be representable as a strongest-stable-set operator on a finite probability space. It yields a complete characterisation of probabilistically stable revision operators, and thus provides qualitative semantics for the logic of probabilistically stable belief. Along the way, we exhibit some unusual features of the resulting non-monotonic logic: probabilistically stable belief revision validates strong monotonicity properties while failing even weak forms of case-reasoning (except for a very weak version of the \textsf{Or}-rule from non-monotonic logic). 

Our main results also admit a number of applications in other areas. The first application concerns a problem in the theory of comparative probability orders: in showing the representation theorem, we prove a general result giving necessary and sufficient conditions for the joint representation of a pair of (respectively, strict and non-strict) comparative probability orders, answering an open question raised by \cite{KON} (about conditions for the strong representability of comparative probability orders). The second application is a general method for axiomatising the logic of probability ratio comparisons of the form ``event $E$ is at least $q$ times more likely than event $F$'' for any fixed rational number $q$. The third application concerns the theory of revealed preference and rational choice: our representation theorem identifies the choice functions of cautious agents who consider an option acceptable only if it has higher utility than the cumulative utility of all unacceptable options combined (more generally: $q$ times higher utility for a fixed rational number $q$). The fourth application concerns the theory of simple voting games \citep{TAY}: our result gives necessary and sufficient conditions for the simultaneous numerical representation of a collection of simple voting games. Equivalently, it characterises the choice functions which pick out the smallest stably decisive coalitions in a weighted voting game.  

We conclude with some reflections on the `qualitative' nature of our semantics. In investigating bridge principles between probabilistic and categorical belief, one often aims for an account of categorical belief, and belief revision, that admits a `purely qualitative' characterisation---roughly, one that does not explicitly appeal to probabilistic or quantitative notions. In a sense, our selection-function characterisation of probabilistically stable belief revision fits the bill; yet, it is also of a much less `qualitative' flavour than usual accounts of logical belief revision, since it involves intricate combinatorial patterns of reasoning that can be seen as implicitly encoding some probabilistic constraints. For example, it is easy to see that probabilistically stable revision cannot be represented in terms of a plausibility order on states, as is usually done in belief revision and non-monotonic logic. Instead, its axiomatisation depends on a variant of Scott's cancellation axioms from measurement theory \citep{SCO}, capturing a non-trivial combinatorial property which guarantees the probabilistic representability of the underlying selection function. In the last section, we suggest that one can fruitfully apply the framework of probability logics to gauge how much quantitative/probabilistic structure an acceptance rule depends on, on the basis of its definitional complexity, and evaluate how much probabilistic structure is implicitly encoded in the behaviour of the resulting revision operator. 

\section{Probabilistic stability and belief revision}

Bayesian orthodoxy has it that the belief state of an agent can be represented as a probability measure which captures the agent's credences, or degrees of belief. Accordingly, we will model the subjective probabilities of an agent by a probability measure $\mu$ over a fixed finite algebra of propositions. We work with finite probability spaces $(\Omega,\mathfrak{A},\mu)$, with $\mathfrak{A}$ a finite set algebra over a sample space $\Omega$, and $\mu$ a probability measure on $\mathfrak{A}$. We represent propositions $A,B,..., X,Y,Z$ as elements of the set algebra $\mathfrak{A}$. We let $\Delta_{\mathfrak{A}}$ denote the set of all probability distributions on $\mathfrak{A}$. The \emph{conditional probability} of $A$ given $B$ is $\mu(A\,|\,B)=\frac{\mu(A\cap B)}{\mu(B)}$. In what follows, we will sometimes write the conditional probability measure $\mu(\cdot\,|\,B)$ in parametric form as $\mu_{B}$. 

How do a rational agent's categorical, all-or-nothing beliefs relate to the agent's subjective probabilities? We will call a rule for rational acceptance\textemdash or, more simply, an \emph{acceptance rule}\textemdash any function that specifies, given an agent's subjective probability function $\mu$, the agent's set of categorical, all-or-nothing beliefs. In other words, we take the term \emph{acceptance rule} as presupposing a functional dependence of belief on credence, possibly with other parameters specified in the background (the rule can be further parametrised, e.g. by a threshold parameter). 

A commonplace norm on categorical belief is that the belief set of a rational agent be logically closed. Assuming the logical closure of belief, we can simply identify the agent's belief set with the conjunction of all their believed propositions, which constitutes the logically strongest believed proposition. We can then recover the believed propositions as exactly those that are logically entailed by the logically strongest believed proposition. In this way, it will be useful to operate under the following general notion of an acceptance rule:

\begin{mydef}[\textbf{Acceptance rules}]
    An \emph{acceptance rule} $\alpha$ is a map $\alpha:\Delta_{\mathfrak{A}}\to \mathfrak{A}$. The interpretation is that $\alpha$ maps each probability distribution $\mu$ in $\Delta_{\mathfrak{A}}$ to the logically strongest accepted proposition $\alpha(\mu)\in\mathfrak{A}$; we then say an agent with credences $\mu$ accepts (or `believes') a proposition $X\in\mathfrak{A}$ if and only if $\alpha(\mu)\subseteq X$. 
\end{mydef}

One of the most discussed acceptance principles is the Lockean Thesis, which recommends the acceptance of all and only propositions that are above some threshold $t$ fixed in advance. The substantial cost of this principle, highlighted by the Lottery Paradox \citep{KYB}, is that one must give up on the general requirement to have logically closed and consistent belief sets. Leitgeb's stability theory of belief proposes some amendments to the Lockean Thesis. According to \cite{LEI}, the key to categorical belief is the property of \emph{probabilistic stability}: 

\begin{mydef}[\textbf{Stability}]\label{Stability}
Let $(\Omega,\mathfrak{A},\mu)$ a probability space and $t\in[0.5,1]$. A set $X\in\mathfrak{A}$ is $(\mu,t)$\emph{-stable} if and only if $\forall \,E\in\mathfrak{A}$ such that $X\cap E\neq\emptyset$ and $\mu(E)>0$, $\mu(X\,|\,E)> t$.
\end{mydef}

A hypothesis $H$ is probabilistically stable for a threshold $t\geq 1/2$ if it has resiliently high probability, in the sense that it retains above-threshold probability after updating by any new evidence, as long as that evidence does not logically exclude that hypothesis $H$. According to the \emph{stability rule for belief}, introduced by \cite{LEI}, an agent's logically strongest belief\textemdash the conjunction of all of their beliefs\textemdash should be resilient to new information in that sense. More precisely, the rule requires that a Bayesian agent believe all and only the logical consequences of the logically strongest stable proposition in the algebra of propositions over which they have credences.

\begin{mydef}[\textbf{The} $\tau$\textbf{-rule}]\label{tau rule}
Fix a finite probability space $(\Omega,\mathfrak{A},\mu)$ and a threshold $t\in(0.5,1]$. We define the \emph{stability rule for threshold} $t$, written $\tau_t$, as 
$$
\tau_{t}(\mu):=\min_{\subseteq}\{S\in\mathfrak{A}\,|\, S \text{ is $(\mu,t)$-stable}\}
$$
And the \emph{belief set generated by $\tau$} is given by: 
\begin{align*}
B_{\tau}(\mu)&:= \{X\in\mathfrak{A}\,|\,  \tau_t(\mu) \subseteq  X \}  \\
&\phantom{:}=\{X\in\mathfrak{A}\,|\,  A\subseteq X \text{ for some $(\mu,t)$-stable }A\in\mathfrak{A}\}  
\end{align*}
\end{mydef}

The key to understanding this definition is the following observation: one can show, given a (finite) probability space and a stability threshold, that the probabilistically stable propositions/events always form a nonempty, well-ordered collection under entailment/inclusion \citep{LEI}: in particular, on a finite probability space, there always exists a logically strongest probabilistically stable proposition $\tau(\mu)$. Given this, we can justify the second equality in the definition: a proposition being entailed by the strongest probabilistically stable proposition $\tau(\mu)$ is equivalent to its being entailed by \emph{some} probabilistically stable proposition. So we can equivalently characterise the $\tau$ rule as recommending that the agent believe exactly those propositions that admit a probabilistically stable deductive justification. Since the belief set thus defined is closed under logic, we can identify the belief set of the agent with the intersection of all accepted propositions $\bigcap B_{\tau}(\mu)$, which simply coincides with $\tau(\mu)$. 

The stability rule requires that the conjunction of the agent's beliefs be probabilistically stable, and therefore resilient in the face incoming information. Leitgeb offers several distinct arguments for this requirement, based on the intrinsic merits of probabilistic stability \citep{LEI2}, its connection to the Lockean Thesis for belief \citep{LEIBOOK} and with the AGM theory of belief revision \citep{LEI}, or its properties as an error-minimizing method of qualitatively approximating probabilistic credal states \citep{LEI3, LEIBOOK}. 
 
 This stability requirement forms the core of Leitgeb's  \emph{Humean Thesis} on belief \citep{LEIBOOK}. The Humean thesis is more permissive than the stability rule: it insists on the probabilistic stability requirement for belief, but without requiring one's belief set to be generated by the \emph{strongest} stable proposition. \cite{LEIBOOK} provides several equivalent characterizations of the Humean Thesis, one of which characterizes the belief sets satisfying the Humean Thesis as consisting of the propositions that have resiliently high probability upon conditioning on propositions that are \emph{not disbelieved}. Leitgeb then shows that this is equivalent to satisfying the Lockean Thesis with respect to a threshold that is equal to the measure of a probabilistically stable proposition. This is one sense in which the Humean Thesis is a refinement of the Lockean Thesis: the choice of \emph{which} stable set one picks as strongest accepted proposition amounts to a choice of Lockean threshold. On that account, stability provides a constraint on belief/credence pairs without reducing belief to credence (even once the stability parameter $t$ is fixed). The stability rule complies with the Humean Thesis, but is more restrictive: it uniquely specifies the agent's beliefs, given their credences and threshold $t$. Thus Leitgeb's theory of belief comes in two flavours: the more permissive Humean Thesis, and its `reductionist' variant captured by the stability rule \citep{LEI}. The stability rule requires in addition that the conjunction of the agent's beliefs be the \emph{logically strongest} probabilistically stable proposition. One way to motivate this requirement is to see it as an approximation to the Lockean thesis for the fixed threshold $t$. Suppose the threshold $t$ captures the agent's (contextually determined) standards for what counts as a high-probability proposition. We get the stability rule by requiring that one should satisfy as many instances of the Lockean thesis for threshold $t$ as is possible\textemdash that is, believe as many high-probability (above threshold $t$) propositions as possible\textemdash without violating the stability requirement \citep{LEI, MIE}. A slightly different way to arrive at the stability rule is as follows. A probabilistically stable event is one that has resiliently high probability: one that remains above the Lockean threshold $t$ even after learning any information that does not logically rule it out. If we take a proposition having resiliently high probability as sufficient for belief, and require closure of belief under logic, one gets the stability rule. Lastly, we may simply characterise the stability rule $\tau$ as capturing exactly the \emph{boldest} probabilistically stable belief set that an agent might adopt \citep{LEI}. Another motivation for investigating the stability rule, more germane to the main topic of this paper, comes from considerations having to do with the \emph{dynamics} of belief.

\subsection{The tracking problem}\label{The tracking problem}

Any satisfactory account of the relation between credences and plain beliefs must also account for belief dynamics. Given a bridge principle between belief and credence, does a rational agent's policy for belief revision harmonise with their policy for revising their credences? This is the question that \cite{LIK} raise. Here the objects of interest are \emph{belief revision policies}. A \emph{revision plan for a belief set $B$} is a map $E\mapsto B^{\ast}E$, mapping each proposition $E\in\mathfrak{A}$ (the revision input) to a revised belief set $B^{\ast}E$. (Given that we are representing belief sets by a single proposition, we take the $B^{\ast}E$ to be the strongest accepted proposition after revision by $E$). A belief revision policy \emph{based on} an acceptance rule $\alpha$ associates, for each probability measure $\mu$, the belief set $\alpha(\mu)$ together with a revision plan for $\alpha(\mu)$. \cite{LIK} propose the following criterion of compatibility between Bayesian conditioning and one's adopted belief revision policy:

\begin{mydef}[\textbf{Tracking}]\label{Tracking}
A belief revision policy $\ast$ based on an acceptance rule $\alpha$ \emph{tracks Bayesian conditioning} if and only if, for any finite probability space $(\Omega, \mathfrak{A}, \mu)$ and $E\in\mathfrak{A}$ such that $\mu(E)>0$, we have $\alpha(\mu_{E})=\alpha(\mu)^{\ast}E$. 
\end{mydef}

Tracking requires that, for any probability measure, Bayesian update followed by acceptance yields the same belief set as acceptance followed by (categorical) belief revision. In other words, tracking requires that the belief set obtained by updating one's prior by new evidence coincides with the result of revising the initial belief set obtained from one's prior by the same evidence. To put it in yet another way: if an agent has a prior $\mu$ and a belief set $B$ which are compatible with the (fixed) acceptance rule, in the sense that $B$ is the belief set obtained from $\mu$ via the acceptance rule, then for any new evidence $E$ the agent may condition on, the revised belief set $B^{*}{E}$ remains compatible with their posterior $\mu_E$. We may then say that, given the acceptance rule $\alpha$, the belief revision policy is a `qualitative' counterpart of Bayesian conditioning: it tracks exactly the effects of Bayesian conditioning on one's categorical beliefs. The \emph{tracking problem} is the problem of identifying a well-behaved acceptance rule and a belief revision policy based on it which tracks Bayesian conditioning. 

The question of how to track Bayesian conditioning using Leitgeb's stability rule is a particularly salient one, due to a close connection between the stability rule and one of the most canonical accounts of logical belief revision: the AGM theory of belief revision \citep{AGM}. One can very naturally associate to the stability rule a belief revision policy which complies with the well-known AGM postulates \citep[Chapter 4]{LEIBOOK}. The idea is simple: given the agent's prior probability measure, the well-ordered collection of probabilistically stable sets constitutes a system-of-spheres in the sense of \cite{GRO} (equivalently, a total preorder on states) which are well-known to generate AGM belief revision operators. We can think of the strongest stable proposition as containing the most plausible worlds, given the agent's prior: and the various stable sets as constituting a system of `fallback' belief states for the agent to land on in case they were to learn a proposition which contradicts some of their beliefs. Call this the stability-induced AGM revision policy.\footnote{Formally: given a prior probability $\mu$, the belief set of the agent is given by the logically strongest stable set $\tau(\mu)$ (equivalently, the $\subseteq$-minimal stable set) and the revision plan for it is given by $\tau(\mu)^*E:=S\cap E$ where $S$ is the logically strongest stable set consistent with $E$.}

This connection may inspire the hope that the stability rule gives us not only an account of the relation between credences and plain beliefs, but also one that yields a perfect dynamic compatibility between Bayesian conditioning and AGM belief revision. This is not so. The stability-induced belief revision policy fails to track Bayesian conditioning. This follows from a more general impossibility theorem by \cite{LIK}, which shows that no `sensible' acceptance rule (satisfying certain minimal conditions) can allow AGM belief revision to track Bayesian conditioning. 

In the light of this negative result, several responses have been offered. \cite{LIK} reject both the stability rule and the AGM belief revision postulates, and propose an alternative acceptance rule (the odds-threshold rule, inspired by \cite{Levi1996}) and an alternative belief revision policy (Shoham-driven revision) that together track Bayesian conditioning. Another approach is to retain the stability rule and show that there is nonetheless a weaker sense in which AGM revision is compatible with Bayesian conditioning: it is shown in \cite{MIE} that AGM reasoners can be rationalised as Bayesian agents who comply with the stability rule and who rely on the most equivocal (maximum entropy) probabilistic representation of their qualitative beliefs.\footnote{The idea in a nutshell: start with a qualitative characterisation of the agent's belief state (a belief set or a total preorder, corresponding to a system of spheres). Among all the probability measures that, according to the stability rule, are compatible with this belief state, take the (necessarily unique) maximum entropy probability measure. To revise, apply Bayesian conditioning to that maximum entropy representative and apply the stability rule to obtain a revised belief state. The resulting operation is always an AGM belief revision operator. In that sense AGM belief revision operators emerge, by an application of the \emph{principle of maximum entropy}, from Bayesian conditioning and Leitgeb's stability rule.} 

Leitgeb's own preferred approach in \cite{LEIBOOK} is to weaken the stability rule itself, and instead endorse the Humean Thesis, which only insists that the strongest belief of the agent be given by a probabilistically stable set (which need not be the logically strongest one). This version of the Humean Thesis turns out to render Bayesian conditioning compatible with AGM revision, albeit in a weaker sense: this is simply because the stability-induced AGM policy mentioned above always yields revised belief sets which are probabilistically stable with respect to the agent's posterior probability measure \textemdash satisfying the stability requirement imposed by the Humean thesis\textemdash even though the revised belief sets need not be the logically strongest ones, as required by the stability rule. The Humean Thesis thus can be seen as a weakening of the stability rule which allows for dynamic compatibility between AGM revision and Bayesian conditioning. Nonetheless, this comes at a certain cost: for one thing, this move requires the relationship between credence and belief to be ruled by \emph{two} distinct threshold parameters: one threshold for \emph{stability} (which Leitgeb suggests that we can always make equal to 1/2) and another \emph{Lockean} threshold which specifies a necessary and sufficient probability threshold for the acceptance of any proposition, and \emph{which one must allow to shift during the revision in order to make one's revised beliefs compatible with the result of AGM revision}.\footnote{One might take the decoupling of the stability and Lockean thresholds either as a feature of as a bug; and we shall not take position on this here. See a discussion of this point in \citep[Chapter 4]{LEIBOOK} and \citep{MIE}).} It is not clear how to motivate this shift in threshold as rational unless one independently endorses the AGM postulates as desiderata for rational belief revision, which is indeed the path recommended in \citep[Chapter 4]{LEIBOOK}. The Humean Thesis neither dictates what one's belief set should be, given one's credences (even once a stability threshold is fixed), nor does it dictate what one's revised belief set ought to be: but it certainly \emph{allows} the agent to adopt an AGM-compliant belief revision policy. (At one extreme, nothing in the Humean stability constraint on belief rules out the possibility of shifting to a Lockean threshold of 1 after an update, and adopt as strongest belief the strongest proposition of probability one). For better or worse, in the resulting picture, one leaves behind the tighter\textemdash but also more `reductionist'\textemdash connection required by tracking, under which one's belief set \emph{and} one's belief revision policy arises exclusively \emph{from} one's acceptance rule (and Bayesian conditioning). 

Another route suggests itself. It consists in retaining the stability rule and asking what belief revision policies the stability rule imposes, if one wants to track Bayesian conditioning. Taking that route amounts to addressing a question that, quite aside from its importance as an alternative solution to the tracking problem, is of independent logical interest: it is the question of characterizing the belief revision operators that are induced by stability rule and Bayesian conditioning. This is the main question we will answer in this paper. 

\subsection{Probabilistically stable belief revision}

In the discussion of the tracking problem for AGM belief revision, we were interested in bridging an independently 
characterised class of revision operators with Bayesian conditioning. Now, rather than asking which acceptance rules can bridge a chosen class of belief revision operators with Bayesian reasoning, one can instead study the qualitative revisions that emerge from Bayesian conditioning as a result of adopting a given acceptance principle. Indeed, any acceptance rule generates a belief revision policy that automatically tracks Bayesian conditioning: given your prior probability, simply define your revised belief set to be whatever the acceptance rules tells you to accept after Bayesian conditioning. 

\begin{mydef}[\textbf{Acceptance-induced belief revision policies}]
Let $(\Omega,\mathfrak{A},\mu)$ a probability space and $\alpha$ an acceptance rule.  The \emph{$\alpha$-induced belief revision plan generated by $\mu$} is a map $r^{\alpha}_{\mu}:\mathfrak{A}\to\mathfrak{A}$ defined by $r^{\alpha}_{\mu}(E):= \alpha(\mu_{E})$. The \emph{$\alpha$-induced belief revision policy} is given by the map $r^{\alpha}: \Delta_\mathfrak{A}\times \mathfrak{A}\to \mathfrak{A}$, defined as $r^{\alpha} (\mu, E):= r^{\alpha}_{\mu}(E)$.\footnote{I.e., $r^{\alpha}_{\mu}$ is the projection $r^{\alpha}(\mu,\cdot):\mathfrak{A}\to\mathfrak{A}$. Here we let the map be \emph{partial} and assume that $r^{\alpha}_{\mu}$ is defined only for those events $E\in\mathfrak{A}$ for which $\mu(E)>0$.}
\end{mydef}

Given a prior $\mu$ and an acceptance rule $\alpha$, the induced belief revision plan $r^{\alpha}_{\mu}$ captures an agent's belief set with a plan for revising it. The proposition $r^{\alpha}_{\mu}(E)$ captures the agent's revised belief state after learning $E$: it is given simply by the strongest believed proposition after updating by $E$ and applying the acceptance rule $\alpha$. We assume each such revision plan is defined exactly for those events that have positive $\mu$-measure. The initial beliefs of the agent are given by the belief set obtained from their prior: this is given by $r^{\alpha}_{\mu}(\Omega) = \alpha (\mu)$\textemdash their beliefs conditional on the tautologous proposition $\Omega$, which triggers no revision. The $\alpha$-induced belief revision policy is the policy that, given a prior $\mu$, implements the belief revision plan that $\alpha$ induces for $\mu$. We can also think of it as giving, for each probability measure $\mu$, the structure of the agent's \emph{conditional beliefs} induced by conditionalization, with the initial belief set $\alpha(\mu)$ capturing their unconditional prior beliefs. Any $\alpha$-induced belief revision policy automatically tracks Bayesian conditioning. 

In this vein, suppose we keep Leitgeb's stability rule fixed and ask: what is the qualitative revision generated by Leitgeb's rule and Bayesian conditioning? We employ Leitgeb's rule to obtain a belief revision policy \emph{from} Bayesian conditioning. Starting from a probability distribution $\mu$ and new update input $E$ (with $\mu(E)>0$), consider the (restricted) revision operator ${\ast}_{\tau}$ that takes as input the proposition $E$ and the current belief state $\tau(\mu)$, and outputs the revised belief state $\tau(\mu){^{\ast}}E := \tau(\mu_{E})$, as illustrated in Figure \ref{NONMON}. We call these induced operators \emph{probabilistically stable revision operators} (see Figure \ref{NONMON}). These qualitative revision operators capture exactly the belief revision plans that Bayesian conditioning generates via the stability rule. 

\begin{mydef}[\textbf{Probabilistically stable revision operators}]
Let $(\Omega,\mathfrak{A},\mu)$ a probability space and $t\in[0.5,1]$. The \emph{probabilistically stable revision operator} generated by $\mu$ and $t$ is a map $\sigma_{\mu,t}:\mathfrak{A}\to\mathfrak{A}$ defined by $\sigma_{\mu,t}(E):= \tau_{t}(\mu_{E})$. 
\end{mydef}

\begin{figure}[t]
\centering
\begin{tikzpicture}[scale=3]
\node (P) {$\mu$};
\node (B) [node distance=6cm, right of=P]{$\mu_{\scriptscriptstyle E}$};
\node (A) [node distance=2cm, below of=P] {$\tau(\mu)$};
\node (C) [node distance=6cm, right of=A] {$\tau(\mu_{\scriptscriptstyle E})$};
\draw[->] (P) to node [above]{$|_{\scriptscriptstyle E}$} (B);  
\draw[->] (P) to node [below]{Bayesian conditioning} (B);

\draw[->] (P) to node [left]{Leitgeb rule} (A);

\draw[->] (P) to node [right]{$\tau$} (A);

\draw[dashed, ->] (A) to node  [below]{generated revision} (C);
\draw[->] (B) to node [right]{$\tau$} (C);

\node(G)[node distance=3cm, right of=A]{};
\node(H)[node distance=1.5cm, below of=G]{Resulting consequence relation: $E\nc_{\tau}X$ iff $\tau(\mu_{E})\vdash X$};

\end{tikzpicture}
\caption{Probabilistically stable revision generated by Leitgeb's rule.}
\label{NONMON}
\end{figure}

The agent's (unconditional) beliefs are given by $\sigma_{\mu,t}(\Omega) =\tau_{t}(\mu_{\Omega})= \tau_{t}(\mu)$. 
Probabilistically stable revision is defined to track Bayesian conditioning. Keeping the threshold $t$ implicit, the $\tau$-generated revisions take the form $\tau(\mu)\mapsto\tau(\mu_{E})$, where $\mu$ is the agent's subjective probability distribution and $E$ a new revision input. Keeping the initial doxastic state $\tau(\mu)$ implicit in the background, we can fully characterise each such revision as a map $E\mapsto \tau(\mu_{E})$ (sending a proposition $E\in\mathfrak{A}$\textemdash the revision input\textemdash to another proposition $\tau(\mu_{E})$\textemdash the strongest accepted proposition, representing the updated belief state). Each such map can be seen as a \emph{strongest-stable-set operator}, sending each $E$ to the strongest ($\subseteq$-least) stable set given $E$. 

An important step towards solving the representation problem is to characterise those revisions in a purely qualitative way: that is, to describe all maps $E\mapsto \tau(\mu_{E})$ in a way that does not depend on the underlying probability measures $\mu$. By providing such a characterization of probabilistically stable revision operators, we will obtain a tracking result for Bayesian conditioning which relies on Leitgeb's stability-based acceptance.

\section{The logic of Leitgeb acceptance}\label{LogLei}

Before we turn the representation problem proper, it will first be useful to situate the problem by comparing it with the approach proposed by \cite{GEO, LIK}, which consists in characterizing the non-monotonic consequence relations (alternatively, the logic of flat conditionals) that are induced by an acceptance rule via a probabilistic version of the Ramsey test for conditionals (\S \ref{The Ramsey test and tau-models}). This perspective will offer a helpful starting point towards characterising probabilistically stable revision operators: we will compare the non-monotonic logic generated by the $\tau$-rule with well-known systems from the nonmonotonic logic literature (\S \ref{Some preliminary observations}). We note two properties of the $\tau$-rule which make the resulting logic rather unusual: the logic of probabilistically stable revision validates Rational Monotonicity, while it does not validate the \textsf{Or} rule. We then turn to a discussion of the representation problem. Our task is to find a class of purely `qualitative' (non-probabilistic) structures that capture the behaviour of probabilistically stable revision operators, so as to reveal the key structural properties of probabilistically stable reasoning. We will motivate the use of \emph{selection structures} from non-monotonic logic for this purpose (\S \ref{QualiMod}). The problem of characterising this class of models amounts to finding an axiomatic description of \emph{strongest-stable-set operators}, which send every event $X$ to the strongest stable set given $X$ (that is, they map the update input $X$ to the logically strongest event in the probability space that is probabilistically stable once the prior has been updated by $X$). In order to achieve this, we first discuss the geometry of the stability rule: a simple geometric analysis highlights some structural properties of the stability rule that play an important role in the axiomatic description of strongest-stable-set operators (\S \ref{The geometry of probabilistically stable revision}). This will pave the way for our main results: the representation theorems for probabilistically stable revision operators (\S \ref{secrep}).

\subsection{The Ramsey test and the non-monotonic logic of uncertain acceptance}\label{The Ramsey test and tau-models}

The question we are investigating here is a particular case of the following general problem: given an acceptance rule $\alpha$, what is the class of revision plans ${r}^{\alpha}:E\mapsto \alpha(\mu_{E})$ generated by it? A well-known approach to characterising structural features of belief revision operators is to ask what non-monotonic logic they generate. One way to investigate acceptance-induced belief revision plans is to study their associated nonmonotonic consequence relations $\nc_{\mathfrak{M}}$ understood via a Ramsey-test semantics. Fix a classical propositional language $\mathcal{L}$ in which we represent the propositions that the agent is reasoning about. For given formulas $\varphi$ and $\psi$, say that $\varphi\nc_{\mathfrak{M}}\psi$ holds if and only if, given the agent's (probabilistic) credal state given by a probabilistic structure $\mathfrak{M}$, conditioning on $\varphi$ leads, through the acceptance rule $\alpha$, to a new doxastic state where the agent believes $\psi$. Equivalently, this means that applying the generated revision $r^{\alpha}$ on input $\varphi$ leads the reasoner to accept $\psi$. More precisely, the models of such a logic are given by the following structures, introduced by \cite{GEO}: 

\begin{mydef}[\textbf{Probabilistic $\alpha$-models}]
Fix an acceptance rule $\alpha$.  Define an \emph{$\alpha$-model} as a structure $\mathfrak{M}:=(\Omega,\mathfrak{A},\mu, \e{\cdot}, \alpha)$, where $(\Omega,\mathfrak{A},\mu)$ is a probability space and $\e{\cdot}:\mathcal{L}\rightarrow\mathfrak{A}$ is a Boolean valuation. Set
\begin{center}
$\varphi\nc_{\mathfrak{M}}\psi$  if and only if $\alpha(\mu_{\e{\varphi}})\subseteq\e{\psi}$ or $\mu(\e{\varphi})=0$
\end{center}
\end{mydef}

\noindent Given an acceptance rule $\alpha$, the consequence relations of the form $\nc_{\mathfrak{M}}$ (where $\mathfrak{M}$ is an $\alpha$-model) provide a description of $\alpha$-generated revision in the sense above: they characterise conditional belief statements of the form $\varphi\nc\psi$, expressing that $\psi$ is believed after applying $\alpha$-generated revision by $\varphi$. In this sense, to capture the class of consequence relations of this form is to capture the logic of Bayesian conditional belief generated by $\alpha$. 

At this point, it is useful to recall some key elements of the nonmonotonic logic framework, as introduced by \cite{KLM}. Here, `logics' are identified with a class of consequence relations, and the idea is to classify consequence relations $\nc$ via the collection of inference rules under which $\nc$ is closed. 

\begin{mydef}[\textbf{KLM-style nonmonotonic logics}]\label{Non-monotonic logics DEF}
System \emph{\textsf{C}} consists of the rules \emph{(\textsf{Ref})}, \emph{(\textsf{Left Equivalence})}, \emph{(\textsf{Right Weakening})}, \emph{(\textsf{Cut})} and \emph{(\textsf{CM})} below. System \emph{\textsf{P}} is obtained from System \emph{\textsf{C}} by adding the \emph{(\textsf{Or})} rule. System \emph{\textsf{R}} is obtained from \emph{$\textsf{P}$} by adding \emph{(\textsf{RM})}. System $\textsf{O}$ consists of the rules (\textsf{Ref}), (\textsf{Left Equivalence}), (\textsf{Right Weakening}), (\textsf{WAnd}) and (\textsf{Wor}), as well as (\textsf{VCM}). (Note that $\vdash$ denotes classical entailment here).\\
$$
\infer[\emph{(\textsf{Ref})}] {\varphi\nc\varphi}{}
$$

$$
\infer[\emph{(\textsf{Left Equivalence})}] {\psi\nc\gamma}{\varphi\dashv\vdash\psi & \varphi\nc \gamma}
\qquad
\infer[\emph{(\textsf{Right Weakening})}] {\varphi\nc\gamma}{\varphi\nc\psi & \psi\vdash\gamma}
$$
$$
\infer[\emph{(\textsf{Cut})}] {\varphi\nc\gamma}{\varphi\wedge\beta\nc\gamma & \varphi\nc\beta}
$$
$$
\infer[\emph{(\textsf{And})}] {\varphi\nc\psi\wedge\gamma}{\varphi\nc \psi & \varphi\nc \gamma}
\qquad
\infer[\emph{(\textsf{Or})}] {\varphi\vee\psi\nc\gamma}{\varphi\nc\gamma & \psi\nc\gamma}
$$

$$
\infer[(\textsf{WAnd})] {\varphi\nc\psi\wedge\gamma}{\varphi\nc \psi & \varphi\wedge\neg\gamma\nc \gamma}
\qquad
\infer[(\textsf{WOr})] {\varphi\nc\gamma}{\varphi\wedge\psi\nc\gamma & \varphi\wedge\neg\psi\nc\gamma}
$$

$$
\infer[\emph{(\textsf{CM})}] {\varphi\wedge\psi\nc\gamma}{\varphi\nc \psi & \varphi\nc \gamma}
\qquad
\infer[\emph{(\textsf{RM})}] {\varphi\wedge\psi\nc\gamma}{\varphi\nc \gamma & \varphi\not\nc\neg\psi }
$$

$$
\infer[(\textsf{VCM})] {\varphi\wedge\psi\nc\gamma}{\varphi\nc \psi \wedge \gamma}
$$
\end{mydef}

In this setting, we say that an acceptance rule $\alpha$ \emph{validates} an inference rule if and only if, for any $\alpha$-model $\mathfrak{M}$, the consequence relation $\nc_{\mathfrak{M}}$ is closed under the inference rule. We have a basic notion of derivability: let $\Gamma\cup\{\Phi\}$ a finite set of flat conditionals of the form $\varphi\nc\psi$. Given a system of inference rules $\textsf{S}$, the notation $\Gamma \vdash_{\textsf{S}} \Phi$ means that the conditional $\Phi$ is derivable from $\Gamma$ using (finitely many applications of) the rules from the system $\textsf{S}$. Similarly, we say that a class of probabilistic models $\mathcal{M}$ validates this inference (written $\Gamma \vDash_{\mathcal{M}} \Phi$) whenever it is the case that, for any model $\mathfrak{M}$ in $\mathcal{M}$, if $\nc_{\mathfrak{M}}$ validates all conditionals in $\Gamma$, it also validates the conditional $\Phi$. In this framework, it is natural to ask what the logic of a fixed acceptance rule is. The non-monotonic logic generated by an acceptance rule $\alpha$ can be characterised through a completeness result: the completeness problem amounts to characterising the consequence relation  $\vDash_{\mathcal{M}}$, where $\mathcal{M}$ is the class of all $\alpha$-models. In other words, it consists in finding a system of inference rules $\mathsf{S}$ such that $\vdash_{\textsf{S}}$ and $\vDash_{\mathcal{M}}$ coincide.

For an example of such a completeness theorem, it is worth reminding a result by \cite{LIN} and  \cite{LIK, GEO} who, in addressing the tracking problem, provide a completeness theorem for their preferred acceptance rule. The rule in question is \emph{Shoham-driven acceptance}, which we will here denote by $\kappa$. Lin and Kelly define the $\kappa$-rule as follows. If $(\Omega, \mathcal{P}(\Omega),\mu)$ is a discrete probability space and $q\in\mathbb{R}$ ($q\geq 1$), we have: 
\begin{align*}
\kappa_{q}&:\Delta_{\mathfrak{A}}\rightarrow\mathfrak{A} ,\text{ defined as}\\
\kappa_{q}(\mu)&:=\Big\{\omega_{i}\in\Omega\,\, \biggr\rvert \,\, \frac{\mu(\omega_{i})}{\max_{\omega\in\Omega}\mu(\omega)}\geq \frac{1}{q}\Big\}   
\end{align*}
Alternatively, Lin and Kelly show we can characterise the acceptance rule and the resulting revision as an order-minimisation operation:
\begin{align*}
          \kappa_{q}(\mu)  &= \min(\prec_{\mu}),  \text{     where $\omega_{i}\prec_{\mu} \omega_{j}$ if and only if $\frac{\mu(\omega_{i})}{\mu(\omega_{j})}> q$} \\
          \kappa_{q}(\mu)^{\ast}X&:=\kappa_{q}(\mu_{X}) = \min(\prec_{\mu (\cdot\,|\,X)}) = \min(\prec_{\mu}\restriction X) \text{    for all $X$ with }\mu(X)>0  
\end{align*}

The general idea is that given an agent's prior, we can generate a partial order $\prec_{\mu}$ on states, whereby $\omega_i$ is more plausible than $\omega_j$ (written $\omega_{i}\prec_{\mu} \omega_{j}$) if the odds-ratio $\mu(\omega_i)/\mu(\omega_j)$ is sufficiently high.\footnote{In more sophisticated versions of the rule, one can allow the threshold for `sufficiently high' odds-ratios to depend on both states being compared, so that that for every pair of states $\omega_i, \omega_j$ there is a separate threshold $q_{i,j}$.} Then (1) the beliefs of an agent with credence function $\mu$ are always given by the disjunction of all states that are $\prec_{\mu}$-minimal and (2) the effect of using $\kappa$-induced belief revision can be described as simply restricting the ordering $\prec_{\mu}$ to those states consistent with the new evidence. Thus the belief revision operation induced by Bayesian conditioning can be described in purely qualitative terms as simply restricting a plausibility order to those basic outcomes not ruled out by the evidence. That is, the resulting logic of conditionals admits \emph{preferential semantics} \citep{KLM}, under which conditional judgments of the form $\varphi\nc\psi$ are true whenever the most preferred (order-minimal) $\varphi$-states are $\psi$-states. Lin and Kelly obtain the following completeness result for $\kappa$-models:

\begin{theorem}[\textbf{System \textsf{P} completeness}, \cite{LIN}, \cite{GEO}]\label{System P completeness}
Let $\Gamma\cup\{\Phi\}$ a finite set of flat conditionals of the form $\varphi\nc\psi$, and $\mathcal{K}$ the class of $\kappa$-models. Then,
\begin{center}
$\Gamma \vdash_{\mathsf{P}} \Phi$ if and only if $\Gamma \vDash_{\mathcal{K}} \Phi$.
\end{center}
\end{theorem}

System \textsf{P} has been a long-time favourite amongst the systems of nonmonotonic logic. A probabilistic semantics for it was already present in Adams' early work deriving from his Ph.D. thesis \citep{ADA}. Lin's result provides a semantics for system \textsf{P} that is significantly less cumbersome than Adams' original `$\delta-\epsilon$' account, and more intuitive: it employs the Ramsey test to directly relate conditional probabilities to conditional beliefs. The result above shows that we can obtain simple probabilistic semantics for system \textsf{P} in a relatively natural way, via a well-chosen acceptance rule: thus one answer to the tracking problem yields an alternative way to arrive at System \textsf{P} . 

Naturally, an analogous question arises in the context of Leitgeb's acceptance principle. What is the logic of probabilistically stable revision? Is it as a known member of the well-studied family of KLM-style nonmonotonic logics \citep{KLM, SCHL}? Can we characterise these revision operators using preferential semantics?

As we shall see, the logic of stability-based acceptance is a rather unusual beast and, accordingly, the tools required for the representation result will also take us beyond the usual toolbox of belief revision theory and nonmonotonic logic. Firstly, as we will see next, probabilistically stable revision validates certain strong monotonicity principles, while failing even mild instances of case-reasoning. This already places the resulting belief revision operators outside the main classical systems of non-monotonic reasoning. Secondly, probabilistically stable revision does not admit preferential semantics (whereby the revision operation is represented as a minimisation operator for an underlying plausibility order). This is in sharp contrast with Lin and Kelly's $\kappa$-rule and Shoham-driven revision. As we cannot associate probabilistically stable revision operators with preferential models based on plausibility orders, preferential models are of no help for our main task of giving a qualitative description of probabilistically stable revision. In what follows, we shall instead solve the problem by appealing to selection function semantics and the theory of comparative probability orders.

\subsection{Rational Monotonicity, the Or rule, and the logic of conditional probability}\label{Some preliminary observations}

Let us identify some salient patterns of inference validated by probabilistically stable revision. Stripping the Ramsey-test definition above of its syntactic clothing, we define the following semantic consequence relation. Suppose we work with a finite sample space $\Omega$ and an algebra $\mathfrak{A}$ over it, which without loss of generality we can assume to be the whole powerset. (To set aside superfluous issues concerning valuations and the definability of sets of states in $\Omega$ via Boolean formulae, we simply treat propositions as subsets of $\Omega$.) Given a probability space of the form $(\Omega, \mathfrak{A},\mu)$, consider the relation $\nc_{\mu}$ defined directly on $\mathfrak{A}$ as:
\begin{center}
$A\nc_{\mu} B$ if and only if $\tau_{t}(\mu_{A})\subseteq B$ or $\mu(A)=0$, 
\end{center}
In words, $A\nc_{\mu} B$ holds whenever the strongest stable proposition conditional on $A$ entails $B$ or, in other words, the agent's belief set contains $B$ after learning $A$. Each relation $\nc_{\mu}$ (or, equivalently, each qualitative revision generated by $\tau$) represents a doxastic state together with `contingency plans'. The current unconditional beliefs $\tau(\mu)$ are given by all $A$ such that $\Omega\nc_{\mu}A$\textemdash that is, all propositions that $\nc_{\mu}$-follow from the tautology. All other entailments of the form $A\nc_{\mu}B$ specify the agent's contingency plan for revision. Note of course that $A\nc_{\mu} B$ simply amounts to the claim that either $\mu(A)=0$  or $\sigma_{\mu,t}(A)\subseteq B$, i.e., given evidence $A$, the belief revision plan generated by $\mu$ and $t$ yields a belief state where $B$ is believed.

We can immediately note some interesting features of the resulting consequence relation. First of all, it directly follows from this setup that (\textsf{Left Equivalence}) and (\textsf{Right Weakening}) hold, as the $\tau$-rule operates directly on an algebra of propositions. Since belief states are always closed under deduction, it follows that  (\textsf{And}) is validated, as well. Next, it follows from the definition of stability that $\tau(\mu_{A})\subseteq A$, and so (\textsf{Ref}) is validated, too. But now there are two particularly interesting aspects to this version of nonmonotonic consequence. Firstly, note that Rational Monotonicity is validated:

\begin{observation}\label{RATMON}
The $\tau$-rule satisfies \emph{\textsf{(RM)}}: that is, for any measure $\mu$, we have that 
\begin{center}
If $A\not\nc_{\mu}B^{c}$ and $A\nc_{\mu}C$, then $A\cap B\nc_{\mu}C$.
\end{center}
\end{observation}
\begin{proof}
We need to show
\begin{center}
If $\tau(\mu_{A})\not\subseteq B^{c}$, and $\tau(\mu_{A})\subseteq C$, then 
$\tau(\mu_{A\cap B})\subseteq C$.
\end{center}
Assume $\tau(\mu_{A})\not\subseteq B^{c}$ and $\tau(\mu_{A})\subseteq C$. This entails $\tau(\mu_{A})\cap B\neq\emptyset$, and moreover $\tau(\mu_{A})\cap B\subseteq C$. 
We prove that $\tau(\mu_{A\cap B})\subseteq \tau(\mu_{A})\cap B$. It is enough to show that $\tau(\mu_{A})\cap B$ is stable with respect to $\mu_{A\cap B}$: the desired inclusion then follows since $\tau(\mu_{A\cap B})$ is the $\subseteq$-least $\mu_{A\cap B}$-stable set. 
So let $Y\in\mathfrak{A}$ such that $(\tau(\mu_{A})\cap B)\cap Y \neq\emptyset$ and $\mu_{A\cap B}(Y)>0$. We need to show that 
$$\mu_{A\cap B}\Big(\tau(\mu_{A})\cap B \,\big|\,Y\Big)>t.$$
Note that $\mu_{A\cap B}(Y)>0$ entails $\mu_{A}(B\cap Y)>0$, and since $\tau(\mu_{A})\cap (B\cap Y)\neq\emptyset$, the $\mu_{A}$-stability of $\tau(\mu_{A})$ entails $\mu_{A}(\tau(\mu_{A})\,|\,B\cap Y)>t$. But then

\begin{align*}
\mu_{A}\Big(\tau(\mu_{A})\,\big|\,B\cap Y\Big) &= \mu_{A}\Big(\tau(\mu_{A})\cap B\,\big|\,B\cap Y\Big)   \\
&= \mu_{A\cap B}\Big(\tau(\mu_{A})\cap B \,\big| \,Y\Big) > t,
\end{align*}
as desired. So the set $\tau(\mu_{A})\cap B$ is $\mu_{A\cap B}$-stable, and therefore $\tau(\mu_{A\cap B})\subseteq \tau(\mu_{A})\cap B \subseteq C$.
\end{proof}

\noindent Note that nothing in the verification of Observation \ref{RATMON} relies on our specific choice of threshold: indeed, (\textsf{RM}) holds for any choice of threshold. 

Secondly, $\tau$-acceptance does \emph{not} satisfy the (\textsf{Or}) rule. 

\begin{observation}\label{failure of Or}
For any discrete set algebra $(\Omega,\mathfrak{A})$ with $|\Omega|\geq 3$ and any threshold $t\in[0.5,1)$, there are measures $\mu$ on it such that $\nc_{\mu}$ fails the \emph{(\textsf{Or})} rule. 
\end{observation}
\begin{proof}

Let $\Omega=\{\omega_{1},\omega_{2}, \omega_{3},\dots,\omega_{n}\}$. Fix a threshold $t\geq 1/2$. Now pick any $a$ such that $\frac{t}{2-t}<a\leq t$, and consider the measure $\mu$ such that $\mu(\omega_1)=a$ and $\mu(\omega_2)=\mu(\omega_3)= (1-a)/2$. Now set $A=\{\omega_{3}\}$, $B:=\{\omega_{1}, \omega_{2}\}$, and  $C:=\{\omega_{1}, \omega_{3}\}$.  We can easily compute $\tau(\mu_{A})= A$, $\tau(\mu_{B})= \{\omega_{1}\}$, and $\tau(\mu_{A\cup B})=A\cup B$. 
So we have $A\nc_{\mu} C$ and $B\nc_{\mu} C$, but $A\cup B\not\nc_{\mu} C$.
\end{proof}

Here is a simple, concrete counterexample. Your stability threshold is 1/2. A ball was drawn uniformly at random from an urn containing exactly 4 Red balls, 3 Green, and 3 Blue balls, and nothing else. (We assume you know the composition of the urn, and that the ball was drawn uniformly \emph{by your lights}, and so on). You judge that 
$$\textsf{Red} \vee \textsf{Blue} \nc_{\mu} \neg \textsf{Blue},$$
i.e., upon learning that the drawn ball is either \textsf{Red} or \textsf{Blue}, you accept $\neg \textsf{Blue}$ (\textsf{Red} is sufficiently more likely than \textsf{Blue} to be above threshold and, it being one of the most fine-grained possibilities that you entertain in this context, its probability cannot decrease upon learning anything that doesn't contradict it). Since \textsf{Red}, \textsf{Green}, and \textsf{Blue} are mutually exclusive, you also judge, by logic alone, that $$\textsf{Green} \nc_{\mu} \neg \textsf{Blue}.$$ But given your evidence right now, you do not already accept $\neg \textsf{Blue}$: 
$$(\textsf{Red} \vee \textsf{Blue}) \vee \textsf{Green} \not\nc_{\mu} \neg \textsf{Blue}$$
This is simply because no hypothesis stronger than $(\textsf{Red} \vee \textsf{Blue}) \vee \textsf{Green}$ is probabilistically stable: none has probability that is resiliently above $1/2$. In particular, \textsf{Red} is not above threshold: and, while every disjunct of two basic outcomes is above threshold, none of them is \emph{stably} above threshold.\footnote{Each of them admits a defeater: $\textsf{Red}\vee \textsf{Green}$ is defeated by $\textsf{Green}\vee\textsf{Blue}$ (and mutatis mutandis for $\textsf{Red}\vee \textsf{Blue}$), while $\textsf{Green}\vee \textsf{Blue}$ is defeated by $\textsf{Red}\vee \textsf{Green}$, for instance.}

In this example, the hypothesis $\neg\textsf{Blue}$ is $\nc_{\mu}$-entailed by both $\textsf{Green}$ and its negation $\neg\textsf{Green}$ (equivalent to $ (\textsf{Red}\vee\textsf{Blue}) $), but is not $\nc_{\mu}$-entailed by the tautological disjunction $\textsf{Green}\vee \neg\textsf{Green}$. In logical terms, this means that the $\tau$-rule fails an even weaker principle, capturing the restriction of (\textsf{Or}) to mutually exclusive and exhaustive propositions:
$$
\infer[(\textsf{exOr})] {\textsf{T}\nc \psi}{\varphi\nc \psi & \neg\varphi\nc\psi }
$$
 Given an algebra of propositions, how common are measures that fail the (\textsf{Or}) rule? Building on counterexamples like the above, one can rather easily show that for any discrete set algebra $(\Omega,\mathfrak{A})$ and any threshold $t$, there is an open neighbourhood of distributions $\mathcal{N}$ in the probability simplex such that, for any $\mu\in\mathcal{N}$, $\nc_{\mu}$ fails the (\textsf{Or}) rule. In this sense, the failures of (\textsf{Or}) for the $\tau$ rule are non-neglibible.\footnote{There is a natural connection between the  (\textsf{Or}) rule and the oft-discussed principle of \emph{conglomerability} (notably discussed\textemdash and rejected\textemdash by de Finetti \citep{DEF}). We say a probability measure $\mu$ is (countably) \emph{conglomerable} if, for any countable partition $\Pi$ of the underlying probability space, we have $$\mu(X)\in \big[\inf_{E\in\Pi}\mu(X\,|\,E), \sup_{E\in\Pi}\mu(X\,|\,E) \big]$$ for any $X\in\mathfrak{A}$. In other words, conglomerability requires that the unconditional probability of an event remains within the bounds fixed by the most extreme values its probability could take when conditioning on cells from the partition. That is, conglomerability requires the following: if learning any event $E\in\Pi$ yields a probability $\mu(\cdot\,|\,E)$ such that $\mu(X\,|\,E)\in [a,b]$ (for some fixed $a,b\in[0,1]$ with $a\leq b$), then we should already have $\mu(X)\in[a,b]$ unconditionally. The analogy with the (\textsf{Or}) rule is tempting (but see \cite{HOW} who defends the view that the analogy with an infinitary Or-introduction rule is `merely illusory' \cite[p. 12]{HOW}, if this Or-rule is considered as a purely deductive inference rule). Non-conglomerability can only occur if the measure $\mu$ fails countable additivity \citep{SEI} and, with a few extra assumptions, Seidenfeld et al. have shown that this phenomenon generalises to higher cardinalities \citep{SEI2} (that is, whenever the measure is not $\kappa$-additive, conglomerability fails for partitions of cardinality at most $\kappa$). At a very high level (and perhaps \emph{only} at a very high level) the failure of the $(\textsf{Or})$ rule can be seen as showing that the stability rule forces a finite, qualitative counterpart of non-conglomerability, even though the underlying measures are always assumed to be conglomerable; for we can have a partition $\Pi=\{A, A^{c}\}$ and an event $X\in\mathfrak{A}$ such that $X$ has a $\mu_{A}$-stable subset \emph{and} a $\mu_{A^{c}}$-stable subset, but \emph{no} $\mu$-stable subset.}

This marks a notable difference between probabilistically stable revision and AGM reasoning: unlike probabilistically stable revision, AGM revision always satisfies the (\textsf{Or}) rule.\footnote{Two key axioms of AGM belief revision are the axioms Inclusion (when revising, do not adopt any beliefs that are not logically entailed by your beliefs together with the new evidence) and Preservation (when revising by a proposition that is not disbelieved, all previously held beliefs are maintained). In the presence of the Success (that one should always believe the proposition being received as new evidence), Inclusion amounts to validating the \textsf{Or} rule. Probabilistically stable revision satisfies Preservation, but it does not satisfy Inclusion: see \citep[p. 562]{MIE} for some examples and a discussion of this point.}

The same example as above also shows that probabilistically stable revision also fails the rule (\textsf{WOr}) from Definition \ref{Non-monotonic logics DEF}. 
$$
\infer[(\textsf{WOr})] {\varphi\nc\gamma}{\varphi\wedge\psi\nc\gamma & \varphi\wedge\neg\psi\nc\gamma}
$$
(\textsf{WOr}) is a weakening of the (\textsf{Or}) rule tailored to logics of conditional probability, investigated in \citep{HAW, HAWMAK, PARSIM}. Given a measure $\mu$ and threshold parameter $t\in[0,1]$, define 
$$
A\nc_{\lambda} B \,\,\Longleftrightarrow \,\,\mu_{A}(B)>t \text{ or }\mu(A)=0 
$$
These are perhaps the simplest probabilistic semantics for a non-monotonic logic one could envisage: $A$ defeasibly entails $B$ whenever the conditional probability of $B$ \emph{given} $A$ is high. Note that, in our setting, we can view these semantics as the non-monotonic logic that results from a \emph{Lockean belief revision policy}: i.e. the belief revision policy that results from adopting a Lockean notion of acceptance, according to which you accept all and only propositions with probability above a fixed threshold parameter (which need not be above 1/2).  $A\nc_{\lambda} B$ means that $B$ is accepted after learning $A$ under the Lockean principle (provided $A$ has non-zero probability). Providing a completeness result for such Lockean revision policies is a notoriously involved problem (see \citep{PARSIM}).\footnote{\cite{HAWMAK} showed that System \textsf{O} does not capture exactly the probabilistically sound consequence relations (in our parlance, those generated by a Lockean belief revision policy): in fact, no set of finite-premised Horn rules can even derive all the Horn rules that are sound under these probabilistic semantics. They conjectured that System \textsf{O} is complete for finite-premised Horn rules, in the sense that it derives every probabilistically sound finite-premised Horn rule. This conjecture was settled in the negative, through subtle and surprisingly involved methods (including an excursion into non-standard analysis), by Paris and Simmons \citep{PARSIM}.} Since \cite{HAW}, the core system involved in this analysis has been System \textsf{O} (see Definition \ref{Non-monotonic logics DEF}), all of whose rules are all validated by these Lockean semantics. The rule (\textsf{WOr}) is particularly salient when comparing probabilistically stable revision to the logic of Lockean belief revision policies: it can easily be seen that probabilistically stable revision satisfies \emph{all the rules of System \emph{$\textsf{O}$} except for \emph{(\textsf{WOr})}}.

Some may see the failure of the (\textsf{Or}) as a damning feature of probabilistically stable revision. We will leave it to the reader to judge whether, and to what degree, one should be troubled by the failure of the \textsf{Or} rule (the reader is invited to revisit the counterexamples and ascertain how \emph{irrational} they deem these inferences). Here are three remarks to help contextualise that question.

First, abandoning the  (\textsf{Or}) rule may well be an inevitable cost of preserving Rational Monotonicity (\textsf{RM}). The proof of Lin \& Kelly's impossibility theorems \citep{LIK, GEO} highlights a general tension between the (\textsf{Or}) rule and (\textsf{RM}). Their impossibility result for tracking shows that, given some plausible assumptions on one's underlying acceptance rule, one's belief revision policy cannot validate both (\textsf{Or}) and (\textsf{RM}). Our observations here show that, unlike Lin \& Kelly's preferred Shoham-driven revision policy, probabilistically stable revision falls on the (\textsf{RM}) side of the dilemma. Second, as we will later show (\S\ref{MINIMOP}), it turns out that probabilistically stable belief revision does obey a certain cognate of the (\textsf{Or}) rule, albeit a much weaker one. 

Lastly, let us conclude with a word on interpreting the  (\textsf{Or}) rule. The usual arguments for the (\textsf{Or}) rule present it as a commonsense desideratum for handling \emph{case reasoning}. The rationale for case reasoning is best illustrated in the context of arguments for the (weaker) rule (\textsf{exOr}). Suppose we have $\varphi\nc\psi$ and $\neg\varphi\nc\psi$, while also having $\varphi\vee\neg\varphi\not\nc\psi$. This means that the agent believes $\psi$ conditionally on learning either $\varphi$ or $\neg\varphi$, but does not currently accept $\psi$. This seems to run counter to the following intuitive principle, expressed by Lin and Kelly:
\begin{quote}
if you know that you will accept a proposition regardless what you learn, you should accept it already. \cite[p. 960]{LIK} 
\end{quote} 
The rule can also be seen as analogous to Savage's decision-theoretic \emph{sure-thing principle} \cite{SAV}: if, in a decision problem, the agent knows she would select action $a$ conditional on event $X$ being true, \emph{and} she would select action $a$ conditional on its negation $X^{c}$, then she ought to select action $a$ outright. 

A natural way to excuse the failure of the (\textsf{Or}) rule consists in marking the distinction between \emph{hypothetical} conditionals and \emph{actual update} conditionals. On the first reading, the expression $\varphi\nc\psi$ captures a contingency plan that is transparent to the agent herself: under this reading, the agent knowingly commits to accepting $\psi$ in both cases $\varphi$ and $\neg\varphi$. Refusing to accept $\psi$ appears incongruous, since she also knows that one of $\varphi$ and $\neg\varphi$ already obtains (even though she may not know which). 

On the second reading, in the conditional $\varphi\nc\psi$, the antecedent $\varphi$ ought to be read as a \emph{truly dynamic} operator: in a spirit closer to dynamics logics \citep{PLA, BMS}, we could conceive of learning $\varphi$ as an \emph{event} taking place. Following the dynamic tradition, we can suggestively write this informational event as $[!\varphi]$ to distinguish it from $\varphi$, the mere proposition \emph{that $\varphi$ holds}. Then the formula $\varphi\nc\psi$ \textemdash properly read as $[!\varphi]\psi$\textemdash expresses that the informational state of the agent is such that, after \emph{actually, truthfully learning} $\varphi$, she accepts $\psi$. It may be the case that the information events $[!\varphi]$ and $[!\neg\varphi]$ both lead her to accept $\psi$; but, at the current stage, neither event actually happened, so nothing prompted the agent to adopt that belief. On this reading, there is no rationality failure on the part of the agent: for one thing, she need not believe that either event $[!\varphi]$ or $[!\neg\varphi]$ will occur (even though she knows that either $\varphi$ or $\neg\varphi$ is the case). For another, the outcomes of learning events need not be transparent to her.

We will remain neutral on this interpretative issue. In particular, we will not pursue in detail this second reading here: it is most fruitful to first understand the behaviour of basic stability-induced conditionals, regardless of matters of interpretation, before we can resort to more expressive logics that capture the distinction between propositions $\varphi$ and their corresponding learning events $[!\varphi]$. For our purposes, the failure of even weak versions of the (\textsf{Or}) rule is significant primarily because it renders the logic of probabilistic stability an unusual and interesting object of study: it is an indication that the logic of probabilistic stability differs in important respects from the ``plain'' logic of conditional probability embodied in the Lockean belief revision semantics). 

Let us take stock: the characteristics of $\tau$-generated consequence relations that we highlighted here are rather unusual, and they make it difficult to place, be it very approximately, the resulting logic of probabilistic stability on the family tree of known nonmonotonic logics. To begin with, most known nonmonotonic systems in the literature, like McCarthy's circumscription logics \citep{MAC}, non-monotonic inference via the \emph{closed world assumption on Horn clauses} \cite{MAKREI}, systems \textsf{P} and \textsf{R} \citep{KLM}, as well as most logics of conditionals, extend system \textsf{C}. But a small modification of the counterexample to (\textsf{Or}) used in the proof of Observation \ref{failure of Or} shows that the $\tau$-generated consequence relations do not in general satisfy (\textsf{Cut}); hence, we are not dealing with a $\textsf{C}$-complying class of consequence relations (and of course, our consequence relations of the form $\nc_{\mu}$ do not extend system \textsf{P}). Further, $\tau$-consequence is always closed under $\textsf{(RM)}$; but the $\textsf{(RM)}$-complying systems presented in the literature usually also satisfy the $(\textsf{Or})$ rule, which $\tau$-consequence violates. This means that we are dealing with a notion of consequence which is, in one sense, unusually weak, in that it does not validate very common (and, arguably, rather intuitive) rules like $(\textsf{Or})$ or $(\textsf{Cut)}$; nor does it validate even weaker principles of disjunctive reasoning $(\textsf{WOr})$ obeyed by the non-monotonic logic of conditional probability. In another sense, however, this notion of consequence is rather strong, for it validates $\textsf{(RM)}$, itself considered a strong requirement on nonmonotonic consequence\textemdash and one which very easily fails in probabilistic contexts, including Lockean conditional probability semantics, as well as under the acceptance policies put forward by Lin and Kelly. The logic of probabilistically stable revision thus bears some hallmarks of paradigmatically `qualitative' systems (typically equipped with plausibility order semantics), as well as markers of typically probabilistic behaviour, like the failure of the $(\textsf{Or})$ rule. All this indicates that probabilistically stable revision occupies an unusual position in the landscape, and we need to do some more work to isolate the correct axioms for stability-based reasoning. 

\subsection{Selection structures and the representation problem}\label{QualiMod} 

A prominent model for qualitative belief revision operators is given by \emph{order-based} revisions: these are revisions that are defined by order-minimisation and ``qualitative'' conditioning. Start a plausibility order $\leq$ defined on $\Omega$, where $w \leq v$ is interpreted as the statement that $w$ is \emph{at least as plausible} than $v$ (or: $w$ is \emph{at most as implausible as} $v$). The agent's strongest accepted proposition is given by the set $\min(\leq)$: intuitively, it is the disjunction of all the most plausible worlds (states) that the agent entertains. Revision by a proposition $E$ amounts to restricting $\leq$ to $E$ and taking the collection of all $\leq$-minimal elements of the restricted order. Thus the induced belief revision plan is defined as the map $E\mapsto \min(\leq \restriction \hspace{-0.2em} E)$: given new evidence $E$, the agent's strongest accepted proposition is the disjunction of all $\leq$-minimal states consistent with the evidence. 

Many qualitative revision operators admit characterisations in terms of order-based revisions\textemdash not least among which are AGM operators (involving total preorders) as well as Lin and Kelly's preferred Shoham-driven revision (involving partial orders). It is natural to ask if one could give such an order-based characterisation for the class $\tau$-generated revision operators: that is, whether each map taking $E$ to $\tau(\mu_{E})$ can be described as an operation of the form $E\mapsto\min(\leq\restriction  \hspace{-0.2em} E)$, for an appropriate choice of a binary relation $\leq$. Can probabilistically stable belief revision be given order-based semantics? Unfortunately, it is easy to see that this cannot be done:  

\begin{ex}\label{ExampleofNon-rationalizability}
Let $t=1/2$, $\Omega =\{\omega_{1}, \omega_{2}, \omega_{3}\}$ with $\mu$ given by $\mu(\omega_{1}) = 0.4$, $\mu(\omega_{2}) = 0.35$  and $\mu(\omega_{3}) = 0.25$. Let $E:=\{\omega_{1},\omega_{2}\}$. Note that $\tau(\mu)=E$: the hypothesis $E$ is already accepted by the agent. Now suppose the agent now learns $E$ with certainty: we have the updated probabilities given by $\mu_{E}(\omega_{1}) = 0.5\bar{3}$, $\mu_{E}(\omega_{2}) = 0.4\bar{6},$  and $\mu_{E}(\omega_{3}) = 0$. The new belief set is given by $\{\omega_{1}\}\neq E$. But any relation $\leq$ on $\Omega$ for which $\min(\leq)= E$ will also satisfy $\min(\leq\restriction  \hspace{-0.2em} E) = E$, and so the order-based revision by $E$ will not change the belief set. 
\end{ex}

What this example illustrates is that $\tau$-generated revision cannot be \emph{tracked} using order-based revision: that is, there is no way to translate each probability measure $\mu$ into an order $\leq$ on $\Omega$ so that each operation $\min(\leq)\mapsto\min(\leq\restriction  \hspace{-0.2em} E)$ coincides with the revision $\tau(\mu)\mapsto\tau(\mu_{E})$. 

A well-known generalisation of order-based semantics, stemming from the logic of conditionals \citep{STA}, consists in employing \emph{selection structures} (or \emph{selection function semantics}). In our semantic setting, a selection structure is a structure $(\Omega, \mathfrak{A}, \sigma)$ where $(\Omega, \mathfrak{A})$ is a set-algebra of propositions, equipped with a map $\sigma:\mathfrak{A}\rightarrow\mathfrak{A}$ called a \emph{selection} function. To each selection structure $\mathfrak{M}=(\Omega, \mathfrak{A}, \sigma)$ corresponds a consequence relation $\nc_{\sigma}$ on propositions, defined as:
\begin{center}
$E\nc_{\sigma}H$ if and only if $\sigma(E)\subseteq H$.
\end{center}
It is common to deal with selection functions that satisfy $\sigma(E)\subseteq E$ for any $E\in\mathfrak{A}$. With this in mind, the idea of the above semantics is simply that $E$ defeasibly entails $H$ whenever all the \emph{selected} $E$-worlds are $H$-worlds (rather than requiring that \emph{all} $E$-worlds be $H$-worlds, as classical entailment does). This is a strict generalisation of the order-minimisation semantics, in that we do not assume here that the selected worlds are those that are minimal under some underlying order. 

Of course, if we simply interpret the selected subset $\sigma(E)$ as the strongest accepted proposition upon learning $E$, a selection function is nothing but a belief revision plan; and the consequence generated from $\sigma$ is nothing but a consequence relation obtained via a Ramsey test for conditionals. 

Now, identifying the class of probabilistically stable revision plans simply amounts to identifying those selection structures that are probabilistically representable, in the following sense:

\begin{mydef}[\textbf{Representable selection structures}]
Let $(\Omega, \mathfrak{A}, \sigma)$ a selection structure. Given a threshold $t\in[0.5,1)$, we say that a measure $\mu$ on $\mathfrak{A}$ \emph{represents} $\sigma$ for threshold $t$ if and only if
\begin{align*}
\forall E\in\mathfrak{A}, \,\,\, &\sigma(E)=\tau_{t}(\mu_{E})\text{ if }\mu(E)\neq0,\\
\text{and }  &\sigma(E)=\emptyset \hspace{2.1em}\text{ if }\mu(E)=0.
\end{align*}
A selection structure is \emph{$t$-representable} if there exists a measure $\mu$ on $\mathfrak{A}$ that represents $\sigma$ for threshold $t$. 
\end{mydef}

 Recall our earlier notation: $\sigma_{\mu,t}$ denotes the revision plan induced from the measure $\mu$ by the stability rule with threshold $t$. Thus a structure $(\Omega, \mathfrak{A}, \sigma)$ is $t$-representable exactly if $\sigma=\sigma_{\mu,t}$ for some probability measure $\mu$. The representable selection structures are those for which $\sigma(E)$ represents the strongest probabilistically stable proposition after conditioning on $E$: we can say that $\sigma$ is a \emph{strongest-stable-set operator}. 

Our main task is to identify exactly the probabilistically stable revision plans. We may formulate this straightforwardly as a representation problem. The representation problem for probabilistically stable revision consists in imposing the right axioms on a selection function $\sigma$ so that it is representable as a probabilistically stable revision plan. In other words: we want to find some informative necessary and sufficient conditions that guarantee the representability of a selection function as a strongest-stable-set operator. 

Of course, in order for this to count as a informative characterisation, we want a reasonably `qualitative' axiomatisation of the resulting revision operators that does not explicitly appeal to probabilistic notions. This is a desideratum echoed by Lin and Kelly: 

\begin{quote}
It is easy to achieve perfect tracking: just define [\emph{the revised belief set to be $\alpha(\mu_{E})$, where $\alpha$ is a chosen acceptance rule}]. To avoid triviality, one must specify what would count as a propositional approach to belief revision that does not essentially peek at probabilities to decide what to do. \cite[p. 963]{LIK}
\end{quote}

In the remainder of this section, we will build our way towards a representation result for probabilistically stable revision. We begin by providing a geometric characterisation of probabilistically stable revision, which establishes that each probabilistically stable revision plan (equivalently, each $\tau$-generated consequence relation) can be uniquely identified by a specific system of linear inequalities. This allows to identify certain important properties that the selection function $\sigma$ must satisfy in order to be probabilistically representable. Next, we focus on a representation for the special case of Leitgeb's rule with strict threshold $t=1/2$. This case is a natural choice from the logical side, since it can be seen as the most `qualitative' version of stability-based acceptance (as advocated by \cite{LEI3}). We will see that the representation problem for the case $t=1/2$ bears a close connection with the theory of comparative probability orders. We will then exploit this connection to prove our main results in Section 4. There, we will first prove a representation theorem for $t=1/2$, after which we provide a general representation result for all rational threshold values: this constitutes our most general characterisation of probabilistically stable belief revision operators.

\subsection{The geometry of probabilistically stable revision.} \label{The geometry of probabilistically stable revision}

In what follows we work with probability spaces of the form $(\Omega, \mathcal{P}(\Omega),\mu)$, with $\Omega=\{\omega_{1},\dots,\omega_{n}\}$ a finite sample space. We fix a stability threshold $t\in[0.5,1)$. We provide a geometric characterisation of the representation problem, which helps isolate relevant structural properties of selection functions.    

Consider the simplex $\Delta^{n-1}$ of all probability distributions over $\Omega$, conceived of as the set of vectors $\{(x_1,\dots,x_n)\,|\, x_i\geq 0, \sum^{n}_{i=1} x_{i}=1\}$, each measure $\mu$ corresponding to the vector $x_i=\mu(\omega_{i})$. Say that two measures are equivalent if they generate the same revision plan:

\begin{mydef}[\textbf{Revision-equivalent measures}]
    Two probability measures $\mu$, $\rho\in\Delta^{n-1}$ are revision-equivalent (for threshold $t$), written $\mu\sim \rho$, if and only if $\sigma_{\mu,t}= \sigma_{\rho,t}$: equivalently, if and only if $\nc_{\mu}=\nc_{\rho}.$
\end{mydef}

Note that identifying the probabilistically stable revision plan generated by a probability measure is tantamount to identifying the consequence relation generated by that measure.\footnote{This follows from the general observation that, in our setting of finite algebras of propositions, we can evidently identify a selection function $\sigma$ with its consequence relation $\nc_{\sigma}$. First, two different selection functions manifestly give to different consequence relations: if $\sigma_{1}(A)\neq \sigma_{2}(A)$\textemdash without loss of generality, assume $\sigma_1(A)\nsubseteq\sigma_{2}(A)$ \textemdash then $A \not\nc_{\sigma_{1}} \sigma_{2}(A)$, while manifestly $A\nc_{\sigma_{2}}\sigma_{2}(A).$ In the other direction, it is enough to observe that, given a consequence relation $\nc$ generated by a selection function $\sigma$, we always have $\sigma(X)= \bigcap \{Y\in\mathfrak{A}\,|\, X\nc Y\}$. Since $\mathfrak{A}$ is a finite algebra, this event is in $\mathfrak{A}$. We have $A\nc \sigma(A)$ and so $\bigcap \{Y\in\mathfrak{A}\,|\, A\nc Y\}\subseteq \sigma(A)$. In the other direction, by definition $A\nc Y $ entails $\sigma(A)\subseteq Y$ so $\sigma(A)\subseteq \bigcap \{Y\in\mathfrak{A}\,|\, A\nc Y\}$. Thus two selection functions that give rise to the same consequence relation must be equal.} So switching between talk of revision operators $\sigma_{\mu}$ and talk of stability-generated consequence relations $\nc_{\mu}$ is an innocuous, and at times helpful, shift in perspective. 

Now, what information do we need to identify the revision plan generated by a probability measure? Since we want each selection function $\sigma$ to represent one particular revision plan (equivalently, consequence relation), this means that we want axioms for selection functions such that, for any $\sigma$ satisfying these axioms, the set of distributions in $\Delta^{n-1}$ representing $\sigma$ is an equivalence class of $\sim$. What do those equivalence classes look like, and how do we find the selection functions that pick out exactly those classes?
We first need a criterion for identifying when two distributions generate the same revision plan (i.e., when two distributions generate the same consequence relations). We begin with the following observation: 

\begin{prop}\label{prop:stab1}
Let $(\Omega, \mathcal{P}(\Omega),\mu)$ a finite probability space, and $t\in[0.5,1)$ a threshold. Then for any $A$ such that $\mu(A)>0$, we have that 
$\tau(\mu_{A})=B$ if and only if the following hold:
\begin{itemize}
\item[(i)] $\forall \omega\in B$, $\mu(\omega)> \frac{t}{1-t}\cdot\mu(A\setminus B)$
\item[(ii)] $\forall X\subset B$, $\exists \omega\in X$, $\mu(\omega)\leq \frac{t}{1-t}\cdot\mu(A\setminus X)$.
\end{itemize}
\end{prop}
\begin{proof}
We sketch it only, as it easily follows from the definition of stability (a proof is implicit in Leitgeb's discussion of an algorithm for generating stable sets in \cite[p.1363]{LEI}). The key fact is that a proposition $X$ is $(\mu,t)$-stable exactly when there are no \emph{defeater states}\textemdash that is, $\omega\in X$ such that $\mu(\omega\,|\,X^{c}\cup\{\omega\})\leq t$. Then we use the fact that $\mu(\omega)\leq \frac{t}{1-t}\cdot\mu(A\setminus X)$ is equivalent to $\mu_{A}(\omega\,|\,X^{c}\cup\{\omega\})\leq t$. Condition (i) says that no $\omega\in B$ is a defeater with respect to $A$: so $B$ is indeed $\mu_{A}$-stable. Condition (ii) says that each strict subset of $B$ has at least one defeater with respect to $A$: so no strictly smaller subset is $\mu_{A}$-stable. The two conditions together mean that $B$ is indeed the $\subseteq$-least $\mu_{A}$-stable set.
\end{proof}

Note that either side of the biconditional entails $B\subseteq A$, so we can rewrite Proposition \ref{prop:stab1} in the following equivalent form: 
\begin{center}
$\tau(\mu_{A})=B$ if and only if $B\subseteq A$ and 
\begin{enumerate}[align=left]
\item[(i)] \hspace{0.3em}$\forall \omega\in B$, $\mu(\omega)> \frac{t}{1-t}\cdot\mu(A\setminus B)$;
\item[(ii)] $\forall X\subset B$, $\exists \omega\in X$, $\mu(\omega)\leq \frac{t}{1-t}\cdot\mu(A\setminus X)$.
\end{enumerate}
\end{center}

This very elementary observation points to an important fact about the behaviour of the $\tau$-rule. In essence, Proposition \ref{prop:stab1} reduces the problem of determining each selected subset $\sigma_{\mu,t}(A)=\tau(\mu_{A})$ to that of comparing inequalities of the form $\mu(\omega)>\frac{t}{1-t}\cdot\mu(X)$ and $\mu(\omega)\leq\frac{t}{1-t}\cdot\mu(X)$ for certain sets $X$. This allows to visualise each $\sim$-equivalence class as a certain convex polytope in $\Delta^{n-1}$. Define the following:

\begin{mydef}[\textbf{Fixed-odds hyperplanes}]
Let $\mathcal{H}_{n}$ be the collection of hyperplanes in $\mathbb{R}^{n}$ defined by equations of the form 
$$
x_{i}=\frac{t}{1-t}\Big(\sum_{x_{j}\in X} x_{j}\Big),  \text{  where $X$ is a set of variables such that $x_{i}\not\in X$}.
$$
We call $\mathcal{H}_{n}$ the collection of \emph{fixed-odds hyperplanes}. 
\end{mydef}

We are interested in the way that $\mathcal{H}_{n}$ partitions the probability simplex $\Delta^{n-1}$. For each $\mu\in\Delta^{n-1}$, it is enough to look at equalities of the form $\mu(\omega_{i})=\frac{t}{1-t}\cdot\mu(X)$ for sets $X\subseteq\Omega$ not containing $\omega_{i}$. We have:

\begin{observation}\label{GEOMETRY}
Let $\mu$, $\rho\in\Delta^{n-1}$. We have $\mu\sim \rho$ if and only if $\mu$ and $\rho$ lie in the same region of the hyperplane arrangement $\mathcal{H}_{n}$ and have the same support.\footnote{That is, $\mu$ and $\rho$ agree on which propositions have measure 0.}
\end{observation}

\begin{proof}
Suppose $\mu$ and $\rho$ do not lie in the same region of $\mathcal{H}_{n}$. Then there is some hyperplane with equation $x_{i}=\frac{t}{1-t}\Big(\sum_{x_{j}\in X} x_{j}\Big)$, $x_{i}\not\in X$, that separates them. In terms of probabilities, this means that there is some 
$\omega_{i}$ and set $X$ with $\omega_{i}\not\in X$ such that (wlog):
$$
\mu(\omega_{i})>\frac{t}{1-t}\mu(X)\,\,\,\text{ and }\,\,\,\rho(\omega_{i})\leq\frac{t}{1-t}\rho(X)
$$
This entails that $\mu(\omega_{i}\,|\,X\cup\{\omega_{i}\})> t$, which in turn entails that $X\cup\{\omega_{i}\}\nc_{\mu}\{\omega_{i}\}$, while by the same reasoning $X\cup\{\omega_{i}\}\not\nc_{\rho}\{\omega_{i}\}$, which means $\nc_{\mu}\neq\nc_{\rho}$.
Conversely, suppose $\mu$ and $\rho$ lie in the same region of $\mathcal{H}_{n}$. This means that they lie on the same side of every hyperplane in $\mathcal{H}_{n}$. We then have that for every state $\omega_{i}$ and set $X$ such that $\omega_{i}\not\in X$, we have 
\begin{equation}
\mu(\omega_{i})>\frac{t}{1-t}\mu(X)\,\,\,\text{ iff }\,\,\,\rho(\omega_{i})>\frac{t}{1-t}\rho(X)
\end{equation}
This just means that $\mu(\omega_{i}\,|\,X\cup\{\omega_{i}\})> t$ holds if and only if $\rho(\omega_{i}\,|\,X\cup\{\omega_{i}\})> t$. Now let propositions $A$, $B\subseteq\Omega$ such that $A\nc_{\mu}B$. We assume $\mu(A)\neq0$ (otherwise, we immediately have $A\nc_{\rho}B$\textemdash this is because  $\mu$ and $\rho$ have the same support, and therefore if $\mu(A)=0$ entails $\rho(A)=0$). Thus $A\nc_{\mu}B$ means $\tau(\mu_{A})\subseteq B$. Suppose towards a contradiction that $\tau(\rho_{A})\not\subseteq B$. Then, by Proposition \ref{prop:stab1}, either property (i) or (ii) must fail for $\rho$. If (i) fails, then 
$$\exists\omega\in B,\,\,\rho(\omega)\leq\frac{t}{1-t}\cdot\rho(A\setminus B).$$
But, by (1), this means that $\mu(\omega)\leq\frac{t}{1-t}\cdot\mu(A\setminus B)$, and so (i) does not hold for $\mu$, contradicting the fact that $\tau(\mu_{A})\subseteq B$. Similarly, if (ii) fails, then we have
$$
\exists X\subset B,\,\,\forall \omega\in X,\,\, \rho(\omega)> \frac{t}{1-t}\cdot\rho(A\setminus X).
$$
But then (1) entails that, for all $\forall \omega\in X,\,\, \mu(\omega)> \frac{t}{1-t}\cdot\mu(A\setminus X)$, which by Proposition \ref{prop:stab1} again contradicts the fact that $\tau(\mu_{A})\subseteq B$. So both (i) and (ii) hold for $\rho$, and therefore we have $\tau(\rho_{A})\subseteq B$, hence $A\nc_{\rho}B$ as required. This concludes the proof.
\end{proof}

Thus we have characterised the way in which the probabilistically stable revision plans divide the probability simplex: each is determined by hyperplane equations of the form given above.\footnote{Regions of hyperplane arrangements (chambers) do not contain points on the hyperplanes themselves (they are open regions). What of distributions that lie on the hyperplanes? We simply add all distributions $\mu$ satisfying $\mu(\omega_{i})=\frac{t}{1-t}\mu(X)$ to the adjacent regions satisfying  $\mu(\omega_{i})<\frac{t}{1-t}\mu(X)$. Then Observation \ref{GEOMETRY} still holds for this extended notion of `regions' (via the same argument). This takes care of all distributions: we have fully partitioned the simplex by $\sim$-equivalence classes.} That is: 
\begin{center}
\begin{itemize}
\item For each sample space of size $n$, we get a hyperplane arrangement $\mathcal{H}_{n}$, given by $n\times (2^{n-1}-1)$ equations of the form $\mu(\omega)=\frac{t}{1-t}\cdot\mu(X)$ for pairs $(\omega,X)$ such that $\omega\not\in X$.\\
\item Each revision plan (equivalently, each relation $\nc_{\mu}$) can be identified uniquely by checking in which region of $\mathcal{H}_{n}$ the point $\mu$ is.
\end{itemize}
\end{center}
The significance of this characterisation is that it gives us a strategy for finding the right axioms for selection functions. Figure \ref{Fixed-odds lines for Leitgeb's tau-rule} represents an example of a hyperplane arrangement whose cells correspond to a particular probabilistically stable revision plan/consequence relation. 

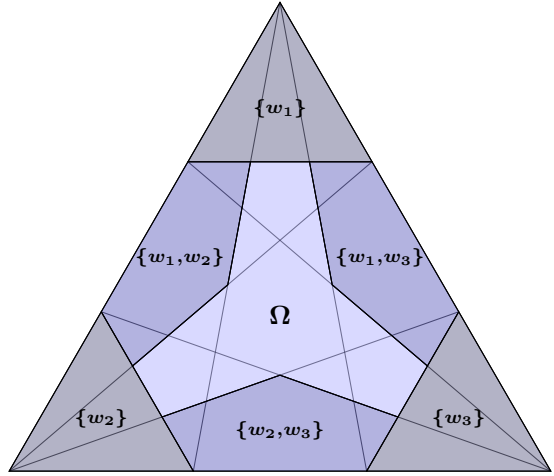
\begin{figure}[h!]
\centering
\begin{tikzpicture}[scale=0.6]
\coordinate (A) at (0:0);
\coordinate (C) at (-60:12cm);
\coordinate (B) at (240:12cm);
\draw (A) -- (C) -- (B) -- cycle;
\path (A) edge[draw opacity=0] (B) [right, pos=0.20] (B) (A|- B);
\draw [name path=lAC, opacity=0.6] (A) -- ($(B)!.66!(C)$);
\draw [name path=lAB, opacity=0.6] (A) -- ($(C)!.66!(B)$);
\draw [name path=lBC, opacity=0.6] (B) -- ($(A)!.66!(C)$);
\draw [name path=lBA, opacity=0.6] (B) -- ($(C)!.66!(A)$);
\draw [name path=lCB, opacity=0.6] (C) -- ($(A)!.66!(B)$);
\draw [name path=lCA, opacity=0.6] (C) -- ($(B)!.66!(A)$);
\draw [name path=tC, opacity=0.6]($(B)!.66!(C)$) -- ($(A)!.66!(C)$);
\draw [name path=tB, opacity=0.6]($(A)!.66!(B)$) -- ($(C)!.66!(B)$);
\draw [name path=tA, opacity=0.6]($(B)!.66!(A)$) -- ($(C)!.66!(A)$);

\path [name intersections={of=lBC and lCB,by=mBC}];
\path [name intersections={of=lAB and lBA,by=mAB}];
\path [name intersections={of=lAC and lCA,by=mAC}];
\path [name intersections={of=lBA and tB,by=tBA}];
\path [name intersections={of=lBC and tB,by=tBC}];
\path [name intersections={of=lCA and tC,by=tCA}];
\path [name intersections={of=lCB and tC,by=tCB}];
\path [name intersections={of=lAB and tA,by=tAB}];
\path [name intersections={of=lAC and tA,by=tAC}];
\coordinate (BC1) at ($(C)!.66!(B)$);
\coordinate (BC2) at ($(B)!.66!(C)$);
\coordinate (AC1) at ($(C)!.66!(A)$);
\coordinate (AC2) at ($(A)!.66!(C)$);
\coordinate (AB1) at ($(B)!.66!(A)$);
\coordinate (AB2) at ($(A)!.66!(B)$);
\coordinate (pT) at (barycentric cs:A=1 ,B=1,C=1);
\coordinate (pAB) at (barycentric cs:AB1=1,tAB=1,mAB=1,tBA=1,AB2=1);
\coordinate (pAC) at (barycentric cs:AC1=1,tAC=1,mAC=1,tCA=1,AC2=1);
\coordinate (pBC) at (barycentric cs:BC1=1,tBC=1,mBC=1,tCB=1,BC2=1);
\coordinate (pA) at (barycentric cs:A=1,tAC=1,tAB=1);
\coordinate (pB) at (barycentric cs:B=1,tBA=1,tBC=1);
\coordinate (pC) at (barycentric cs:C=1,tCB=1,tCA=1);

\draw [fill=blue!50!white, fill opacity=0.3, name path=zT] (tBA) -- (mAB) -- (tAB) -- (tAC) -- (mAC) -- (tCA) -- (tCB) -- (mBC) -- (tBC) -- cycle;

\draw [fill=blue!60!black, fill opacity=0.3, name path=zAvB] ($(B)!.66!(A)$) -- (tAB) -- (mAB) -- (tBA)-- ($(A)!.66!(B)$) -- cycle;
\draw [fill=blue!60!black, fill opacity=0.3, name path=zBvC] ($(C)!.66!(B)$) -- (tBC) -- (mBC) -- (tCB)-- ($(B)!.66!(C)$) -- cycle;
\draw [fill=blue!60!black, fill opacity=0.3, name path=zAvC] ($(C)!.66!(A)$) -- (tAC) -- (mAC) -- (tCA)-- ($(A)!.66!(C)$) -- cycle;

\draw [fill=blue!25!black, fill opacity=0.3, name path=zA] (A) -- ($(C)!.66!(A)$) -- ($(B)!.66!(A)$)-- cycle;
\draw [fill=blue!25!black, fill opacity=0.3, name path=zB] (B) -- ($(C)!.66!(B)$) -- ($(A)!.66!(B)$)-- cycle;
\draw [fill=blue!25!black, fill opacity=0.3, name path=zC] (C) -- ($(B)!.66!(C)$) -- ($(A)!.66!(C)$)-- cycle;

\node at (pT) {$\bm{\Omega}$};
\node at (pAB){$\scriptstyle\bm{\{w_{1},w_{2}\}}$};
\node at (pBC) {$\scriptstyle\bm{\{w_{2},w_{3}\}}$};
\node at (pAC) {$\scriptstyle\bm{\{w_{1},w_{3}\}}$};
\node at (pA) {$\scriptstyle\bm{\{w_{1}\}}$};
\node at (pB) {$\scriptstyle\bm{\{w_{2}\}}$};
\node at (pC) {$\scriptstyle\bm{\{w_{3}\}}$};

\end{tikzpicture}
\caption[Fixed-odds lines for Leitgeb's $\tau$-rule]{Fixed-odds lines for Leitgeb's $\tau$-rule with $t=2/3$ and $|\Omega|=3$. Each cell of the hyperplane arrangement determines a specific consequence relation $\nc$. The colored regions represent \emph{acceptance zones}: they correspond to sets of probability distributions which agree on the unconditional strongest stable set.}\label{Fixed-odds lines for Leitgeb's tau-rule}
\end{figure}

\subsection{Representing selections through systems of linear inequalities}\label{SUBSECSEL} 
Knowing what the revision plan $\sigma_{\mu,t}$ is amounts to knowing, for each pair $(\omega,X)$ with $\omega\not\in X$, which of $\mu (\omega)> \frac{t}{1-t}\cdot\mu(X)$ or $\mu(\omega)\leq\frac{t}{1-t}\cdot\mu(X)$ holds. But this simply corresponds to checking, for each such pair, whether ${\mu(\omega\,|\,X\cup\{\omega\})>t}$ or not. In turn, it is straightforward to see that $\mu(\omega\,|\,X\cup\{\omega\})>t$ if and only if $\tau(\mu_{X\cup\{\omega\}})=\{\omega\}$. That is, $\{\omega\}$ is the strongest stable proposition with respect to the measure $\mu_{X\cup\{\omega\}}$: this fact can simply be written as $\sigma_{\mu, t} (X\cup \{\omega\}) = \{\omega\}$. This immediately suggests the following approach.
\begin{itemize}
\item \textbf{Step I}: First, we want to impose axioms on selection functions which will ensure that the selection function mirrors the behaviour of the $\tau$-rule shown in Proposition \ref{prop:stab1}. That is, our axioms must ensure that $\sigma(A)=B$ holds if and only if $B\subseteq A$ and
\begin{enumerate}[align=left]
\item[(i$^{\prime}$)] $\forall b\in B$, $\sigma((A\setminus B)\cup\{b\})=\{b\}$;
\item[(ii$^{\prime}$)] $\forall X\subset B$, $\exists x\in X$, $\sigma((A\setminus X)\cup\{x\})\neq\{x\}$.
\end{enumerate}

\item \textbf{Step II}: Given $\sigma$, we want to be able to translate all statements of the form $\sigma((A\setminus B)\cup\{b\})=\{b\}$ into a set of linear inequalities that admits a solution in $\Delta^{n-1}$. We construct a system of linear inequalities $L_{\sigma}$ as follows: for each pair $(\omega, X)$ with $\omega\not\in X$,
\begin{itemize}
\item whenever $\sigma(X\cup\{\omega\})=\{\omega\}$, add the constraint $\mu(\omega)> \frac{t}{1-t}\cdot\mu(X)$,
\item otherwise, add the constraint $\mu(\omega)\leq \frac{t}{1-t}\cdot\mu(X)$ .
\end{itemize}
Our axioms need to ensure that the resulting system of linear inequalities $L_{\sigma}$, together with the constraint that $\mu$ lie in the simplex $\Delta^{n-1}$, admits a solution.
\end{itemize}

If these desiderata are satisfied, then we have fully characterised the class of $\tau$-generated revision maps, and therefore the class of consequence relations $\nc_{\mu}$. We can argue as follows: suppose a measure $\mu$ satisfies all linear inequalities in the resulting system $L_{\sigma}$. Then, in particular, whenever $\sigma(A)=B$, the properties (i$^{\prime}$) and (ii$^{\prime}$) hold, and since $\mu$ satisfies the resulting inequalities, the corresponding properties (i) and (ii) from Proposition \ref{prop:stab1} hold as well on the probabilistic side; hence, $\tau(\mu_{A})=B$. Conversely, if $\tau(\mu_{A})=B$, then $B\subseteq A$, and both (i) and (ii) hold. Via our translation, this means that properties (i$^{\prime}$) and (ii$^{\prime}$) hold as well, which yields $\sigma(A)=B$. Thus, if we can show that $L_{\sigma}$ always admits a solution for any $\sigma$ satisfying our axioms, we have guaranteed the existence of some probability measure $\mu$ such that $\sigma_{\mu,t}=\sigma$, and we are done.

Naturally, in order to capture all and only probabilistically stable operators, our axioms also must be sound with respect to our probabilistic interpretation: that is, they must hold of every probabilistically stable revision plan (for a fixed threshold). Here are some axioms for selection functions which are sound in this sense: 
\begin{enumerate}[align=left] \label{SLIST}
\item [\textsf{(S1)}] \hspace{0.5em}$\sigma(X)=\emptyset$ only if $X=\emptyset$;
\item [\textsf{(S2)}] \hspace{0.5em}$\sigma(X)\subseteq X$;
\item [\textsf{(S3)}] \hspace{0.18em} If $\sigma(A)\cap B\neq\emptyset$, then $\sigma(A\cap B)\subseteq \sigma(A)\cap B$;
\item [\textsf{(S4$_{n}$)}] For any $n$: if $\sigma(A\cup X_{i})=X_{i}$ for all $i\leq n$, then $\sigma(A\cup\bigcup_{i\leq n}X_{i})\subseteq \bigcup_{i\leq n}X_{i}$;
\item [\textsf{(S5)}] \hspace{0.4em}For any $A, B, C$ which are pairwise disjoint, if $\sigma(A\cup B)\subseteq A$ and $\sigma(B\cup C)\subseteq B$, then\vspace{-0.4em}
\item[]  \hspace{0.5em}$\sigma(A\cup C)\subseteq A$.
\end{enumerate}

Axiom \textsf{(S1)} may look surprising here, since at first sight nothing excludes the possibility of `negligible' propositions, i.e. nonempty propositions $X$ such that $\sigma(X)=\emptyset$. By imposing this axiom, in effect we restrict attention to all selection functions that are probabilistically representable by \emph{regular measures}: i.e. measures for which only the empty event has probability zero. This restriction is not substantial, but a mere technical convenience: for any representable selection function, we can exclude all propositions with $\sigma(X)=\emptyset$ from the domain of $\sigma$, restrict attention to what $\sigma$ does on this restricted domain, represent it via a \emph{regular} probability measure $\mu$ on this smaller algebra, and then extent the measure assigning all these `negligible' propositions probability zero.\footnote{On the axiomatic side, one simple way of extending our results to the case of non-regular measures is to (1) add an additional axiom stating the family of negligible subsets $\{X\in \mathcal{P}(\Omega)\,|\,\sigma(X)=\emptyset\}$ forms an ideal: i.e., it is closed under subsets and unions, and (2) restrict all the axioms from our representation result to non-neglibible propositions. The restriction to the regular case allows us to avoid the additional bookkeeping without losing any generality.} So, without loss of generality, we can focus on representing probabilistically stable revision operators for regular measures $\mu$. 

We omit verifications for soundness as they are straightforward. Note that the axioms \textsf{(S2)} and \textsf{(S3)} suffice for $\sigma$ to validate (\textsf{Ref}) and (\textsf{RM}), respectively. The soundness of \textsf{(S2)} is immediate, and that of \textsf{(S3)} essentially amounts to our earlier verification of (\textsf{RM}). Now, the axioms \textsf{(S1)-(S4)} suffice for selection functions to mirror Proposition \ref{prop:stab1} on the qualitative side:

\begin{prop}\label{prop:stab2}
If a selection function $\sigma:\mathfrak{A}\rightarrow\mathfrak{A}$ satisfies \emph{\textsf{(S1)-(S4)}}, then $\sigma(A)=B$ if and only if $B\subseteq A$ and 
\begin{itemize}
\item[(i$^{\prime}$)] $\forall b\in B$, $\sigma((A\setminus B)\cup\{b\})=\{b\}$;
\item[(ii$^{\prime}$)] $\forall X\subset B$, $\exists x\in X$, $\sigma((A\setminus X)\cup\{x\})\neq\{x\}$.
\end{itemize}
\end{prop}
\begin{proof}
Suppose $\sigma(A)=B$. Then $B\subseteq A$ because of \textsf{(S2)}. If $A=\emptyset$, then by \textsf{(S1)} so is $B$, and then (i$^{\prime}$)
and (ii$^{\prime}$) hold vacuously as $B$ contains neither elements nor strict subsets. 
So we can assume $A$, $B \neq\emptyset$. We show (i$^{\prime}$).
Take $b\in B$.
Now we have $\sigma(A)\cap\big((A\setminus B) \cup \{b\}\big)\neq\emptyset$, since it is equal to $B \cap\big((A\setminus B) \cup \{b\}\big) = \{b\}$. So by \textsf{(S3)}, we must have 
\begin{align*}
\sigma\Big(A\cap \big((A\setminus B)\cup\{b\}\big)\Big)&\subseteq \sigma(A)\cap((A\setminus B) \cup \{b\}) \\
\sigma\Big((A\setminus B)\cup\{b\}\Big) &\subseteq B\,\,\,\cap((A\setminus B) \cup \{b\}),\\
\text{ so  }\hspace{2em} \sigma\Big((A\setminus B)\cup\{b\}\Big) & \subseteq\{b\}.
\end{align*}
Now by \textsf{(S1)}, $\sigma\big((A\setminus B)\cup\{b\}\big)\neq\emptyset$, so it is equal to $\{b\}$, as desired. 

To show that (ii$^{\prime}$) holds, proceed by contradiction. Suppose that $\exists X\subset B$, $\forall x\in X$, $\sigma((A\setminus X)\cup \{x\})=x$. $X$ is a finite set, so write $X=\{x_{1},\dots, x_{k}\}$. We then have $\forall i\leq k$, $\sigma((A\setminus X)\cup \{x_{i}\})=\{x_{i}\}$. So, by \textsf{(S4$_{k}$)}, we have that 
$$\sigma((A\setminus X)\cup \bigcup_{i\leq k}\{x_{i}\})\subseteq \bigcup_{i\leq k}\{x_{i}\}.$$ 
The left-hand side then is simply $\sigma((A\setminus X)\cup X)=\sigma(A)$(by \textsf{(S2)} we have $B\subseteq A$, hence $X\subset A$), while the right-hand side is $X$. So we have $\sigma(A)\subseteq X$. But we know $\sigma(A)=B$ and $X\subset B$, which contradicts $\sigma(A)\subseteq X$. So (ii$^{\prime}$) holds after all.

For the other direction, suppose (i$^{\prime}$) and (ii$^{\prime}$) hold. We show this entails $\sigma(A)=B$. First, list elements $b_{1},...,b_{k}$ of $B$ and use \textsf{(S4$_{k}$)} as before to conclude $\sigma(A\cup B)\subseteq B$, and since $B\subseteq A$ we have $\sigma(A)\subseteq B$.
Now assume $B\not\subseteq\sigma(A)$.  Since we already know $\sigma(A)\subseteq B$, this means $\sigma(A)\subset B$. By (ii$^{\prime}$), we have that 
$$
\exists a\in\sigma(A),\,\,\,\,\sigma\Big((A\setminus \sigma(A))\cup\{a\}\Big)\neq\{a\}
$$
But this is impossible, since\textemdash  as we have just shown above using \textsf{(S3)}\textemdash in general $\sigma(A)=S$ entails $\forall s\in S$, $\sigma((A\setminus S)\cup\{s\})=\{s\}$. So $B\subseteq\sigma(A)$ after all, and since we know $\sigma(A)\subseteq B$ we can conclude  $\sigma(A)=B$.
\end{proof}

Thus axioms \textsf{(S1)-(S4)} suffice to complete Step I of our plan above: we have sound axioms for selection functions which ensure the key properties (i$^{\prime}$) and (ii$^{\prime}$). The only remaining task is that of carrying out Step II: that is, imposing axioms on $\sigma$ that will ensure that the solvability of the resulting system of linear inequalities $L_{\sigma}$ obtained from $\sigma$. The question of ensuring the consistency of $L_{\sigma}$ is more involved. It can be seen as a special and somewhat more intricate case of a representation problem for comparative probability orders. To build our way towards a solution, we first consider the particular case for a threshold $t=1/2$, before addressing representability for different thresholds.

\subsection{Connection with comparative probability orders}\label{Connection with comparative probability orders}

The next step involves a direct connection to the theory of comparative probability. From now on, we deal with the case of the $\tau$-rule with threshold $t=1/2$. To make the comparison with comparative probability salient, we employ here a more suggestive notation. Define the relation $\succ_{\sigma}$ between states $\omega\in\Omega$ and sets $X\subseteq\Omega$ as follows: for each pair $(\omega,X)$ with $\omega\not\in X$, let 
\begin{center}
$\omega\succ_{\sigma} X$ if and only if $\sigma(X\cup\{\omega\})=\omega$.
\end{center}
(When the selection function $\sigma$ is clear form the context, we will omit the subscript). 
We are given a system of inequalities as follows: for each pair $(\omega,X)$, we have that either $\omega\succ_{\sigma} X$ or $\neg(\omega\succ_{\sigma} X)$. Each expression $\omega\succ_{\sigma} X$ translates into the constraint $\mu(\omega)>\mu(X)$, while each expression $\neg(\omega\succ_{\sigma} X)$ translates into $\mu(\omega)\leq\mu(X)$. The question is: what axioms do we need to impose on $\succ_{\sigma}$ for it to be probabilistically representable, in the sense of guaranteeing the existence of a probability measure meeting these constraints?

The problem we have to solve is one of representation of a partial comparative probability order. In the theory of comparative probability, one usually starts with a full ordering $\preceq$ on an algebra of events $\mathcal{P}(\Omega)$, and the task consists in finding a probability measure which represents $\preceq$. We say that a measure $\mu$ represents $\preceq$ if and only if, for any $A, B\subseteq\Omega$,
$$
A\preceq B\Leftrightarrow \mu(A)\leq\mu(B).
$$ 
Under what circumstances is such an order representable? This is a question going back to de Finetti \cite{DEF1}. One of the most important classical results in the early theory of comparative probability is a general answer to this question. 

\begin{theorem}[\cite{KPS, SCO}]\label{Scott's Theorem}
Let $(\Omega, \mathcal{P}(\Omega))$ a finite set algebra, and $\preceq$ a reflexive total order on $\mathcal{P}(\Omega)$. There is a probability measure $\mu$ on $(\Omega, \mathcal{P}(\Omega))$ representing $\preceq$ if and only if the following hold for all $A$, $B$, $C\subseteq\Omega$:
\begin{itemize}
\item[\emph{\textsf{(Q1)}}] $\Omega\not\preceq\emptyset$;
\item[\emph{\textsf{(Q2)}}] $\emptyset \preceq A$;
\item[\emph{\textsf{(Q3)}}] If $(A\cup B)\cap C = \emptyset$, then $(A\preceq B \Leftrightarrow A\cup C \preceq B \cup C)$;
\item [\emph{\textsf{(Q4)}}] If $(A_{i})_{i\leq n}$ and $(B_{i})_{i\leq n}$ are balanced sequences and $\forall i< n$, $A_{i}\preceq B_{i}$, then $A_{n}\succeq B_{n}$.
\end{itemize}
\end{theorem}

The last requirement \textsf{(Q4)}, sometimes called the \emph{Finite Cancellation Axiom}, plays a crucial role, and it is worth taking a moment here to explain its meaning. The notion of two sequences of events being \emph{balanced} is to be understood as follows: given two such sequences $(A_{i})_{i\leq n}$ and $(B_{i})_{i\leq n}$, we write $(A_{i})_{i\leq n} \geq_{0} (B_{i})_{i\leq n}$ whenever 
$$
\sum_{i\leq n}\ind{{A_{i}}}\geq \sum_{i\leq n}\ind{{B_{i}}}.
$$
This vector inequality simply means that for each $\omega\in\Omega$, $\omega$ is in at least as many $A_{i}$'s as $B_{i}$'s: in other words, we have $\big|\{i\leq n\,|\, \omega\in A_{i}\}\big|\geq \big|\{i\leq n\,|\, \omega\in B_{i}\}\big|$. When $(A_{i})_{i\leq n} \geq_{0} (B_{i})_{i\leq n}$ and $(A_{i})_{i\leq  n} \leq_{0} (B_{i})_{i\leq n}$, we write 
$$
(A_{i})_{i\leq n} \equiv_{0} (B_{i})_{i\leq n}
$$
and we say that $(A_{i})_{i\leq n}$ and $(B_{i})_{i\leq n}$ are \emph{balanced} sequences. Thus:

\begin{mydef}[\textbf{Balanced sequences}]
Let $(\Omega, \mathfrak{A})$ a finite set algebra, and $(A_{i})_{i\leq n}$ and $(B_{i})_{i\leq s}$ two sequences of events from $\mathfrak{A}$. The sequences are \emph{balanced} if and only if 
$$\sum_{i\leq n} \mathds{1}_{A_{i}} = \sum_{i\leq n} \mathds{1}_{B_{i}}.$$
We then write $(A_{i})_{i\leq n} \equiv_{0} (B_{i})_{i\leq n}$.
\end{mydef}

We can of course express this property directly in terms of \emph{occurrence counts}. Given $\omega\in\Omega$ and a sequence of sets $S = (S_{i})_{i\leq n}$, let $\eta_{S}(\omega):=\big| \{ i\leq n \,|\, \omega\in S_{i}\} \big| $ count the number of occurrences of $\omega$ in $S$. Then $(A_{i})_{i\leq n}$ and $(B_{i})_{i\leq n}$ are balanced whenever $\eta_{A}(\omega)=\eta_{B}(\omega)$ for every $\omega\in\Omega$. This means that for each state $\omega$, the number of occurrences of $\omega$ in the $A_{i}$ sets is the same as the number of occurrences of $\omega$ in the $B_{i}$ sets: and similarly, to say that  $(A_{i})_{i\leq  n} \leq_{0} (B_{i})_{i\leq n}$ is simply to say that for every $\omega\in \Omega$, we have $\eta_{A}(\omega)\leq \eta_{B}(\omega)$.

Scott's well-known proof of Theorem \ref{Scott's Theorem} \citep{SCO} appeals to a hyperplane-separation theorem. Scott identifies each set $X$ with its characteristic vector $\mathds{1}_{X}$ and shows, via a crucial application of Scott's axiom \textsf{(Q4)}, that the sets $\{\mathds{1}_{A} - \mathds{1}_{B} \,|\,A\succ B\}$ and $\{\mathds{1}_{A} - \mathds{1}_{B} \,|\,A \preceq B\text{ and }B\preceq A\}$ can be separated by a hyperplane with equation $\overrightarrow{v}\cdot \mathbf{x}=0$. Then, letting $$\mu(X):=\frac{\overrightarrow{v}\cdot\mathds{1}_{X}}{\overrightarrow{v}\cdot\mathds{1}_{\Omega}}$$ gives the desired probability representation. 

Our case, however, is a little more intricate. We have an ordering $\succ_{\sigma}$ that is of a very specific kind: it is non-total, and its domain is restricted: it relates singletons (more generally, atoms in the event algebra) on one side to sets on the other. Furthermore, we need to deal with strict and non-strict constraints simultaneously\textemdash both of which are given as primitives\textemdash rather than deriving one set of constraints from the other (as it is usually done for total probability orders). 

We can begin by observing some plausible candidates for representability conditions. The following properties of the relation $\succ_{\sigma}$ are sound with respect to our desired probabilistic interpretation. 

\begin{itemize}\label{SLISTb}
\item[(M1)] If $\omega\succ_{\sigma} X$ and $X\supseteq Y$, then $\omega\succ_{\sigma} Y$.
\item[(M2)] If $\omega\succ_{\sigma} X$ and $v\in X$, $v\succ_{\sigma} Y$ then $\omega\succ_{\sigma} (X\setminus\{v\})\cup Y$.
\item[(M3)] If $\omega_{1}\not\succ_{\sigma} B_{1}$ and $\omega_{2}\not\succ_{\sigma} B_{2}$ for $B_{1}\cap B_{2}=\emptyset$, then $\forall \omega\,(\omega\succ_{\sigma} B_{1}\cup B_{2}\rightarrow\omega\succ_{\sigma} \{\omega_{1},\omega_{2}\})$.
\item[(M4)]  If $\exists X, \exists \omega$ such that $\omega\not\succ_{\sigma} X\cup\{v_{1}\}$ and $\omega\succ_{\sigma} X\cup\{v_{2}\}$, then $v_{1}\succ_{\sigma} \{v_{2}\}$. 
\item[(Sc)] If $(A_{i})_{i\leq m}$ and $(B_{i})_{i\leq m}$ are balanced sequences and $(\omega_{i})_{i\leq m}$ a sequence of states, then (${\forall i\leq m,\, \omega_{i}\succ_{\sigma} A_{i})} \rightarrow (\exists i\leq m,\,\omega_{i}\succ_{\sigma} B_{i})$. 
\end{itemize}

In what follows, we apply the hyperplane separation techniques form the theory of comparative probability orders, and rely a method analogous to the proof of \cite{SCO}. Given that we are dealing with a mixed-constraints case, the relevant hyperplane-separation result to be employed here is the Motzkin Transposition Theorem\footnote{See, for instance, \cite[p. 33]{SCH} or \cite{SW70}.} \citep{MOT}. Here, one should also expect a Scott-like cancellation axiom to do most of the work: as it turns out, this is indeed the case, but what is needed is an axiom stronger than the property (Sc) listed above. In what follows, we employ this strategy to prove a probabilistic representation theorem for selection structures. 

\section{A representation theorem}\label{secrep}

In this section, we solve the representation problem for selection structures: we give necessary and sufficient conditions for a selection function to be representable as a probabilistically stable revision operator.

We begin by proving a result on comparative probability orders: it gives necessary and sufficient conditions for two partial orders on a finite algebra (one strict and one non-strict) to be jointly weakly representable by a probability measure (Theorem  \ref{strongrep}). We then discuss the relation between our representation result and previous results in the comparative probability literature on other types of probabilistic representability. In particular, our result answers the (hitherto open) problem of providing necessary and sufficient conditions for the \emph{strong representability} of comparative probability orders, generalising a result by \cite{KON} on sufficient conditions for strong representation. Our characterisation also subsumes as a special case the characterisation of \emph{partial representability} due to \cite{FISH} and that of \emph{almost representability} by \cite{KPS}. We then use our result to prove our representation theorem for probabilistically stable revision with threshold $t=1/2$ (Theorem \ref{REPSEL}).  

How sensitive is the structure of probabilistically stable revision to the choice of a stability threshold? To answer this, we next investigate the relations between the classes probabilistically stable revision operators generated by different thresholds. We show, via a simple construction, that each value of the threshold generates a different class of representable selection functions (Proposition \ref{IncreasingThresholds}). Finally, we provide a general representation result for all rational threshold values, which yields our most general characterisation of probabilistically stable belief revision operators (Theorem \ref{GENREPSEL}).

\subsection{Joint representation of comparative probability orders}

Suppose we are given two binary relations $<$ and $\leqslant$ on an algebra of events and we want to know whether we can find a probability measure that jointly represents them, in the sense that $<$ entails \emph{strictly smaller probability} and $\leqslant$ entails \emph{smaller or equal probability}. In other words, we want to know which pairs of relation satisfy the following condition:

\begin{mydef}[Joint weak representability]
Let $\Omega$ a finite set, and let $<$ and $\leqslant$ each a binary relation on 
$\mathcal{P}(\Omega)$. We say that a pair of relations $(<, \leqslant)$ is jointly weakly representable whenever there is a probability measure $\mu$ on $\mathcal{P}(\Omega)$ such that, for all $A$, $B\subseteq\Omega$,
\end{mydef}

\begin{itemize}
\item $A<B \Rightarrow \mu(A)<\mu(B)$ ($\mu$ \emph{partially represents} $<$), and 
\item $A\leqslant B \Rightarrow \mu(A)\leq \mu(B)$ ($\mu$ \emph{almost represents} $\leqslant$)
\end{itemize}

Theorem \ref{strongrep} below gives necessary and sufficient conditions for the joint weak representability of two relations on a finite algebra (one strict, one non-strict). Our criterion for jointly weakly representable (pairs of) relations will play a key role in the representation result for probabilistically stable revision. But besides this particular application, this is a rather natural notion of representability of intrinsic measurement-theoretic interest. For this reason, 
we will first give a general solution to the problem of characterising orders that admit joint weak representability, before applying this result to probabilistically stable revision. 

Joint weak representability is related to the notion of \emph{strong agreement} discussed in the measurement-theoretic literature \citep{KON, REG}. The only difference, but it is a substantial one, is that strong agreement assumes that the $\leqslant$ relation \emph{defines} the $<$ relation, by the following equivalence: 
\begin{equation}\label{orderequivalence}\tag{EQ}
    A<B \Leftrightarrow (A\leqslant B\text{ and }B\not\leqslant A) 
\end{equation}
That is, we say that $\mu$ strongly agrees with $\leq$ whenever $\mu$ is a joint weak representation of $(<,\leqslant)$ and the two relations satisfy \eqref{orderequivalence}. By contrast, the notion of joint weak representation is quite a bit weaker, and allows for a good deal of independence between the two relations. In fact, in the case that interests us \textemdash probabilistically stable revision\textemdash we will apply the notion of joint weak representability to a pair of relations where \eqref{orderequivalence} fails: this is another reason why a representation result is needed for that specific notion of representability.

Now, it is clear that our necessary and sufficient conditions for joint weak representability, coupled with the property \eqref{orderequivalence} as an additional axiom, will directly yield necessary and sufficient conditions for strong representability. This gives an answer to the problem, highlighted by \cite{KON}, of finding necessary and sufficient conditions for strong agreement without relying on auxiliary conditions such as `non-triviality' ($\emptyset < \Omega$) and `non-negativity' ($\emptyset \leqslant X$ for all $X\subseteq \Omega$).\footnote{\cite{KON} presents this as an open problem. In a sense, this characterisation is perhaps a little generous: as this author realised upon completion of the present work, a necessary and sufficient condition similar to ours can be extracted, with only a little bit of work, from \citep{KPS}, although Kraft et al. neither explicitly state the relevant condition nor prove a general representation result for joint weak representability.}

Let us now turn to our first representation theorem. We shall appeal to the rational version of a classic result by Motzkin\footnote{Here, for vectors $\vec{w}=(w_1,...,w_n)$ and $\vec{v}=(v_1,...,v_n)\in\mathbb{R}^n$, the notation $\vec{w}\geq \vec{v}$ means that $w_i \geq v_{i}$ for all $i\leq n$, and similarly for $\vec{w}> \vec{v}$. Given a vector $\alpha\in \mathbb{R}^{k}$ and $k\times n$ matrix $\mathbf{M}$, the product $\alpha^{T}\mathbf{M}$ is the vector in $\mathbb{R}^{n}$ whose $j$-th entry is $\sum_{i\leq k} \alpha_{i} m_{ij}$, with $m_{ij}$ the $(i,j)$-th entry of $\mathbf{M}$.}:

\begin{theorem}[\textbf{Motzkin Transposition Theorem}, \cite{MOT}]
Let $\mathbf{M}_{1}$ be a matrix in $\mathbb{Q}^{k\times n}$, and $\mathbf{M}_{2}$ a matrix in $\mathbb{Q}^{p\times n}$. Then, either there is a vector $\vec{\mu}\in \mathbb{R}^{n}$ such that 
\begin{center}
$\mathbf{M}_{1}\cdot\vec{\mu} \geq \vec{0}$ \\
$\mathbf{M}_{2}\cdot\vec{\mu} > \vec{0}$
\end{center}
or else there exist vectors $\alpha\in\mathbb{Q}^{k}$, $\beta\in\mathbb{Q}^{p}$ such that 
\begin{itemize} 
\item[(a)] $\alpha^{T}\mathbf{M}_{1} + \beta^{T}\mathbf{M}_{2}=\vec{0}$;
\item[(b)] $\alpha\geq\vec{0}$, $\beta\geq\vec{0}$ and $\beta_{i}>0$ for some coordinate $i\leq p$.
\end{itemize}
\end{theorem}

We now prove our first representation result:

\begin{theorem}[\textbf{Representation theorem for joint weak representability}]\label{strongrep}
Let $\Omega$ a finite set, and let $<$ and $\leqslant$ be two binary relations on $\mathcal{P}(\Omega)$ with $<$ nonempty. Then the pair  $(<, \leqslant)$ is jointly weakly representable if and only if the following condition holds:
\begin{itemize}
\item [\textsf{(GS)}] If $(A_{i})_{i\leq n}\leq_{0}(B_{i})_{i\leq n}$ and, for all $i<n$, either $A_{i}\geqslant B_{i}$ or $A_{i}> B_{i}$, then $A_{n}\not>B_{n}$.
\end{itemize}
\end{theorem}

\begin{proof}
Necessity is straightforward: we only prove sufficiency. Let $\Omega=\{\omega_{1},\dots,\omega_{n}\}$. Consider the following two sets of vectors representing inequalities: 
\begin{center}
$\Gamma :=\{\ind{A}-\ind{B}\,|\,A \geqslant B\}$ \\
$\Sigma :=\{\ind{A} - \ind{B}\,|\, A> B \}$,
\end{center}
and we write $|\Gamma|=k$ and $|\Sigma|=p$. Take the following two matrices: $\mathbf{M}_{\Gamma}$ has as rows (transposes of) vectors in $\Gamma$, while the rows of matrix $\mathbf{M}_{\Sigma}$ are (transposes of) vectors in $\Sigma$.
We write them as
$$
\mathbf{M}_{\Gamma}=
\begin{pmatrix}
(\ind{A_{1}}- \ind{B_{1}})^{T}\\
\vdots  \\
(\ind{A_{k}}- \ind{B_{k}} )^{T} 
\end{pmatrix}
\,\,\,\,\text{ and }\,\,\,\,
\mathbf{M}_{\Sigma}=
\begin{pmatrix}
(\ind{A_{k+1}}- \ind{B_{k+1}})^{T}\\
\vdots  \\
(\ind{A_{k+p}}- \ind{B_{k+p}} )^{T} 
\end{pmatrix}
$$
Let $\mathbf{I}_{n}$ the $(n\times n)$ identity matrix. We now prove that there is a vector $\vec{\mu}\in \mathbb{R}^{n}$ such that 
\begin{equation}\label{repsol}
\begin{split}
\mathbf{I}_{n}\cdot\vec{\mu} \geq \vec{0} \\
\mathbf{M}_{\Gamma}\cdot\vec{\mu} \geq \vec{0} \\
\mathbf{M}_{\Sigma}\cdot\vec{\mu} > \vec{0}
\end{split}
\end{equation}
Assume towards a contradiction that there is no such $\vec{\mu}$. The matrices $\mathbf{I}_{n}, \mathbf{M}_{\Gamma}$ and $\mathbf{M}_{\Sigma}$ are rational valued, since they only contain entries in $\{-1,0,1\}$. 
Then, by Motzkin's Transposition Theorem,\footnote{Apply Motzkin's theorem to the system 
\begin{equation*}
\begin{split}
\mathbf{M} \cdot\vec{\mu} \geq \vec{0} \\
\mathbf{M}_{\Sigma}\cdot\vec{\mu} > \vec{0}
\end{split}
\end{equation*}
where $\mathbf{M}$ is the concatenated matrix $\mathbf{M} =\begin{bmatrix}
\mathbf{I}_{n}\\
\mathbf{M}_{\Gamma}
\end{bmatrix}$.} there exists vectors $\gamma$, $\alpha$, and $\beta$ with non-negative rational entries such that
\begin{equation}\label{transp}
\gamma^{T}\mathbf{I}_{n}+ \alpha^{T}\mathbf{M}_{\Gamma} + \beta^{T}\mathbf{M}_{\Sigma}=\vec{0}
\end{equation}
and $\beta$ has at least one positive coordinate. Now, we can in fact assume those are vectors in $\mathbb{N}$, by multiplying the entries by a common denominator. We can then rewrite the above in full as 
$$
\big(\gamma_{1},\dots,\gamma_{n}\big)
\begin{pmatrix}
  1 & 0 & \dots & 0 \\
  \vdots & \vdots & \ddots & \vdots \\
  0 & 0 & \dots & 1 \\
\end{pmatrix}
+
\big(\alpha_{1},\dots,\alpha_{k}\big)
\begin{pmatrix}
(\ind{A_{1}}- \ind{B_{1}})^{T}\\
\vdots  \\
(\ind{A_{k}}- \ind{B_{k}} )^{T} 
\end{pmatrix}
+
\big(\beta_{1},\dots,\beta_{p}\big)
\begin{pmatrix}
(\ind{A_{k+1}}- \ind{B_{k+1}})^{T}\\
\vdots  \\
(\ind{A_{k+p}}- \ind{B_{k+p}} )^{T} 
\end{pmatrix}
= \vec{0}
$$

A piece of notation: given an integer $m$, let us write $mA$ to denote the sequence consisting of the set $A$ repeated $m$ times.  Given this, the sequence of sets $$(\gamma_{1}\{\omega_{1}\}\dots, \gamma_{n}\{\omega_{n}\},\alpha_{1}A_{1},\dots, \alpha_{k}A_{k}, \beta_{1}A_{k+1},\dots \beta_{p}A_{k+p})$$ contains exactly $\gamma_{1}$ copies of the set $\{\omega_{1}\}$ etc., followed by $\alpha_{1}$ copies of the set $A_{1}$, followed by $\alpha_{2}$ copies of the set $A_{2}$, etc. (respectively, $\beta_{j}$ copies of $A_{k+j}$). Now we claim that
\begin{align*}\label{bal}\tag{$\ast$}
\big(\gamma_{1}\{\omega_{1}\},\dots, \gamma_{n}\{\omega_{n}\}, \alpha_{1}A_{1},&\dots, \alpha_{k}A_{k}, \beta_{1}A_{k+1},\dots \beta_{p}A_{k+p}\big)\\
&\equiv_{0} \\
\big(\gamma_{1} \emptyset,\dots,\gamma_{n} \emptyset, \alpha_{1}B_{1},&\dots, \alpha_{k}B_{k}, \beta_{1}B_{k+1},\dots \beta_{p}B_{k+p}\big).
\end{align*}
The two sequences in (\ref{bal}) both have length $l=\sum_{i} \gamma_{i}+\sum_{i} \alpha_{i} + \sum_{j} \beta_{j}$. So let us write the two sequences in (\ref{bal}) as $(A^{\ast}_{i})_{i\leq l}$ and $(B^{\ast}_{i})_{i\leq l}$. 

The fact that those two sequences are balanced follows from the equality (\ref{transp}).
Note that 
\begin{align*}
\gamma^{T}\mathbf{I}_{n}+\alpha^{T}\mathbf{M}_{\Gamma}\, +\, \beta^{T}\mathbf{M}_{\Sigma} &= \sum_{i\leq n} \gamma_{i}\ind{\{\omega_{i}\}}+ \sum_{i\leq k}\alpha_{i} (\ind{A_{i}}-\ind{B_{i}})^{T} + \sum^{p}_{j= k+1}\beta_{j} (\ind{A_{j}}-\ind{B_{j}})^{T} \\
&= \big( \sum_{i\leq n} \gamma_{i}\ind{\{\omega_{i}\}}\, +\,  \sum_{i\leq k} \alpha_{i}\indT{A_{i}} + \sum^{p}_{j= k+1} \beta_{j}\indT{A_{j}} \big) - \big( \sum_{i\leq k} \alpha_{j} \indT{B_{j}}+  \sum^{p}_{j= k+1}  \beta_{j} \indT{B_{j}} \big) = \vec{0}
\end{align*}
In this last equality, the vector $\big( \sum_{i\leq n} \gamma_{i}\ind{\{\omega_{i}\}}+\sum_{i\leq k} \alpha_{i}\indT{A_{i}} + \sum^{p}_{j= k+1} \beta_{j}\indT{A_{j}} \big)$ is a $(1\times n)$ vector $\mathbf{v}=(v_{1},\dots,v_{n})$  where $v_{j}= | \{i\leq l \,|\,\omega_{j}\in A^{\ast}_{i}  \}|$: that is, the $j$-th entry of $\mathbf{v}$ counts the number of $A^{\ast}_{i}$'s in which the state $\omega_{j}$ occurs. By this reasoning, the last line above states that, for any $\omega\in\Omega$, we have $| \{i\leq l\,|\,\omega\in A^{\ast}_{i}  \}| = | \{i\leq l \,|\,\omega\in B^{\ast}_{i}  \}|$. So the sequences in (\ref{bal}) are indeed balanced. From this, it immediately follows that
\begin{equation}\label{dom}\tag{D}
(\alpha_{1}A_{1},\dots, \alpha_{k}A_{k}, \beta_{1}A_{k+1},\dots \beta_{p}A_{k+p}) \leq_{0} (\alpha_{1}B_{1},\dots, \alpha_{k}B_{k}, \beta_{1}B_{k+1},\dots \beta_{p}B_{k+p}) 
\end{equation}
where these two sequences are balanced if and only if all $\gamma_{n}=0$. Now we show that these two sequences constitute a failure of the \textsf{(GS)} axiom. By design of the matrices, we also have that for any $A_{i}$ in the sequence, either $A_{i}\geqslant B_{i}$ or $A_{i} > B_{i}$ obtains. So we have:
\begin{center}
For any $A_{i}$ occuring in the left-hand side sequence in  (\ref{dom}), the corresponding $B_{i}$ on the right-hand side is such that either $ A_{i}\geqslant B_{i}$ or $ A_{i}> B_{i}$.
\end{center}

Next, observe that $\mathbf{M}_{\Sigma}$ is non-empty, since $>$ is a nonempty relation. So at least one strict inequality must occur on the right-hand side in (\ref{dom}) since we know that $\beta$ admits at least one strictly positive coordinate (by $(b)$ of Motzkin's Transposition Theorem). This means that there is a pair $A$, $B$ such $A> B$, the set $A$ occurs somewhere on the left-hand side of (\ref{dom}), and $B$ occurs at the same coordinate on the right-hand-side. Without loss of generality, label all events in the list (\ref{dom}) in such a way that $A^{\ast}_{k+p}:=A$ and $B^{\ast}_{k+p}:=B$. Since for all $i\leq k+p$ on that list, it holds that either $ A_{i}\geqslant B_{i}$ or $ A_{i}> B_{i}$, we can now appeal to \textsf{(GS)} to conclude that we also must have $A^{\ast}_{k+p}\not > B^{\ast}_{k+p}$, i.e. $A\not > B$, which directly contradicts the choice of $A,B$. Thus, the assumption that such vectors $\alpha$ and $\beta$ exist leads to contradiction. By Motzkin's theorem, we can conclude that there is some vector $\vec{\mu}\in \mathbb{R}^{n}$ simultaneously solving $\mathbf{I}_{n}\cdot\vec{\mu} \geq \vec{0} $, $\mathbf{M}_{\Gamma}\cdot\vec{\mu} \geq \vec{0}$ and $\mathbf{M}_{\Sigma}\cdot\vec{\mu} > \vec{0}$. Since $\mathbf{I}_{n}\cdot\vec{\mu}\geq \vec{0}$, we know that the vector $\vec{\mu}$ satisfies
$$
\vec{\mu}\cdot\mathbf{e}_{i} = \vec{\mu}\cdot(\ind{\omega}-\ind{\emptyset})\geq 0
$$
for any $\omega\in\Omega$; but this simply means $\vec{\mu}_{i}\geq 0$ for all $i$. So all coordinates of $\vec{\mu}$ are non-negative. Moreover, the vector $\vec{\mu}$ is not identically zero. Since $<$ is nonempty, there exists at least one strict comparison $A>B$. Note that \textsf{(GS)} then forces $A\neq\emptyset$: no strict comparison can be of the form $\emptyset > B$. To see this, observe the general fact that $A\subseteq B$ entails $A\not >B$: simply note that $(A)\leq_{0} (B)$ and so, according to \textsf{(GS)}, $A\not>B$. This entails that
$
\vec{\mu}\cdot(\ind{A}- \ind{B}) = \sum_{\omega_{i}\in A\setminus B} \mu_i - \sum_{\omega_{j}\in B\setminus A}\mu_{j} > 0 $
where each sum is non-negative, and so there is some $\omega_{i}\in A\setminus B$ with $\mu_i>0$.
The vector $\vec{\mu}$ thus a non-negative, not identically zero weight function $\omega_{i}\mapsto \mu_i$. We can then define the following function $\mu^{\ast}$ on $\mathfrak{A}$:
\begin{center}
for any $X\subseteq\Omega$, let $\mu^{\ast}(X):=\dfrac{\vec{\mu}\cdot\ind{X}}{\vec{\mu}\cdot\ind{\Omega}}$.
\end{center}
Then $\mu^{\ast}$ is the desired probability distribution. It is a function $\mu^{\ast}:\mathfrak{A}\rightarrow [0,1]$ since we have re-normalised by dividing by $||\vec{\mu}||$. Additivity follows from the definition: for disjoint sets $A$, $B$ we have 
$$
\mu^{\ast}(A\cup B) = \dfrac{\vec{\mu}\cdot\ind{A\cup B}}{\vec{\mu}\cdot\ind{\Omega}} = \dfrac{\vec{\mu}\cdot(\ind{A}+\ind{B})}{\vec{\mu}\cdot\ind{\Omega}} = \dfrac{\vec{\mu}\cdot\ind{A}+ \vec{\mu}\cdot\ind{B}}{\vec{\mu}\cdot\ind{\Omega}}= \mu^{\ast}(A)+\mu^{\ast}(B).$$ 
That $\mu^{\ast}$ respects both the strict and non-strict inequalities imposed by the orderings $>$ and $\geqslant$ follows immediately from ($\ref{repsol}$). 
\end{proof}

We can also reformulate the theorem in the following way:

\begin{prop}[\textbf{Representation theorem, second formulation}]\label{strongrep2}
Let $\Omega$ a finite set, and let $<$ and $\leqslant$ be two nonempty binary relations on $\mathcal{P}(\Omega)$ such that $<\,\subseteq \, \leqslant$. The pair  $(<, \leqslant)$ is jointly weakly representable if and only if the following condition holds for all $n>0$:
\begin{itemize}
\item [\textsf{(S+)}] If $(A_{i})_{i\leq n}\leq_{0}(B_{i})_{i\leq n}$ and $A_{i}\geqslant B_{i}$ for all $i<n$, then $A_{n}\not>B_{n}$. 
\end{itemize}
\end{prop}
\begin{proof}
    Assuming $<\subseteq\leqslant$, the axiom \textsf{(S+)} is equivalent to \textsf{(GS)}. Simply observe that [$A_i\geqslant B_i$ or $A_i>B_i$] in the antecedent of \textsf{(GS)} becomes equivalent to $A\geqslant B$.
\end{proof}

Lastly, we can immediately get the following sufficient conditions for representability, which will be more directly applicable to the case of probabilistically stable revision:

\begin{prop}[\textbf{Sufficient conditions for joint representation}]\label{repr}
Let $<$ and $\leqslant$ be two partial relations on $\mathcal{P}(\Omega)$ such that for any $A$, $B\subseteq\Omega$:
\begin{itemize}
\item [\textsf{(A0)}] \hspace{1.1em}$\forall\omega\in\Omega$, $\{\omega\}>\emptyset$;
\item [\textsf{(A1)}] \hspace{1.1em}$A\leqslant B \Rightarrow A\not>B$;
\item [\textsf{(A2)}] \hspace{1.1em}$A<B \Rightarrow A\leqslant B$;
\item [\textsf{(Scott)}] If $(A_{i})_{i\leq n}\equiv_{0}(B_{i})_{i\leq n}$ and $\forall i\leq n$, $A_{i}\geqslant B_{i}$, then $\forall i\leq n$, $A_{i}\leqslant B_{i}$. 
\end{itemize}
Then, there is a regular probability measure $\mu$ on $\mathcal{P}(\Omega)$ such that, for all $A$, $B\subseteq\Omega$,
\begin{itemize}
\item $A<B \Rightarrow \mu(A)<\mu(B)$;
\item $A\leqslant B \Rightarrow \mu(A)\leq \mu(B)$.
\end{itemize}
\end{prop}
\begin{proof}
The property \textsf{(A2)} simply means that $<\,\subseteq\,\leqslant$: we apply the preceding proposition. We show that the premises $\textsf{(A0)}-\textsf{(A2)}$ and \textsf{(Scott)} together entail the condition for sufficiency $\textsf{(S+)}$. Suppose towards a contradiction that $\textsf{(S+)}$ fails. Then there exist two sequences of events $A=(A_{i})_{i\leq n}$  and $B=(B_{i})_{i\leq n}$ such that $(A_{i})_{i\leq n}\leq_{0}(B_{i})_{i\leq n}$, and  $A_{i}\geqslant B_{i}$ for all $i<n$, and $A_{n}>B_{n}$. We construct two balanced sequences of sets by adding to the sequence of $B_i$'s exactly as many instances of each singleton set $\{\omega\}$ as needed to make the sequences balanced. Recall that, given $\omega\in\Omega$ and a sequence of sets $S = (S_{i})_{i\leq n}$, the notation $\eta_{S}(\omega):=\big| \{ i\leq n \,|\, \omega\in S_{i}\} \big| $ denotes the number of occurrences of $\omega$ in $S$. Since $(A_{i})_{i\leq n}\leq_{0}(B_{i})_{i\leq n}$, for each element $\omega_i$ in the finite set $\Omega =\{\omega_i \,|\, i\leq k\}$, we know $\eta_A (\omega_i)\leq \eta_{B}(\omega_i)$. Then $\alpha_{i}:=\eta_{B}(\omega_{i})-\eta_{A}(\omega_{i})$ is the excess occurrences of $\omega_{i}$ in the sequence $B$.   So it is immediate that 
$$
    \big(A_{1},\dots, A_n, \alpha_{1}\{\omega_{1}\}\, ,\dots, \alpha_{k}\{\omega_{k}\}\big) \equiv_{0} \big(B_{1},\dots, B_n, \emptyset, \dots, \emptyset\big)
    $$
    (Here $\alpha_{i}\{\omega_{i}\}$ symbolises a sequence of $\alpha_i$ many copies of the singleton set $\{\omega_i\}$).  Call the sequence on the left $A^{\ast}=(A^{\ast}_{i})_{i\leq l}$ and the sequence on the right $B^{\ast}=(B^{\ast}_{i})_{i\leq l}$, each of length $l=n+\sum_{i\leq k}\alpha_{i}$. We know $A_{i}\geqslant B_{i}$ for all $i< n$: and we have $A_n > B_n$, so by \textsf{(A2)} we have $A_n \geqslant B_n$. By \textsf{(A0)} and \textsf{(A2)}, we have $\{\omega_i\} \geqslant \emptyset$ for all $\omega_{i}$. So we have $A^{\ast}_{i} \geqslant B^{\ast}_{i}$ for all $i\leq l$. By \textsf{(Scott)}, then, we must have $A^{\ast}_{i}\leqslant B^{\ast}_{i}$ for all $i\leq l$, and $A_n \leqslant B_{n}$ in particular. Yet, since we have $A_{n}>B_{n}$, from \textsf{(A1)} we conclude $A_{n}\not\leqslant B_{n}$. So, given axioms $\textsf{(A0)}-\textsf{(A2)}$, a failure of $\textsf{(S+)}$ entails a failure of the $\textsf{(Scott)}$ axiom.
\end{proof}
These conditions are manifestly not necessary. Our discussion in Section \ref{SectionRepAtDiffThresh} contains an example which shows that $\textsf{(S+)}$ is strictly weaker than \textsf{(Scott)}. Proposition \ref{repr} will be the key step in our representation theorem for selection functions for $t=1/2$.

\subsection{Relation to other notions of representability: strong, partial, and almost-representation}

Theorem \ref{strongrep} gives necessary and sufficient conditions for a pair of orders to be $(<,\leqslant)$ to be jointly represented by a single probability measure $\mu$, in the sense that $\mu$ \emph{partially represents} the $<$ order (i.e., $A<B \Rightarrow \mu(A)<\mu(B)$) and simultaneously \emph{almost represents} the $\leqslant$ order (i.e., $A\leqslant B \Rightarrow \mu(A)\leq \mu(B)$). Before we move on to the representation problem for selection structures, it is instructive to compare Theorem \ref{strongrep} to previous results in the comparative probability literature.

Our result gives necessary sufficient conditions for joint weak representability when the first order $<$ (the one to be represented in terms of strict probability comparisons) is nonempty. What about jointly representing a pair $(<,\leqslant)$ when $<$ is empty?\footnote{In this case the \textsf{(GS)} axiom is not sufficient: it holds vacuously, regardless of the almost-representability of $\leqslant$, since the conclusion $A_n\not>B_{n}$ always obtains.} This simply amounts to asking for a necessary and sufficient condition for almost representability. A necessary and sufficient condition for almost representability was given by \cite{KPS}, and it amounts to the following axiom: 
\begin{center}
\textsf{(AR)} \qquad If $(A_{i})_{i\leq n}<_{0} (B_{i})_{i\leq n}$ and $A_{i}\geqslant B_{i}$ for all $i<n$, then $A_{n}\not\geqslant B_{n}$.
\end{center}
The condition $(A_{i})_{i\leq n}<_{0} (B_{i})_{i\leq n}$ means that
$\sum_{i\leq n}\ind{{A_{i}}}< \sum_{i\leq n}\ind{{B_{i}}}$. In other words, it simply means that \emph{every} state in $\Omega$ occurs in \emph{strictly} more $B_{i}$s than $A_{i}$s. 

Now, it is immediate that a pair $(<,\leqslant)$ being jointly weakly representable entails that $\leqslant$ must be almost-representable as well. In particular, whenever $<$ nonempty, the joint representability condition $\textsf{(GS)}$ must then entail the almost representability condition \textsf{(AR)}. Perhaps one might worry that it isn't clear how the $\textsf{(GS)}$ axiom places the requisite constraints on the relation $\leqslant$ itself, particularly since it only warrants conclusions of the form $A_{n}\not> B_n$, rather the intuitively stronger $A_n\not\geqslant B_n$ required by \textsf{(AR)}. But in fact we can directly verify the following equivalence, and it is an elementary but useful exercise:

\begin{prop}
Let $\Omega$ finite and $\leqslant$ a binary relation on $\mathcal{P}(\Omega)$. The following are equivalent:
\begin{itemize}
    \item[(1)] The relation $\leqslant$ satisfies \textsf{(AR)};
     \item[(2)] If $\big( \Omega, A_1,\dots,A_n\big) \leq_{0} \big(\emptyset, B_1,\dots, B_n\big)$  and $A_{i}\geqslant B_{i}$ for all $i< n$, then $A_{n}\not\geqslant B_{n}$;
    \item[(3)] There exists a nonempty binary relation $<$ on $\mathcal{P}(\Omega)$ such that $(<,\leqslant)$ satisfies $\textsf{(GS)}$.
\end{itemize}
\end{prop}
\begin{proof}
(1)$\Rightarrow$(2): Assume $\leqslant$ satisfies $\textsf{(AR)}$. Suppose $\big( \Omega, A_1,\dots,A_n\big) \leq_{0} \big(\emptyset,B_1,\dots, B_n\big)$. Letting $A^{\ast}:= \big( \Omega, A_1,\dots,A_n\big) $ and $B^{\ast}:=\big(\emptyset, B_1,\dots, B_n\big)$, this simply means we have $\eta_{A^{\ast}}(\omega)\leq \eta_{B^{\ast}}(\omega)$ for every $\omega$. Consider now $A:= (A_1,\dots,A_n\big)$ and $B:= (B_1,\dots,B_n\big)$. We have, for each $\omega$, $\eta_{A^{\ast}}(\omega) = \eta_{A}(\omega)+1 \leq  \eta_{B^{\ast}}(\omega) = \eta_{B}(\omega)$, and so $  \eta_{A}(\omega)<  \eta_{B}(\omega)$. So $(A_{i})_{i\leq n}<_{0} (B_{i})_{i\leq n}$. If we have $A_{i}\geqslant B_{i}$ for all $i< n$, we can conclude $A_n\not\geqslant B_n$ by $\textsf{(AR)}$, which establishes (2).  

(2)$\Rightarrow$(3): Assume $\leqslant$ satisfies (2). Let $<:=\{(\emptyset, \Omega)\}$. Now suppose that the pair $(<,\leqslant)$ violates $\textsf{(GS)}$. This means that there are sequences $\big(A_1,\dots,A_n\big) \leq_{0} \big(B_1,\dots, B_n\big)$ where either $A_{i}\geqslant B_{i}$ or $A_{i}>B_{i}$ for all $i< n$, and $A_n > B_n$.  Note that among these there must be some non-strict comparison $A_{i}\geqslant B_{i}$ also, since otherwise we have $\big(\Omega,\dots,\Omega\big) \leq_{0} \big(\emptyset,\dots, \emptyset)$, which is false for $\Omega\neq\emptyset$. Since all strict comparisons $A_i > B_i$ just amount to $\Omega > \emptyset$, we can rearrange the sequences of $A_i$s and $B_i$s as
$\big( \Omega,\dots,\Omega,A^{\ast}_1,\dots,A^{\ast}_n\big) \leq_{0} \big(\emptyset,\dots, \emptyset, B^{\ast}_1,\dots, B^{\ast}_n\big)$, with all the strict comparisons put first followed by non-strict comparisons $A^{\ast}_i\geqslant B^{\ast}_i$. Now observe that $ \big( \Omega,A^{\ast}_1,\dots,A^{\ast}_n\big)\leq_{0} \big( \Omega,\dots,\Omega,A^{\ast}_1,\dots,A^{\ast}_n\big)$, as a result of which $\big( \Omega,\dots,\Omega,A^{\ast}_1,\dots,A^{\ast}_n\big) \leq_{0} \big(\emptyset,\dots, \emptyset, B^{\ast}_1,\dots, B^{\ast}_n\big)$ entails $\big( \Omega,A^{\ast}_1,\dots,A^{\ast}_n\big) \leq_{0} \big(\emptyset, B^{\ast}_1,\dots, B^{\ast}_n\big)$ (the right-hand-side sequences having the same number of occurrences of each $\omega\in\Omega$). This is a violation of (2). 

(3)$\Rightarrow$(1): Suppose $(<,\leqslant)$ satisfies \textsf{(GS)} and $<$ is nonempty: say, we have $C>D$ for some $C,D\in\mathcal{P}(\Omega)$. We show that $\leqslant$ satisfies \textsf{(AR)}. 
Suppose we have $(A_{i})_{i\leq n}<_{0} (B_{i})_{i\leq n}$ and $A_{i}\geqslant B_{i}$ for all $i<n$. Towards a contradiction, suppose $A_{n}\geqslant B_{n}$. This would immediately entail a violation of \textsf{(GS)}. It is enough to note that $(A_{1},\cdots, A_n, C)\leq_{0} (B_{1}, \dots, B_n, D)$: because $(A_{i})_{i\leq n}<_{0} (B_{i})_{i\leq n}$, every state occurs in strictly more $B_i$s than $A_i$s, so adding $C$ to the $A_i$-sequence adds only one occurrence of each $\omega \in C$ on the left-hand side, which is still \emph{no greater} than the total number of occurrences in the $B_i$-sequence\textemdash and, a fortiori, no greater the number of occurrences in in $(B_{1}, \dots, B_n, D)$. Finally, since $A_i\geqslant B_{i}$ for all $i$, \textsf{(GS)} requires $C\not>D$, contradicting the choice of $C,D$ such that $C>D$. 
\end{proof}

\noindent So our condition $\textsf{(GS)}$ does indeed entail the required condition for almost representability. 

Similarly, we can readily verify that \textsf{(GS)} entails the axiom that characterises the \emph{partial representability} of the relation $<$. Here the relevant representation theorem is due to \cite{FISH}. It gives necessary and sufficient conditions for the partial representation of a strict qualitative probability order on a finite algebra.

\begin{theorem}[\cite{FISH}]\label{FISHTHM}
Let $\Omega$ a finite set, and $<$ a binary relation on $\mathcal{P}(\Omega)$. Then there is a measure $\mu$ on $\mathcal{P}(\Omega)$ which \emph{partially represents} the relation $<$, meaning
$$
\forall A, B\in\mathcal{P}(\Omega), \,\,A< B \Rightarrow \mu(A)<\mu(B)
$$
if and only if the following holds for any $n>0$:
\begin{itemize}
\item[\emph{\textsf{(F)}}] If $(A_{i})_{i\leq n} \leq_{0} (B_{i})_{i\leq n}$ and $A_{i}> B_{i}$ for all $i < n$, then $A_{n}\not> B_{n}$.
\end{itemize}
\end{theorem}

 It is immediate that $\textsf{(GS)}$ entails Fishburn's axiom \textsf{(F)}: any sequences of sets meeting the antecedent of $\textsf{(F)}$ automatically meet the antecedent of $\textsf{(GS)}$. 

Recall now the question of providing necessary and sufficient conditions for the \emph{strong representability} of a relation $\leqslant$, which amounts to the existence of a measure that both (i) almost-represents the relation  $\leqslant$ and (ii) partially represents the derived `strict' relation $<$, where $A<B$ is \emph{defined} as [$A\leqslant B$ and $B\not\leqslant A$]. In particular, \cite{KON} asks how to characterise strong representability exactly, without assuming non-necessary conditions like Non-triviality $(\emptyset < \Omega)$ and Non-negativity ($\emptyset \leqslant A$ for all $A\in\mathcal{P}(\Omega)$). From our representation result, we directly obtain as corollary an answer to the question of strong representability:\footnote{\cite{KON} gives a strenghtening of Theorem \ref{Scott's Theorem} (due to Kraft, Pratt, Seidenberg and Scott). Konek's result provides a condition for strong representability under the additional assumptions of Non-triviality and Non-negativity: the condition in question simply amounts to the \textsf{(Scott)} axiom. Note that, in the case where the order $<$ is defined from $\leqslant$, $\textsf{(GS)}$ becomes: 
\begin{itemize}
    \item[] If $(A_{i})_{i\leq n}\leq_{0}(B_{i})_{i\leq n}$ and $A_{i}\geqslant B_{i}$ for all $i<n$,  then either $A_{n}\not\geqslant B_{n}$ or $B_{n}\geqslant A_{n}$.
\end{itemize}
Which is equivalent to: 
\begin{itemize}
    \item[] If $(A_{i})_{i\leq n}\leq_{0}(B_{i})_{i\leq n}$ and $A_{i}\geqslant B_{i}$ for all $i\leq n$, then $B_{n}\geqslant A_{n}$.
\end{itemize}
Here the antecedent is insensitive to the order of the pairs $(A_i, B_i)$, so that $\textsf{(GS)}$ just states:
\begin{itemize}
    \item[] If $(A_{i})_{i\leq n}\leq_{0}(B_{i})_{i\leq n}$ and $A_{i}\geqslant B_{i}$ for all $i\leq n$, then $B_{i}\geqslant A_{i}$ for all $i\leq n$.
\end{itemize}
It is worth noting that, for the defined order $<$, the condition \textsf{GS} then simply amounts to a strengthening of the \textsf{Scott} axiom where the balancedness condition $\equiv_0$ in the premises is strengthened to the $\leq_0$ relation. } 

\begin{cor}[\textbf{Strong representation}]
A binary relation $\leqslant$ on $\mathcal{P}(\Omega)$ is strongly representable if and only if either:
\begin{itemize}
 \item $\leqslant$ is symmetric and satisfies \textsf{(AR)}, or
    \item $\leqslant$ is not symmetric and satisfies \textsf{(GS)}, with $A<B$ defined as [$A\leqslant B$ and $B\not\leqslant A$]. Explicitly, this amounts to 
\begin{itemize}
\item [\textsf{(S+)}] If $(A_{i})_{i\leq n}\leq_{0}(B_{i})_{i\leq n}$ and $A_{i}\geqslant B_{i}$ for all $i<n$, then $A_{n}\not>B_{n}$. 
\end{itemize}
\end{itemize}
\end{cor}
\begin{proof}
If $\leqslant$ is symmetric, then $<$ defined as above is empty, and so strong representability just amounts to the almost-representability of $\leqslant$, which is captured by \textsf{(AR)}. If $\leqslant$ is not symmetric, then $<$ is nonempty, and strong representability amounts to the joint representability of the pair $(<, \leqslant)$, and so Theorem \ref{strongrep} applies. Note that the disjuction [$A\geqslant B$ or $A>B$] becomes equivalent to $A\geqslant B$, so \textsf{(GS)} can be expressed in the simpler form \textsf{(S+)}.
\end{proof}

 We can now move on to giving the solution to our representation problem for $t=1/2$.

\subsection{Representation theorem for selection structures}

We begin by defining an order relation induced by selection functions. 

\begin{mydef}[The $\succcurlyeq^{\ast}_{\sigma}$ order]
Given a selection structure $(\Omega,\mathfrak{A},\sigma)$ and $X\cup \{\omega\} \subseteq\Omega$, write $\omega\succ_{\sigma} X$ whenever $\sigma(X\cup\{\omega\})=\{\omega\} \cap X^{c}$, and $X\succeq_{\sigma} \omega$ whenever $\sigma(X\cup\{\omega\})\neq\{\omega\} \cap X^{c}$. Extend this to a relation $\succcurlyeq^{\ast}_{\sigma}$ given by the condition:
$$
A\succcurlyeq^{\ast}_{\sigma} B \Leftrightarrow  A\succ_{\sigma} B \text{   or   } A\succeq_{\sigma} B.
$$
\end{mydef}
A few words about this order. First note that we have 
$$\succcurlyeq^{\ast}_{\sigma}\,=
\big\{(\{\omega\}, X)\,\big|\, \omega\succ_{\sigma} X\big\} \cup \big\{(X,\{\omega\})\,\big|\, \omega\not\succ_{\sigma} X \big\}$$ with $\omega\in\Omega,\,X\subseteq\Omega$. This means that $A\succcurlyeq^{\ast}_{\sigma} B$ is defined exactly when at least one of $A$, $B$ is a singleton (atom) in the algebra. It is undefined otherwise. Intuitively, $A\succcurlyeq^{\ast}_{\sigma} B$ means that either $A$ is a singleton `dominating' the set $B$\footnote{In the sense that $\mu(A)>\mu(B)$ for any measure that represents $\sigma$.} or that $B$ is a singleton which fails to dominate $A$. An alternative way to define this ordering is by means of the following equivalences:\footnote{These equivalences are making use of the fact that we are dealing with selections functions meeting the regularity property (\textsf{S1}), namely, that $\sigma(X)=\emptyset$ only if $X=\emptyset$.}
\begin{align*}
\omega\succ_{\sigma} X &\Longleftrightarrow \omega\not\in X \text{ and } \sigma(X\cup\{\omega\})=\{\omega\}, \\
X\succeq_{\sigma} \omega &\Longleftrightarrow \sigma(X\cup\{\omega\})\neq\{\omega\}   \text{ or } \omega\in X 
\end{align*}
The following is immediate:
\begin{prop}\label{prop:A1A2}
Let $\sigma$ be a selection function. Then the following holds:
\begin{itemize} 
\item [\emph{\textbf{(A1)}}]$A\preccurlyeq^{\ast}_{\sigma}B \Rightarrow A\not\succ_{\sigma} B$.
\item [\emph{\textbf{(A2)}}] $A\prec_{\sigma} B \Rightarrow A\preccurlyeq^{\ast}_{\sigma} B$.
\end{itemize}
\end{prop}
Consider now the Scott axiom for the relation $\succcurlyeq^{\ast}_{\sigma}$:
\begin{center}
\textsf{(Scott)} \qquad If $(A_{i})_{i\leq n}\equiv_{0}(B_{i})_{i\leq n}$ and $\forall i\leq n$, $A_{i}\succcurlyeq^{\ast}_{\sigma}B_{i}$, then $\forall i\leq n$, $A_{i} \preccurlyeq^{\ast}_{\sigma} B_{i}$. 
\end{center}
The Scott axiom plays a key role in Proposition \ref{repr}. In the context of selection rules, it is indeed a very powerful property: firstly, we check that it is sound with respect to our probabilistic interpretation. 

\begin{prop}
The \emph{(\textsf{Scott})} axiom for $\succcurlyeq^{\ast}_{\sigma}$ is probabilistically sound. That is, given any finite probability space $(\Omega, \mathcal{P}(\Omega), \mu)$, define $\sigma:\mathcal{P}(\Omega)\rightarrow \mathcal{P}(\Omega)$ by $\sigma(A):=\tau(\mu_{\scriptscriptstyle{A}})$ for threshold $t=1/2$. Then the \emph{(\textsf{Scott})} property holds for $\sigma$. 
\end{prop}
\begin{proof}
Let $\succcurlyeq^{\ast}_{\sigma}$ the order obtained from the revison plan $\sigma$ obtained from $\mu$, so that $\mu$ jointly represents both $\succ_{\sigma}$  and $\succcurlyeq^{\ast}_{\sigma}$. We assume the measure $\mu$ is regular (recall we are only interested in representation by regular measures). Assume that $(A_{i})_{i\leq n}\equiv_{0}(B_{i})_{i\leq n}$ and $\forall i\leq n$, $A_{i}\succcurlyeq^{\ast}_{\sigma}B_{i}$. We show that $\forall i\leq n$, $A_{i}\preccurlyeq^{\ast}_{\sigma} B_{i}$\footnote{Here the only subtlety is the following: given the premises, it is immediate that we must have $\mu(A_{i})\leq \mu(B_{i})$ for all $i\leq n$. This in itself does not entail $A_{i}\preccurlyeq^{\ast}_{\sigma} B_{i}$, however: we only have that  $A\succcurlyeq^{\ast}_{\sigma}B$ entails $\mu(A)\geq \mu(B)$, but the converse direction may not hold. For instance, it could be that $\mu(A_{n})=\mu(B_{n})$ and $B_{n}$ is an atom in the algebra while $A_{n}$ is not, in which case we cannot have $A_{n}\preccurlyeq^{\ast}_{\sigma}  B_{n}$, as we can have neither $A_{n}\preceq_{\sigma} B_{n}$ (due to the domain restrictions of $\preceq$, which only allow this to hold when $A_{n}$ is a singleton) nor $A_{n}\prec_{\sigma} B_{n}$ (as this contradicts  $\mu(A_{n})=\mu(B_{n})$). Thus we must make sure that his never occurs in balanced sequences.}. 

First, note that we cannot have any strict relation $A_{i}\succ_{\sigma} B_{i}$, as it would immediately entail 
$$
\sum^{n}_{i=1} \mu(A_{i}) >\sum^{n}_{i=1}  \mu(B_{i}).
$$
But this would contradict the fact that $(A_{i})_{i\leq n}\equiv_{0}(B_{i})_{i\leq n}$, since the fact that these sequences are balanced entails
$$\sum^{n}_{i=1}\sum_{\omega\in A_{i}} \mu(\omega)=\sum^{n}_{i=1} \mu(A_{i}) =\sum^{n}_{i=1} \mu(B_{i})=\sum^{n}_{i=1}\sum_{\omega\in B_{i}} \mu(\omega).$$

By definition of $\succcurlyeq^{\ast}_{\sigma}$ above, we can have $A_{i}\not\succ_{\sigma} B_{i}$ and $A_{i}\succcurlyeq^{\ast}_{\sigma}B_{i}$ only if $B_{i}$ is a singleton $B_{i}=\{b_{i}\}$ such that $\sigma(A_{i}\cup\{b_{i}\})\neq\{b_{i}\}$. So all $B_{i}$'s must be singletons. 

Next we note that all $A_{i}$'s must be singletons as well. Firstly, all $A_{i}$'s must be nonempty, by regularity of $\mu$: for otherwise $\emptyset=A_{i}\succcurlyeq^{\ast}_{\sigma}\{b_{i}\}$ means that $\sigma(\emptyset\cup \{b_{i}\})\neq \{b_{i}\}$, while regularity enforces $\sigma(\{\omega\})=\{\omega \}$ for any state. Now suppose that some $A_{i}$ contains different states $\{a_{1},...,a_{k}\}$. We have $\{a_{1},...,a_{k}\}\succcurlyeq^{\ast}_{\sigma}b_{i}$ and since the sequences are balanced, each of those $a_{j}$'s must appear as a singleton $\{b_{j}\}$ in the sequence $(B_{i})_{i\leq n}$. Now a contradiction follows by a counting argument: count all \emph{occurrences} of elements in the sets $A_{j}$ and do the same the sets  $B_{j}$. Since all $B_{j}$'s are singletons we have 
$$\sum_{\omega\in\Omega}| \{j\leq n\,|\, \omega\in B_{j}\}| = \sum_{j\leq n} |B_{j}|= n.$$
Each occurrence of a state $\omega$ in any of the $A_{j}$ must be matched by an occurrence in one of the $B_{j}$. But there is at least one such occurrence for each $A_{j}$, since they are nonempty, and strictly more than one occurrence for $A_{i}=  \{a_{1},...,a_{k}\}$, so $\sum_{\omega\in\Omega}| \{j\leq n\,|\, \omega\in A_{j}\}|>n$. This entails 
$$
\sum_{i\leq n}\ind{A_{i}}>\sum_{i\leq n}\ind{B_{i}},
$$
contradicting the fact that $(A_{i})_{i\leq n}$ and $(B_{i})_{i\leq n}$ are balanced. So all $A_{i}$'s must indeed be singletons as well.

So the sequences $(A_{i})_{i\leq n}$ and $(B_{i})_{i\leq n}$ are really sequences of singletons $(a_{i})_{i\leq n}$ and $(b_{i})_{i\leq n}$ such that $\mu(a_{i})\geq \mu (b_{i})$ for all $i\geq n$, and since they are balanced it follows that $\mu(a_{i})\leq \mu (b_{i})$. Now note that, given the definition of $\succcurlyeq^{\ast}_{\sigma}$, whenever $a_{i}$ and $b_{i}$ are singletons then $\mu(a_{i})\leq \mu (b_{i})$ entails $a_{i}\preccurlyeq^{\ast}_{\sigma} b_{i}$: for if $a_{i}\not\preccurlyeq^{\ast}_{\sigma} b_{i}$, then $\sigma(\{a_{i},b_{i}\})= a_{i}\setminus\{b_{i}\}\neq\emptyset$, which means $a_{i}\succ_{\sigma} b_{i}$, which entails $\mu(a_{i})>\mu(b_{i})$.  So, we can conclude that $A_{i}\preccurlyeq^{\ast}_{\sigma} B_{i}$ holds for all $i\leq n$, as desired. \end{proof}

One may also verify, although we omit the argument here, that the (\textsf{Scott}) axiom imposed on $\succcurlyeq^{\ast}_{\sigma}$, together with properties (\textsf{S1}) and (\textsf{S2}) listed below (as introduced in Proposition \ref{prop:stab2}), entail \emph{all} of the desired properties for selection functions introduced in the previous section (page \pageref{SLISTb})\footnote{That is, if $\sigma$ satisfies (\textsf{S1}) and (\textsf{S2}), and the induced order $\succcurlyeq^{\ast}_{\sigma}$ satisfies the (\textsf{Scott}) axiom, then $\succ_{\sigma}$ satisfies the properties (M1)-(M5) and (Sc): further, the property (\textsf{S5}), introduced just before Proposition \ref{prop:stab2}, also follows.}. 

We can now prove the representation theorem:
\begin{theorem} [\textbf{Representation theorem for selection structures}] \label{REPSEL} 
Let $(\Omega, \mathcal{P}(\Omega), \sigma)$ be a selection structure satisfying the following:
\begin{itemize}
\item [($\textsf{S1}$)]$\sigma(X)=\emptyset$ only if $X=\emptyset$
\item [($\textsf{S2}$)]$\sigma(X)\subseteq X$
\item [($\textsf{S3}$)]If $\sigma(A)\cap B\neq\emptyset$, then $\sigma(A\cap B)\subseteq \sigma(A)\cap B$
\item [($\textsf{S4}_{n}$)]For any $n$: if $\sigma(A\cup X_{i})=X_{i}$ for all $i\leq n$, then $\sigma(A\cup\bigcup_{i\leq n}X_{i})\subseteq \bigcup_{i\leq n}X_{i}$
\item [\textsf{(Scott)}] If $(A_{i})_{i\leq n}\equiv_{0}(B_{i})_{i\leq n}$ and $\forall i\leq n$, $A_{i}\succcurlyeq^{\ast}_{\sigma}B_{i}$, then $\forall i\leq n$, $A_{i} \preccurlyeq^{\ast}_{\sigma} B_{i}$. 
\end{itemize}
Then there is a (regular) probability measure representing $\sigma$. Conversely, for any probability space $(\Omega, \mathcal{P}(\Omega), \mu)$ with $\mu$ a regular measure, the strongest stable set operator $\sigma_{\mu}: X \mapsto \tau(\mu_{\scriptscriptstyle X})$ for threshold $t=1/2$ satisfies axioms \emph{($\textsf{S1}$)}\textemdash \emph{($\textsf{S4}_{n}$)} and $\emph{\textsf{(Scott)}}$. 
\end{theorem}

\begin{proof}
The second part\textemdash the probabilistic soundness of the axioms\textemdash is straightforward (and has, for the most part, been verified in the previous sections: the only remaining scheme to check, ($\textsf{S4}_{n}$), is indirectly shown to be sound in section \ref{MINIMOP}). We show that the axioms suffice for probabilistic representability. Let $(\Omega, \mathcal{P}(\Omega), \sigma)$ a selection structure as above. Observe that the relations $\succ_{\sigma}$  and $\succcurlyeq^{\ast}_{\sigma}$ satisfy all of the conditions in Proposition \ref{repr}: the (\textsf{Scott}) property is given. Next, ($\textsf{S1}$) ensures that $\sigma(\{\omega\})\neq\emptyset$ for all $\omega\in\Omega$, so that ($\textsf{S2}$) ensures $\sigma(\{\omega\})=\{\omega\}$, so the following holds:
$$\textbf{(A0)}\,\,\,\,\forall\omega\in\Omega,\, \omega\succ_{\sigma}\emptyset.$$
Further, by Proposition \ref{prop:A1A2}, we have 
$$\textbf{(A1)}\,\,\,\,  A\preccurlyeq^{\ast}_{\sigma}B \Rightarrow A\not\succ_{\sigma} B   $$
$$\textbf{(A2)}\,\,\,\,  A\prec_{\sigma} B \Rightarrow A\preccurlyeq^{\ast}_{\sigma} B          $$
Given this, Proposition \ref{repr} entails that there exists a regular probability measure $\mu$ such that for any $A$, $B\subseteq\Omega$: 
\begin{itemize}
\item $A\prec_{\sigma} B \Rightarrow \mu(A)<\mu(B)$
\item $A\preccurlyeq^{\ast}_{\sigma}B \Rightarrow \mu(A)\leq \mu(B)$
\end{itemize}
By definition of $\prec_{\sigma}$ and $\preceq_{\sigma}$, this entails that we have, for any $\omega\in\Omega$ and $X\subseteq \Omega$ with $\omega\not\in X$: 
\begin{itemize}
\item If $\sigma(X\cup\{\omega\})=\{\omega\}$ (equivalently, $X\prec_{\sigma} \omega$) then $\mu(X)<\mu(\omega)$
\item If $\sigma(X\cup\{\omega\})\neq\{\omega\}$, then $\omega \preccurlyeq^{\ast}_{\sigma} X$ and so $\mu(\omega)\leq \mu(X)$
\end{itemize}
This means that the measure $\mu$ agrees with $\sigma$ on all pairs $(\omega,X)\in\Omega\times\mathfrak{A}$, and so solves the system consisting of all $\sigma$-generated inequalities in $L_{\sigma}$. By Observation \ref{GEOMETRY}, the system $L_{\sigma}$ uniquely identifies a consequence relation $\nc_{\mu}$, or equivalently, a probabilistically stable revision plan. The selection function $\sigma$ satisfies all of $(\textsf{S1})-(\textsf{S4})$ so, by Proposition \ref{prop:stab2} (and the discussion immediately preceding it), we have that $\sigma(A)=B$ if and only if $\tau(\mu_{A})=B$, and thus $\mu$ represents the selection function $\sigma$. 
\end{proof}
This gives the solution to our representation problem: the selection function $\sigma$ is a strongest-stable-set operator (generated by some probability measure and threshold $t=1/2$) if and only if it satisfies the properties (\textsf{S1}), (\textsf{S2}), (\textsf{S3}), ($\textsf{S4}_{n}$) and  (\textsf{Scott}). Thus Theorem \ref{REPSEL} gives a full qualitative description of probabilistically stable revision plans on finite probability spaces. This concludes our discussion of the representation of strongest-stable-set operators by means of selection functions for threshold $t=1/2$. Before we move on to the case of other thresholds, we continue with a few remarks about the axioms.

\subsection{Scott axioms and Fishburn axioms for comparative probability}\label{FISHPAR}

Firstly, note that none of the results in this section section relied on the fact that we worked with full powerset algebras on $\Omega$: this simply made our notation more convenient, as we could refer to \emph{singleton sets} $\{\omega\}$ for $\omega\in\Omega$, instead of talking about \emph{atoms in the underlying algebra}. All of the above results \textemdash and Theorem \ref{REPSEL} in particular \textemdash hold for any pair $(\Omega,\mathfrak{A})$ where $\mathfrak{A}$ is a subalgebra of $\mathcal{P}(\Omega)$, as long we work in finite spaces and replace any mention of `singletons in $\mathcal{P}(\Omega)$' with `atoms in $\mathfrak{A}$'.

Secondly, in order to get a better grasp on the connection between selection functions and the theory of comparative probability orders, it will be useful to think about what selection functions can tell us about the underlying probability comparisons. Consider a representable selection structure $(\Omega,\mathfrak{A}, \sigma)$ (that is, a selection structure satisfying the necessary and sufficient conditions from the representation theorem given above). Given a representable selection function $\sigma$, which probabilistic inequalities of the form $\mu(A)>\mu(B)$ must hold for any measure $\mu$ representing $\sigma$? In what ways can we express the fact that $\mu(A)>\mu(B)$, using only the selection function $\sigma$? The following is immediate:
\begin{observation}
For any representable selection function $\sigma$ and measure $\mu$ representing $\sigma$, we have that  $\sigma(A\cup B)\subseteq A\setminus B$ entails $\mu(A)>\mu(B)$. 
\end{observation}
\begin{proof}
For any $\omega\in\sigma(A\cup B)$, by Proposition \ref{prop:stab1} we have $\mu(\omega)>\mu \big((A\cup B)\setminus \sigma(A\cup B) \big)$. Since $\sigma(A\cup B)\subseteq A\cap B^{c}$, we have  $B\subseteq (A\cup B)\setminus \sigma(A\cup B) $ which entails $\mu(\omega)>\mu(B)$ for any $\omega\in\sigma(A\cup B)$. So $\mu\big(\sigma(A\cup B)\big)>\mu(B)$, and we have $\sigma(A\cup B)\subseteq A$, so $\mu(A)>\mu(B)$.
\end{proof}

Of course, we know that representable selection functions always satisfy $\sigma(A)\subseteq A$, and so the condition $\sigma(A\cup B)\subseteq A\cap B^{c}$ can be rewritten as $\sigma(A\cup B)\subseteq B^{c}$.

Now, in order to understand the relation between selection structures and their underlying comparative probability orders, we would like to express the fact that $\mu(A)>\mu(B)$ using only the language of selection functions\footnote{This is meant informally; we have not introduced a formal language. We will expand on this point in Section \ref{Logics of probabilistic stability}.}. The above gives us one sufficient condition; but we can say more. Consider the following case. 

\begin{mydef}[Separation order]\label{SEPDEF}
Let $(\Omega, \mathfrak{A}, \sigma)$ a selection structure with $A, B, D\in\mathfrak{A}$. We say that $D$ \emph{separates} $A$ \emph{from} $B$ (written $A\triangleright_{D} B$) whenever the following conditions hold:
\begin{itemize} 
\item $D\cap (A\cup B) =\emptyset$
\item $\sigma(D\cup B)\subseteq B^{c}$
\item  $\sigma(D\cup A)\not\subseteq A^{c}$
\end{itemize}
\end{mydef}

In other words, $D$ separates $A$ from $B$ whenever $D$ is disjoint form both $A$ and $B$, and in addition $D$ dominates $B$, but does not dominate $A$. It is immediate then that we have:

\begin{prop}
Let $(\Omega, \mathfrak{A}, \sigma)$ a representable selection structure and $\mu$ a probability measure representing $\sigma$. Then for any $A, B, D\in\mathfrak{A}$, if $A\triangleright_{D} B$ then $\mu(A)> \mu(B)$. 
\end{prop}

We can combine the two observations above and define the following relation. 

\begin{mydef}[Dominance relation]\label{DOMDEF}
Let $(\Omega, \mathfrak{A}, \sigma)$ a selection structure. We define the relation $\triangleright\subseteq\mathfrak{A}\times\mathfrak{A}$ as follows:
\begin{center}
$A\triangleright B$ if and only if $\sigma(A\cup B)\subseteq B^{c}$ or $A\triangleright_{D} B $ for some $D\in\mathfrak{A}$.
\end{center}
\end{mydef}

This gives us another way to express the fact that $\mu(A)>\mu(B)$ holds for \emph{any} representing probability function: namely $A\triangleright B$ entails $\mu(A)>\mu(B)$ for any $\mu$ representing the selection function $\sigma$. Thus the extended ordering $\triangleright$ indeed entails that one event must have higher probability than another.

Since the Fishburn axiom \textsf{(F)} from Theorem \ref{FISHTHM} is necessary for any measure weakly representing the comparative ordering on propositions forced by $\sigma$, the resulting order relation must satisfy the corresponding form of the Fishburn axiom. In particular then, a generalised version of the Fishburn axiom must also hold for the order $\triangleright$.

\begin{mydef}[Generalised Fishburn Axiom]
A selection structure $(\Omega,\mathfrak{A}, \sigma)$ satisfies the \emph{Generalised Fishburn Axiom} if and only if, for any $A,B\in\mathfrak{A}$, we have:
\begin{center}
Whenever $(A_{i})_{i\leq n} \leq_{0} (B_{i})_{i\leq n}$ and $A_{i}\triangleright B_{i}$ for all $i\leq n-1$, then $\neg(A_{n}\triangleright B_{n})$.
\end{center}
\end{mydef}

The required properties for representation (namely ($\textsf{S1}$), ($\textsf{S3}$), ($\textsf{S4}_{n}$) and ($\textsf{Scott}$)) entail the Generalised Fishburn Axiom. 

\begin{prop}
Any representable selection structure $(\Omega,\mathfrak{A}, \sigma)$ satisfies the Generalised Fishburn Axiom. 
 \end{prop}
 
This follows from our representation result and Fishburn's theorem \ref{FISHTHM}. It can also be verified by a direct, elementary proof, which we omit but encourage to reader to spell out: the argument, while somewhat tedious, is instructive in as much as it illustrates how the qualitative axioms for selection structures can be put to use to directly derive (without appealing to a geometric or algebraic argument) the condition on systems of linear inequalities that is implicitly captured by the Fishburn axiom. 

The Generalised Fishburn Axiom can be employed to characterise an interesting class of structures that approximate probabilistically stable revision.

\begin{mydef}[Fishburn structures]
A \emph{Fishburn} structure is a selection structure $(\Omega, \mathfrak{A}, \sigma)$ satisfying the following for all $X, A, B \in\mathfrak{A}$:
\begin{itemize}
\item [\emph{($\textsf{S1}$)}]$\sigma(X)=\emptyset$ only if $X=\emptyset$
\item [\emph{($\textsf{S2}$)}]$\sigma(X)\subseteq X$
\item [\emph{($\textsf{S3}$)}]If $\sigma(A)\cap B\neq\emptyset$, then $\sigma(A\cap B)\subseteq \sigma(A)\cap B$
\item [\emph{($\textsf{S4}_{n}$)}]For any $n$: if $\sigma(A\cup X_{i})=X_{i}$ for all $i\leq n$, then $\sigma(A\cup\bigcup_{i\leq n}X_{i})\subseteq \bigcup_{i\leq n}X_{i}$
\item [\emph{\textsf{(GFA)}}] Whenever $(A_{i})_{i\leq n} \leq_{0} (B_{i})_{i\leq n}$ and $A_{i}\triangleright B_{i}$ for all $i\leq n-1$, then $\neg(A_{n}\triangleright B_{n})$.
\end{itemize}
\end{mydef}

We can observe the following: 

\begin{observation}\label{OBSonFISH}
Let $(\Omega,\mathfrak{A},\sigma)$ a Fishburn structure. There exists a measure that partially represents the induced order $\triangleright$. If $\mu$ is any such measure, we have that for any $A\in\mathfrak{A}$, the event $\sigma(A)\in\mathfrak{A}$ is $\mu(\cdot\,|\,A)$-stable (for threshold $t=1/2$).
\end{observation}
\begin{proof}
Given the Generalised Fishburn Axiom \textsf{(GFA)}, by Theorem \ref{FISHTHM}, there exists a measure $\mu$ on $\mathfrak{A}$ such that for any $X$, $Y\in\mathfrak{A}$ we have that $X \triangleright Y$ entails $\mu(X)>\mu(Y)$. Let $A\in\mathfrak{A}$ with $A\neq\emptyset$. We have $\sigma(A)\neq\emptyset$ by ($\textsf{S1}$). We show that $\sigma(A)$ is $\mu_{A}$-stable. Firstly, $\sigma(A)\subseteq A$ by ($\textsf{S2}$). All we need to show is the following: for every $\omega\in \sigma(A)$, we have $\mu_{A}(\omega\,|\,A\setminus \sigma(A))>1/2$. Let $\omega\in\sigma(A)$. Write $B:=(A\setminus\sigma(A))\cup \{\omega\}$. Then $\sigma(A)\cap B = \{\omega\}\neq\emptyset$ so by ($\textsf{S3}$) we have $\sigma (A\cap B) \subseteq \sigma(A)\cap B$, which means 
$$\sigma\big( [A\setminus \sigma(A)] \cup \{\omega\}\big) = \{\omega\}.$$ This entails $\{\omega\}\triangleright [A\setminus \sigma(A)]$, so $\mu(\omega)> \mu(A\setminus \sigma(A))$. We get $\mu_{A}(\omega\,|\,A\setminus \sigma(A))>1/2$, as desired. Note that the argument does not rely on the axiom ($\textsf{S4}_{n}$). 
\end{proof}
The upshot is that Fishburn structures capture a class of revision operators that \emph{respect probabilistic stability}: given any (nonempty) revision input $A$, the strongest accepted proposition $\sigma(A)$ is stable with respect to the updated measure $\mu(\cdot\,|\,A)$. However, it need not be the \emph{logically strongest} stable set, and the selection functions is not guaranteed to capture a probabilistically stable revision plan (a strongest-stable-set operator). 

Fishburn structures thus comply with Leitgeb's `non-reductionistic' Humean thesis on Belief. There, the only criterion imposed on revisions was simply that the strongest accepted proposition be probabilistically stable with respect to the updated measure. In our discussion we noted the stability constraint \emph{alone} was too weak to identify interesting revision operators (and to rule out trivial revisions which always select the least set with probability $1$ after conditioning). By contrast, the revision operators captured by Fishburn structures constitute a relatively well-behaved family that complies with the stability requirement: Fishburn structures approximate probabilistically stable revision, in that they satisfy the strong monotonicity principle ($\textsf{S3}$) corresponding to Rational Monotonicity, as well as the axiom ($\textsf{S4}_{n}$). Moreover, they are partially representable by a probability measure, in that the strict dominance order generated by the selection function is numerically representable. This wider class of revision operators also admits a simpler axiomatisation, as the $\triangleright$ relation\textemdash as opposed to the relation $\succcurlyeq_{\sigma}$ employed in the (\textsf{Scott}) axiom\textemdash does not depend on the property of being an atom in the algebra. We discuss next the role of the scheme ($\textsf{S4}_{n}$) in our representation theorem and its relation to the (\textsf{Or}) rule.  

\subsection{Minimisation operators and the \textsf{Or} rule again} \label{MINIMOP}

As we pointed out in Section \ref{QualiMod}, strongest-stable-set operators cannot be represented as a map $X\mapsto\min(R\hspace{-0.25em}\restriction\hspace{-0.25em} X)$ for some order relation $R\subseteq\Omega^{2}$. In other words, probabilistically stable revision cannot be tracked using a \emph{minimisation operator} for a plausibility relation. We also observed that strongest-stable-set operators, treated as selection functions, do not validate the \textsf{(Or)} rule. The former fact easily follows from the latter. As noted by e.g. \cite{vBE3} and \cite{ROT} the following conditions are necessary and sufficient for representability as a minimisation operator. 

\begin{prop}[\cite{vBE3}]\label{MINIMOPPROP}
Let $\Omega$ a finite set. Given a function $\sigma:\mathcal{P}(\Omega)\rightarrow \mathcal{P}(\Omega)$, the following are equivalent:
\begin{itemize}
\item $\sigma$ satisfies the following properties for any $X_{i}\subseteq \Omega$:
	\begin{itemize}
	\item[(1)] $\sigma(X)\subseteq X$
	\item[(2)] $\sigma(\bigcup_{i\leq n}X_{i})\subseteq \bigcup_{i\leq n} \sigma(X_{i})$
	\item[(3)] $\bigcap_{i\leq n} \sigma(X_{i})\subseteq \sigma(\bigcup_{i\leq n}X_{i})$
	\end{itemize}
\item There is an asymmetric binary relation $R\subseteq\Omega^{2}$ such that for all $X$:
$$
\sigma(X)= \min_{R}(X) := \{\omega\in X\,|\, \neg\exists v \in X,\, R(v,\omega)\}
$$ 
\end{itemize}
\end{prop}

As a quick verification reveals, strongest-stable-set operators validate both (1) and (3), but fail (2). The failure of (2) is unsurprising: interpreting $X_{i}\nc X_{j}$ as $\sigma(X_{i})\subseteq X_{j}$, the property corresponds to the (\textsf{Or}) rule. 

It is worth noting, however that a weaker form of the (\textsf{Or}) rule does obtain for probabilistically stable revision:

\begin{observation}\label{weakOr}
Fix any threshold $t\in[0.5,1)$. For any $t$-representable selection function $\sigma$ on an algebra $\mathfrak{A}$, we have that the following holds for any finite collection of events $X_{i}$ $(i\leq n)$ in $\mathfrak{A}$: 
\begin{equation}\label{wOrsem}\tag{\textsf{wO}}
\text{If $X_{i}\setminus X_{j} \subseteq \sigma(X_{i})$ for all $i\neq j$, then $\sigma(\bigcup_{i\leq n}X_{i})\subseteq \bigcup_{i\leq n} \sigma(X_{i})$.} 
\end{equation}
\end{observation}
\begin{proof}
Let $\mu$ be a probability measure representing the selection function $\sigma$, so that $\sigma(X)=\tau_{t}(\mu_{X})$ for all $X\in\mathfrak{A}$. Assume $X_{i}\setminus X_{j} \subseteq \sigma(X_{i})$ for all $i,j \leq n$ with $i \neq j$. It is enough to show that  $\bigcup_{i\leq n} \sigma(X_{i})$ is stable with respect to the measure $\mu(\cdot\,|\, \bigcup_{i\leq n}X_{i} )$: this suffices, since $\sigma (\bigcup_{i\leq n}X_{i})$ is the \emph{strongest} stable set with respect to $\mu(\cdot\,|\, \bigcup_{i\leq n}X_{i} )$. Let $\omega\in \bigcup_{i\leq n} \sigma(X_{i})$,\footnote{Here again, we take a singleton to represent an atom for simplicity, but this is immaterial: the argument applies for any choice of $\mathfrak{A}$-atom in  $\bigcup_{i\leq n} \sigma(X_{i})$ instead of $\{\omega\}$.} and consider the relative complement  $(\bigcup_{i\leq n}X_{i} )\setminus  \bigcup_{i\leq n} \sigma(X_{i})$. Since $X_{i}\setminus X_{j}\subseteq \sigma(X_{i})$ for all distinct $i,j$, this means that $\bigcup_{i\neq j}(X_{i}\setminus X_{j}) \subseteq \bigcup_{i\leq n} \sigma(X_{i})$. So we get  $(\bigcup_{i\leq n}X_{i} )\setminus  \bigcup_{i\leq n} \sigma(X_{i}) \subseteq  (\bigcup_{i\leq n}X_{i}) \setminus  \bigcup_{i\neq j} (X_{i}\setminus X_{j}) = \bigcap_{i\leq n} X_{i}$. 

This establishes that $(\bigcup_{i\leq n}X_{i} )\setminus  \bigcup_{i\leq n} \sigma(X_{i}) \subseteq \bigcap_{i\leq n}X_{i}$; we can then also write 
$$
(\bigcup_{i\leq n}X_{i}) \setminus  \bigcup_{i\leq n} \sigma(X_{i}) \subseteq X_{j}\setminus \sigma(X_{j})
$$
for all $j\leq n$. But evidently, since $\omega\in \bigcup_{i\leq n} \sigma(X_{i})$ we have $\omega \in \sigma(X_{j})$ for some $j$: this means that $\mu(\omega) > \frac{t}{1-t}\cdot \mu (X_{j} \setminus \sigma(X_{j})$. In particular, 
$$
\mu(\omega) >\frac{t}{1-t}\cdot  \mu\Big(\big(\bigcup_{i\leq n}X_{i}\big) \setminus  \bigcup_{i\leq n} \sigma(X_{i}) \Big)
$$ 
where $\omega$ was arbitrary in $ \bigcup_{i\leq n} \sigma(X_{i})$. This establishes that $\bigcup_{i\leq n} \sigma(X_{i})$ is $t$-stable with respect to the measure $\mu(\cdot\,|\, \bigcup_{i\leq n}X_{i} )$.
\end{proof}

Since the condition $\sigma(\bigcup_{i\leq n}X_{i})\subseteq \bigcup_{i\leq n} \sigma(X_{i})$ captures the (\textsf{Or}) rule, we can see (\ref{wOrsem}) as a substantially weaker form of (\textsf{Or})-style reasoning: it specifies that the (\textsf{Or}) rule can be applied to a set of antecedents $X_{1},...,X_{n}$ provided that they satisfy the side condition $X_{i}\setminus X_{j} \subseteq \sigma(X_{i})$ for all $i\neq j$. This weak (\textsf{wO}) rule can be written in semantic form as follows (here we write its two-premise version): 
$$
\infer[(\mathsf{wO}_{2})] {X_{1}\cup X_{2}\nc A}{X_{i}\setminus X_{j} \subseteq \sigma(X_{i}) \text{ for }i\neq j& X_{1}\nc A & X_{2}\nc A}
$$
The side constraints $X_{i}\setminus X_{j} \subseteq \sigma(X_{i})$ give conditions under which \textsf{Or}-type inferences are valid: all of the $X_{i}\setminus X_{j}$ states must be \emph{typical} given $X_{i}$, in the sense of being in the selected subset. Alternatively, we can see this as a constraint on \emph{atypical} states: namely, all of the \emph{atypical} $X_{i}$ states\textemdash those in $X_{i}\setminus \sigma(X_{i})$\textemdash must be in $X_{j}$. Let $\textsf{atyp}(A):=A\setminus \sigma(A)$ denote the `\emph{atypical}' $A$-states: the possibilities consistent with $A$ that the agent is nonetheless prepared to rule out conditional on learning $A$. We can capture the (two-premise) rule $(\mathsf{wO}_{2})$ as follows: 
$$
\infer[\text{provided } \textsf{atyp}(A)\vdash B \text{ and } \textsf{atyp}(B)\vdash A ] {A\vee B \nc C}{ A\nc C & B\nc C}
$$
That is: if atypical $A$ outcomes are $B$ outcomes, and atypical $B$ outcomes are $A$ outcomes, then case reasoning applies for premises $A$ and $B$.  Here is a very simple example to illustrate this type of inference. 
\begin{ex}
A species is sampled from a population according to the following probabilities:
    \begin{center}
\normalfont
\begin{tabular}{|c|c|c|c|c|c|}
\hline
 \emph{Species} &  Sparrow  & Turtle & Pit Viper & King Cobra & Platypus  \\
 \hline
 \emph{Probability} & 40\% & 20\% & 20\% & 16\% & 4\% \\
 \hline
\end{tabular}
\end{center}
Consider the following propositions one might learn about the sample:
\begin{align*}
\mathsf{Oviparous}&:=\{\text{Sparrow, Turtle, King Cobra, Platypus} \} \\
\mathsf{Venomous}&:=\{\text{Pit Viper, King Cobra, Platypus} \} \\
\mathsf{Furry}&:=\{\text{Platypus} \}
\end{align*}
Given the probabilities above, for a threshold of $t=1/2$ we have
\begin{align*}
\sigma(\mathsf{Oviparous})&=\{\text{Sparrow, Turtle, King Cobra} \}   \\
\sigma(\mathsf{Venomous})&=\{\text{Pit Viper, King Cobra} \}
\end{align*}
Now we have $\mathsf{atyp}(\mathsf{Oviparous})=\mathsf{atyp}(\mathsf{Oviparous})=\{\text{Platypus}\}$. Note then that $\mathsf{atyp}(\mathsf{Oviparous}) \vdash \mathsf{Venomous}$ and  $\mathsf{atyp}(\mathsf{Venomous}) \vdash \mathsf{Oviparous}$, and so the following inference is validated
$$
\infer[] {\mathsf{Oviparous}\vee \mathsf{Venomous} \nc \neg \mathsf{Furry}}{ \mathsf{Oviparous}\nc \neg\mathsf{Furry} & \mathsf{Venomous} \nc \neg \mathsf{Furry}}
$$
\end{ex}
In addition to outlining a very modest extent to which probabilistic stability obeys a version of case reasoning (or the sure thing principle), this rule can be used to obtain an alternative characterisation of probabilistically stable revision plans: it is enough to replace the axiom $\mathsf{S4_{n}}$ by the property \eqref{wOrsem}. This is because the property already entails $\mathsf{S4_{n}}$. 

\begin{observation}\label{weakOr2}
Suppose a selection function $\sigma$ satisfies the property \eqref{wOrsem}. Then it also satisfies $\mathsf{S4_{n}}$. 
\end{observation}
\begin{proof}
Suppose $\sigma$ has property \eqref{wOrsem}. Suppose we have sets $A$, $X_{i}$ $(i\leq n)$, such that $\forall i\leq n$, $\sigma(A\cup X_{i})=X_{i}$. Writing $D_{i}:= A\cup X_{i}$, we have, for each $i\neq j$, $D_{i}\setminus D_{j} = X_{i}\setminus (A\cup X_{j})$. Now $\sigma(D_{i}) = \sigma(A\cup X_{i})=X_{i}$, and so we can write  $D_{i}\setminus D_{j} \subseteq \sigma(D_{i})$. By  \eqref{wOrsem}, we can conclude $\sigma(\bigcup_{i\leq n}D_{i})\subseteq \bigcup_{i\leq n} \sigma(D_{i})$, which is equivalent to $\sigma(\bigcup_{i\leq n}A\cup X_{i})\subseteq \bigcup_{i\leq n} X_{i}$. We have shown ($\mathsf{S4_{n}}$). 
\end{proof}

This observation entails the following reformulation of our representation theorem (Theorem \ref{REPSEL}):

\begin{prop}
Let $(\Omega, \mathfrak{A}, \sigma)$  a selection structure. The following conditions are necessary and sufficient for $\sigma$ to be a strongest-stable-set operator for some underlying (regular) probability measure on $\mathfrak{A}$ and threshold $t=1/2$:
\begin{itemize}
\item [\emph{($\textsf{S1}$)}]$\sigma(X)=\emptyset$ only if $X=\emptyset$
\item [\emph{($\textsf{S2}$)}]$\sigma(X)\subseteq X$
\item [\emph{($\textsf{S3}$)}]If $\sigma(A)\cap B\neq\emptyset$, then $\sigma(A\cap B)\subseteq \sigma(A)\cap B$
\item  [\emph{($\textsf{wO}$)}]\text{If $X_{i}\setminus X_{j} \subseteq \sigma(X_{i})$ for all $i\neq j$, then $\sigma(\bigcup_{i\leq n}X_{i})\subseteq \bigcup_{i\leq n} \sigma(X_{i})$.} 
\item [\emph{\textsf{(Scott)}}] If $(A_{i})_{i\leq n}\equiv_{0}(B_{i})_{i\leq n}$ and $\forall i\leq n$, $A_{i}\succcurlyeq^{\ast}_{\sigma}B_{i}$, then $\forall i\leq n$, $A_{i} \preccurlyeq^{\ast}_{\sigma} B_{i}$. 
\end{itemize}
\end{prop}

This reformulation of the representation theorem is slightly more perspicuous: probabilistically stable revision operations are characterised by reflexivity, rational monotonicity, a weaker form of the (\textsf{Or}) rule, and the (\textsf{Scott})-type axiom for representability (which guarantees that the `dominance' ordering between atomic and other events implied by $\sigma$ can be captured quantitatively). 

\subsection{Representability at different thresholds}\label{SectionRepAtDiffThresh}

 Our representation theorem characterises exactly the probabilistically stable revision plans for threshold $t=1/2$: in other words, the 1/2-representable selections functions (those representable as strongest-stable stable set operators for a threshold of $1/2$). As we mentioned above, this is perhaps the most canonical choice for threshold, as it appears to involve little more than `qualitative' reasoning: we saw that $1/2$-representable operators can be fully characterised as involving no more than comparative judgments of the form ``the state $\{\omega\}$ is (or is not, as the case may be) more likely than event $B$''.

Nonetheless, we would like to be able to characterise the probabilistically stable operators for different choices of $t$, if only to gauge the extent to which a different stability threshold translates into different reasoning patterns for belief revision. How sensitive is the representation theorem to the choice of threshold? That is, how sensitive is the class of probabilistically stable revision operators to the choice of threshold? Do our axioms suffice to get a representation result for other thresholds? The (not entirely unexpected) answer is a resounding no. 
 
\begin{prop}\label{IncreasingThresholds}
   For any $t,q\in [1/2,1)$ with $t<q$, there exists a selection function that is $t$-representable but not $q$-representable. 
\end{prop}
\begin{proof}
For each choice of $n,m\in\mathbb{N}$, take $\Omega=\{b_{1},...,b_{m},c_{1},...,c_{n}\}$, and consider the following constraints: 
    \begin{align}
         \sigma (\{b_{i}, b_{i+1}\}) &= \{ b_{i}\} \text{ for all }i\in\{1,...,m\}  \label{sigmaconstraintA} \\
         \sigma (\{b_{m},c_{i}\}) &= \{b_{m}\} \text{ for all }i\in\{1,...,n\} \label{sigmaconstraintB}  \\
           \sigma (\{b_{1},c_{1},...,c_{n}\}) &\neq \{b_{1}\}  \label{sigmaconstraintC}  
    \end{align}
These constraints impose that the induced order satisfy the following: $ b_{1} \succ b_{2} \succ...\succ b_{m} \succ c_{i}$ and $b_1\preceq \{c_{1},...,c_{n}\}$.\footnote{For notational convenience we omit the subscript $\sigma$ from the relations $\succ_{\sigma}$ and $\succeq_{\sigma}$.} For a threshold $t$, let  $k_{t}:=t/(1-t)$. We begin with the following: 

\begin{observation}\label{OBSTHRESHOLDREP}
    There exists a probabilistically stable revision operator $\sigma_{\mu,t}$ meeting the above constraints for threshold $t$ if and only if $k_{t}^{m}- nk_{t}< 0$. 
\end{observation}

\noindent To see this, suppose $\mu$ is a measure that represents such a selection function for threshold $t$. Let $C=\max\{\mu(c_{i})\,|\,1\leq  i\leq n\}$. Then the last constraint (\ref{sigmaconstraintC}) above entails 
  \begin{equation}\label{clast}
       \mu(b_1)\leq k_{t} n  C 
  \end{equation}
The first two constraints \eqref{sigmaconstraintA} and \eqref{sigmaconstraintB} above entail $\mu(b_1)> k^{m}_{t} C$. Evidently these constraints are unsatisfiable if $k^{m}_{t} \geq nk_{t} $. In the other direction, assume $k^{m}_{t} <  nk_{t} $: then there exists a measure generating a selection function with threshold $t$ which satisfies the above constraints. For example, consider the following (non-normalised) weights, for some real $\epsilon$:
      \begin{align*}
          w (b_{i}) & = k^{m-i+1}_{t} + \epsilon \cdot\sum^{m-i}_{j=0} k^{j}_{t}  \text{ for all }i\in\{1,...,m\}\\
              w (c_{i}) &= 1  \text{ for all }i,j\in\{1,...,n\}
    \end{align*}
These weights were simply chosen so that (I) all $c_{i}$ have an equal weight of 1; (II) $w(b_m) = k_{t}w(c_i)+\epsilon$, and (III) for every $i<m$, we have $b_{i} = k_{t}\cdot w(b_{i+1})+\epsilon$. The first two conditions \eqref{sigmaconstraintA} and \eqref{sigmaconstraintB} on $\sigma$ above require $\epsilon>0$. 
Given these weights, the condition $\sigma (\{b_{1},c_{1},...,c_{n}\}) \neq \{b_{1}\}$ imposes the constraint $ w(b_1) \leq n  k_{t}$. 
That is, we want  
\begin{align*}
  k^{m}_{t}+ \epsilon \cdot\sum^{m-1}_{j=0} k^{j}_{t} &\leq n k_{t}, \text{ or equivalently} \\  
\epsilon &\leq  \frac{nk_{t} - k^{m}_{t}}{\sum^{m-1}_{j=0} k^{j}_{t}}
\end{align*}

\noindent Now it is enough to show that there exist $\epsilon>0$ that satisfies this. This is immediately guaranteed by our assumption that $k^{m}_{t}<nk_{t}$: we know that both the numerator $nk_{t} - k^{m}_{t}$ and denominator $\sum^{m-1}_{j=0} k^{j}_{t}$ on the right-hand side are positive (recall $k_{t}\geq 1$), thus so is $R = (nk_{t} - k^{m}_{t})/\sum^{m-1}_{j=0} k^{j}_{t}$. So any positive choice of $\epsilon \in (0, R]$ will do. For any such choice of $\epsilon$, the weights $w$ give rise to a non-normalized measure on $\Omega = \{b_{1},...,b_{m-1}, c_{1},...,c_{n}\}$ that satisfies all the constraints \eqref{sigmaconstraintA}, \eqref{sigmaconstraintB} \eqref{sigmaconstraintC} for threshold $t$. The normalised measure $\mu(x) = w(x) /\sum_{\omega\in\Omega} w(\omega)$ then automatically satisfies the same constraints, as normalising preserves all ratios.

Next, we need to ensure that, for any two thresholds $t<q$, there exists some value of $m,n$ such that $k_{t}^{m}- nk_{t}< 0$ but $k_{q}^{m}- nk_{q}\geq 0$. This amounts to showing that for any $t<q$, there exist $m,n\in\mathbb{N}$ such  the solution set $\{x\in\mathbb{R}\,|\, x^{m}-nx< 0\}$ contains $k_t$ but not $k_q$.

\begin{lem}
   The set of solutions to $x^{m}-nx=0$ with $m,n\in\mathbb{N}$ is dense in $\mathbb{R}\cap [1,\infty)$.
\end{lem}
\begin{proof}
    Let $[a,b)\subset\mathbb{R}\cap [1,\infty)$. Positive solutions to $x^{m}-nx=0$ are of the form $x=n^{1/(m-1)}$. These are evidently dense: given $1\leq a<b$, we have $\lim_{m\to\infty} (b^{m}-a^{m})=\infty$, so there exists some $m\geq 2\in\mathbb{N}$ such that $b^{m-1} - a^{m-1} > 1 $. But then there exists some $n\in\mathbb{N}$ such that $a^{m-1} < n < b^{m-1}$, hence $a < n^{1/(m-1)} < b$. Since $1\leq a^{m-1}$, note that we can always take $n>1$.
\end{proof}
\noindent The desired result follows. Let $t<q$: we then have $k_t<k_q$. By the Lemma, there exists $m,n\in\mathbb{N}$ and some $k_t<x_{0}<k_q$ such that $x_{0}^{m}-nx_{0}=0$, i.e. $x_{0}=n^{1/(m-1)}$, and we can take $m\geq 2, n>1$. Since $t,q\geq 1/2$, we have $k_t, k_q\geq 1$.  Since the polynomial $x^{m}-nx$ is negative for $1< x < x_{0}$ and positive for $x> x_{0}$ this means that $k_{t}^m - nk_{t}<0$ but $k_{q}^m - nk_{q}> 0$. By Observation \ref{OBSTHRESHOLDREP}, there exists a $t$-representable selection function meeting the constraints \eqref{sigmaconstraintA}, \eqref{sigmaconstraintB}, \eqref{sigmaconstraintC} above, and that selection function is not $q$-representable. 
\end{proof}

The upshot is that distinct thresholds generate distinct classes of probabilistically stable revision operators. Thus any representation theorem for any particular threshold $t$ will require a distinct axiomatisation. On the other hand, it is very easy, but informative, to observe the following: 

\begin{observation}
     For any $t>\frac{1}{2}$, there exists a $t$-representable selection function that is not $\frac{1}{2}$-representable. 
\end{observation}

\begin{proof}
    Take any $t\in (1/2, 1)$. Then $k= \frac{t}{1-t}>1$. Let $\Omega=\{a,b,c\}$ and take the weights $(m(a), m(b), m(c)) = (1,k,k^{2})$. The corresponding normalised measure $\mu(\omega):=\frac{m(\omega)}{k^{2}+k+1}$ gives rise to a selection function which satisfies $\sigma (\{b,c \}) \neq \{ c\}$, $\sigma (\{a,b \}) \neq \{ b\} $ and $ \sigma (\{a,c \}) = \{ c\}  $.
    Note that this simply amounts to a failure of transitivity of the $\preceq_{\sigma}$ relation: these constraints tell us that we have $ c \preceq_{\sigma} b $ and  $ b \preceq_{\sigma} a$, but $c\not\preceq_{\sigma} a$ since $c\succ_{\sigma} a$. These constraints are evidently not satisfiable for $t=1/2$. 
\end{proof}

This observation also illustrates that the transitivity of the $\preceq_{\sigma}$ relation, entailed by the \textsf{(Scott)}  axiom, is sound only for $t=1/2$. So the \textsf{(Scott)} axiom is not sound for any threshold $t>1/2$. Thus, if we hope to have a charaterisation of probabilistically stable revision for other thresholds, we will have to abandon the \textsf{(Scott)} axiom. Moreover, our observation above means that, whenever $q\neq t$, the axioms for $q$-representable operators will have to differ from the axioms for $t$-representable operators. In order to have a satisfactory characterization of probabilistically stable revision, perhaps the best we can hope for is a \emph{parametric} axiomatization specifying, given a threshold value $t$, a scheme of axioms $\Gamma_t$ that captures the $t$-representable selection functions. We turn to this question next. 

\subsection{General representation theorem}

Here we provide a general representation theorem which characterises, for every rational threshold $t\geq 1/2$, exactly the $t$-representable selection functions, or equivalently, the probabilistically stable revision operators for threshold $t$. We will then establish that there is a sense in which \emph{every} probabilistically stable belief revision operator is captured by our general representation theorem, because every representable selection function is representable by a rational threshold (Proposition \ref{QREP}). 

\begin{mydef}[$t$-balanced sequences]\label{TBalanced}
Let $t=p/(p+q)$ with $p,q \in \mathbb{N}$. Let $(\Omega, \mathcal{P}(\Omega))$ a finite set algebra, $\sigma$ a selection function on $\mathcal{P}(\Omega)$, and $(A_{i})_{i\leq n}$ and $(B_{i})_{i\leq n}$ two sequences of events from $\mathcal{P}(\Omega)$. Define the following $\{0,1\}$-valued functions:
\begin{align*}
 \mathds{1}^{\succ}_{A_{i}}(\omega) &:= 1 \,\,\,\text{ iff }\,\,\,[\omega \in A_i \text{ and } A_i \succ_{\sigma} B_{i} ] \\
 \mathds{1}^{\succeq}_{A_{i}}(\omega) &:= 1 \,\,\,\text{ iff }\,\,\,[\omega \in A_i \text{ and } A_i \not\succ_{\sigma} B_{i}\text{ and } A_i \succeq_{\sigma} B_{i}] \\
  \mathds{1}^{\succ}_{B_{i}}(\omega) &:= 1 \,\,\,\text{ iff }\,\,\,[\omega \in B_i \text{ and } A_i \succ_{\sigma} B_{i}] \\
   \mathds{1}^{\succeq}_{B_{i}}(\omega) &:= 1 \,\,\,\text{ iff }\,\,\,[\omega \in B_i \text{ and } A_i \not\succ_{\sigma} B_{i} \text{ and }A_i \succeq_{\sigma} B_{i}] 
\end{align*}
The sequences $(A_{i})_{i\leq n}$ and $(B_{i})_{i\leq n}$ are \emph{$t$-balanced} if and only if for every $\omega\in\Omega$,
$$q\cdot\sum_{i\leq n} \mathds{1}^{\succ}_{A_{i}}(\omega) + p\cdot\sum_{i\leq n} \mathds{1}^{\succeq}_{A_{i}} (\omega)= p\cdot \sum_{i\leq n} \mathds{1}^{\succ}_{B_{i}} (\omega)+ q\cdot\sum_{i\leq n} \mathds{1}^{\succeq}_{B_{i}}(\omega) .$$
We then write $(A_{i})_{i\leq n} \equiv^{t}_{0} (B_{i})_{i\leq n}$.
\end{mydef} 

The notion of $t$-balancedness, although notationally cumbersome, is a straightfoward generalization of the usual notion of a balanced sequence. A selection function $\sigma$ induces \emph{strict} inequalities of the form $A\succ_{\sigma}B$ as well as \emph{weak} inequalities of the form $A\succeq_{\sigma}B$ with $A\not\succ_{\sigma}B$. Suppose we have two sequences of sets $(A_{i})_{\leq n}$ and $(B_{i})_{i\leq n}$, and we fix a state $\omega\in\Omega$. Rather than simply counting the number of occurrences of $\omega$ among the sets in each family, the count is adjusted depending on which side of an inequality of the form $A\succcurlyeq^{\ast}_{\sigma} B$ the state $\omega$ occurs in. Each occurrence of $\omega$ is counted $q$ times whenever it occurs either on the dominating side of a strict inequality or on the dominated side of a weak inequality. Each occurrence of $\omega$ is counted $p$ times whenever it occurs either on the dominated side of a strict inequality or on the dominating side of a weak inequality. A balanced set of inequalities is one where, for every state, the total adjusted count of occurrences on the dominating side equals the total adjusted count on the dominated side. Here is another way to formulate the $t$-balancedness criterion. Given a selection function $\sigma$ and two sequences of events $A=(A_{i})_{i\leq n}$ and $B=(B_{i})_{i\leq n}$, for each state $\omega$, define
\begin{align*}
 [\succ, A](\omega) &:= |\{i\,|\,  \omega \in A_i \text{ and } A_i \succ_{\sigma} B_{i} \} |\\
  [\succeq, A](\omega)&:= |\{i\,|\,  \omega \in A_i \text{ and } A_i \not\succ_{\sigma} B_{i} \text{ and }A_i \succeq_{\sigma} B_{i} \} |\\
   [\succ, B](\omega)&:= |\{i\,|\,  \omega \in B_i \text{ and } A_i \succ_{\sigma}  B_{i}\} |\\
    [\succeq, B](\omega)&:= |\{i\,|\,  \omega \in B_i \text{ and } A_i \not\succ_{\sigma} B_{i} \text{ and } A_i \succeq_{\sigma}  B_{i} \} |
\end{align*} 
Each of these counts the number of occurrences of a state on the dominating (respectively, dominated) side of a strict (respectively, weak) $\sigma$-induced inequality between $A_i$ and $B_i$. Let the total adjusted \emph{left} count of $\omega$ be $L(\omega)= q \cdot  [\succ, A] + p \cdot [\succeq, A]$, and its total adjusted  \emph{right} count $R(\omega)=  p \cdot [\succ, B]+ q \cdot [\succeq, B]$. Then $t$-balancedness requires that, for every state $\omega$,
\begin{align*}
L(\omega)&= R(\omega),\\
\text{which means}\qquad q \cdot  [\succ, A] + p \cdot [\succeq, A] &=  p \cdot [\succ, B]+ q \cdot [\succeq, B]
\end{align*}

Two important remarks are in order about the notion $t$-balancedness and its relation to the ordinary notion of balancedness. Firstly, the notion of $t$-balancedness is defined in terms of a particular ratio representation of $t$ as $p/(p+q)$. We can conventionally fix $p,q$ so that $p/(p+q)$ is the irreducible fraction representation of $t$, but this is immaterial: whether two sequences are $t$-balanced does not depend on the particular choice of $p,q$, as long as they yield the same ratio.\footnote{Here is the argument in a nutshell. If $t=p/(p+q)$ is the irreducible fraction representation of $t$, then any other ratio representation of the form $t=a/(a+b)$ has the feature that $a=Np$ and $b=Nq$ for some $N\in\mathbb{N}$. The $t$-balancedness for $a,b$ means that, for every $\omega$,  we have $b \cdot  [\succ, A] + a \cdot [\succeq, A] =  a \cdot [\succ, B]+ b \cdot [\succeq, B]$. Since $(a,b)=(Np,Nq)$, this just means
$$
 Nq \cdot  [\succ, A] + Np \cdot [\succeq, A] =  Np \cdot [\succ, B]+ Nq \cdot [\succeq, B],
 $$
which holds if and only if
$$
q \cdot  [\succ, A] + p \cdot [\succeq, A] =  p \cdot [\succ, B]+ q \cdot [\succeq, B].
$$
So this is equivalent to the condition of $t$-balancedness for $p,q$.} 

Thirdly, note that the \textsf{(Scott)} axiom can be formulated in terms of $\frac{1}{2}$-balancedness. We can formulate \textsf{(Scott)} as follows: 
\begin{center}
\textsf{(Scott)} \qquad If $(A_{i})_{i\leq n}\equiv_{0}(B_{i})_{i\leq n}$ and $A_{i}\succcurlyeq^{\ast}_{\sigma}B_{i}$ for all $i<n$, then $A_{n} \not\succ_{\sigma} B_{n}$.
\end{center}
Now, in the antecedent, replace the statement $(A_{i})_{i\leq n}\equiv_{0}(B_{i})_{i\leq n}$ with the statement $(A_{i})_{i\leq n}\equiv^{t}_{0}(B_{i})_{i\leq n}$ for $t=1/2$. 
\begin{center}
\textsf{(Scott[1/2])} \qquad If $(A_{i})_{i\leq n}\equiv^{1/2}_{0}(B_{i})_{i\leq n}$ and $A_{i}\succcurlyeq^{\ast}_{\sigma}B_{i}$ for all $i<n$, then $A_{n} \not\succ_{\sigma} B_{n}$.
\end{center}
These two are equivalent. This is not to say that the ordinary notion of a balanced sequence here is equivalent to $\frac{1}{2}$-balancedness: this is \emph{not} the case, due to the fact that $\frac{1}{2}$-balancedness does not count any occurrences of a state within any $A_i$, $B_i$ for which $A_i \not\succcurlyeq^{\ast}_{\sigma}B_i$.\footnote{The adjusted count does not count any occurrences of a state within an incomparable pair. Suppose $\{a\} \succ \{b\}$, $\{d\} \succ \{c\}$, and the sets $\{b,c\}$ and $\{a,d\}$ are incomparable (that is, $\succcurlyeq^{\ast}_{\sigma}$ doesn't hold in either direction). Then $\big(\{b,c\},\{a\},\{d\}\big) \equiv_0 \big(\{a,d\}, \{b\}, \{c\}\big)$ but $\big(\{b,c\},\{a\},\{d\}\big) \ \not\equiv^{1/2}_0 \big(\{a,d\}, \{b\}, \{c\}\big)$: note that the adjusted left-count of $b$ is $L(b)=0$ while $R(b)=1$. The other implication also fails for a similar reason: $\frac{1}{2}$-balancedness does not entail balancedness in the usual sense. For example, if $\{a\}\succeq \{a\}$ while $\{b,c\}$ and $\{c,d\}$ are incomparable, then $\big(\{a\},\{b,c\}\big) \equiv^{1/2}_0 \big(\{a\}, \{c,d\}\big)$ but $\big(\{a\},\{b,c\}\big) \not\equiv_0 \big(\{a\}, \{c,d\}\big)$.} However, the two notions of balancedness $\equiv_0$ and $\equiv^{1/2}_0$ \emph{are} equivalent for sequences $A$ and $B$ where $A_i \succcurlyeq^{\ast}_{\sigma}B_i$ for each $i$. To see this, first note that, due to $A_i \succcurlyeq^{\ast}_{\sigma}B_i$, every occurrence of a state in any $A_i$ or $B_i$ gets counted towards the adjusted count $L(\omega), R(\omega)$, since each occurrence falls under one of the four cases ($[\succ, A]$, $[\succeq, A]$, $[\succ, B]$ or $[\succeq, B]$); second, the property of $\frac{1}{2}$-balancedness corresponds to the case where $t=p/(p+q)$ with $p=q=1$, so in the adjusted count each occurrence of a state is counted \emph{equally}, regardless on which side of an inequality it occurs.\footnote{We could of course easily extend the notion of $t$-balancedness so as to make $\equiv^{1/2}_{0}$ equivalent to plain balancedness for sequences $A$, $B$ in which every pair $(A_i, B_i)$ is comparable\textemdash that is, when either $B_i\succcurlyeq^{*}_{\sigma}A_i$ or $A_i\succcurlyeq^{*}_{\sigma}B_i$ holds for all $i$. But this additional complication is not needed: even the more restricted notion suffices to formulate a version of the \textsf{(Scott)} axiom that suffices for a representation theorem.} It immediately follows that any counterexample to \textsf{(Scott)} will in fact constitute a counterexample to \textsf{(Scott[1/2])}, and vice-versa, since any counterexample will yield a sequence of events where $A_i \succcurlyeq^{\ast}_{\sigma}B_i$.

This observation suggests an amendment to the \textsf{(Scott)} axiom tailored for each rational threshold $t$. We simply adjust our earlier axiom to apply to $t$-balanced sequences. This leads to our most general result:

\begin{theorem}\label{GENREPSEL}
\textbf{Representation theorem for rational thresholds}\\
Let $t\in[1/2,1)\cap\mathbb{Q}$. A selection structure $(\Omega, \mathcal{P}(\Omega), \sigma)$ is a strongest-stable-set operator for threshold $t$ if and only if it satisfies the following:
\begin{itemize}
\item [($\textsf{S1}$)]$\sigma(X)=\emptyset$ only if $X=\emptyset$
\item [($\textsf{S2}$)]$\sigma(X)\subseteq X$
\item [($\textsf{S3}$)]If $\sigma(A)\cap B\neq\emptyset$, then $\sigma(A\cap B)\subseteq \sigma(A)\cap B$
\item [($\textsf{S4}_{n}$)]For any $n$: if $\sigma(A\cup X_{i})=X_{i}$ for all $i\leq n$, then $\sigma(A\cup\bigcup_{i\leq n}X_{i})\subseteq \bigcup_{i\leq n}X_{i}$
\item [\textsf{(Scott$[t]$)}] If $(A_{i})_{i\leq n}\equiv^{t}_{0}(B_{i})_{i\leq n}$ and $A_{i}\succcurlyeq^{\ast}_{\sigma}B_{i}$ for all $i<n$, then $A_{n} \not\succ_{\sigma} B_{n}$.
\end{itemize}
\end{theorem}
\begin{proof}
The soundness of axioms $\textsf{S1-S4}$ was already verified independently of the value of the threshold. To show the necessity of these axioms, it only remains to show the soundness of \textsf{(Scott$[t]$)} for any rational threshold $t\geq1/2$. We can write any such rational threshold as $t=p/(p+q)$ for positive integers $p\geq q$. 

Let $\sigma$ be the strongest stable set operator for the threshold $t$, generated by some measure $\mu$ (i.e., $\sigma=\sigma_{\mu,t}$). Suppose, towards a contradiction, that we have sequences $(A_{i})_{i\leq n}\equiv^{t}_{0}(B_{i})_{i\leq n}$ with $A_{i}\succcurlyeq^{\ast}_{\sigma}B_{i}$ for all $i< n$, and $A_n \succ_{\sigma} B_{n}$. Given that $t=p/(p+q)$, we have $\omega \succ_{\sigma} X$ if and only if $\mu(\omega) > \frac{p}{q} \mu(X)$ or, equivalently, $q\cdot \mu(\omega) > p \cdot\sum_{\omega_{i}\in X} \mu(\omega_{i})$. Then we know that the measure $\mu$ satisfies the following system of $n$ inequalities, corresponding to each statement $A_{i}\succcurlyeq^{\ast}_{\sigma} B_{i}$:
\begin{align*}
&\text{(I) For each }A_{i} \succ_{\sigma} B_{i}\, (\text{where }A_{i}=\{a\}), 
&q\cdot \mu(a) > p\cdot \sum_{\omega\in B_{i}} \mu(\omega) \\
&\text{(II) For each }A_{i} \succeq_{\sigma} B_{i} \text{ with } A_{i} \not\succ_{\sigma} B_{i} \,(\text{where }B_{i}=\{b\}), 
 &p \cdot\sum_{\omega\in A_{i}} \mu(\omega)\geq q\cdot \mu(b)
\end{align*}
Since $A_{n}\succ_{\sigma}B_{n}$, we know that at least one inequality of the form (I) obtains. This means that the sum $\sum_{i\leq n} L_{n}$ of the terms on the left-hand-side, across all these inequalities, is greater than the sum $\sum_{i\leq n} R_{n}$ of the terms of the right hand side: 
\begin{equation}\label{LEFTRIGHTSUM}
    \sum_{i\leq n} L_{n} > \sum_{i\leq n} R_{n}
\end{equation}
Writing out the sum of terms of the left-hand-side $\sum_{i\leq n} L_{n}$ as a sum of terms $\mu(\omega)$, each $\mu(\omega)$ occurs in that sum exactly $L(\omega)$ times (using the $L(\omega)$ notation introduced just below Definition \ref{TBalanced}). So $\sum_{i\leq n} L_{n}=\sum_{\omega\in\Omega} \mu(\omega)L(\omega)$ where $L(\omega)=q\cdot\sum_{i\leq n} \mathds{1}^{\succ}_{A_{i}}(\omega) + p\cdot\sum_{i\leq n} \mathds{1}^{\succeq}_{A_{i}} (\omega)$. By the same reasoning, the sum of terms on the right hand side is $\sum_{\omega\in\Omega} \mu(\omega)R(\omega)$ where $R(\omega)=p\cdot\sum_{i\leq n} \mathds{1}^{\succ}_{B_{i}}(\omega) + q\cdot\sum_{i\leq n} \mathds{1}^{\succeq}_{B_{i}} (\omega)$. Finally, the sequence being $t$-balanced tells us that $L(\omega)=R(\omega)$ for every $\omega$, hence $\sum_{\omega\in\Omega}\mu(\omega)L(\omega) = \sum_{\omega\in\Omega}\mu(\omega)R(\omega)$, which means $\sum_{i\leq n} L_{n} = \sum_{i\leq n} R_{n}$, contradicting \eqref{LEFTRIGHTSUM}. Thus we cannot have $A_{n}\succ_{\sigma} B_{n}$. 
 
For sufficiency, we proceed similarly to the earlier proof for $t=1/2$. Let $\Omega=\{\omega_{1},\dots,\omega_{n}\}$. Take a selection function $\sigma$ on $\mathcal{P}(\Omega)$ satisfying all the above axioms. List all weak inequalities and all strict inequalities generated by $\sigma$: we have $k$ weak inequalities $\{A_{i} \succeq_{\sigma} B_{i}\text{ and }A_i\not\succ_{\sigma}B_i \}_{i\leq k}$ and $m$ strict inequalities $\{A_{i} \succ_{\sigma} B_{i}\}_{k<i\leq k+m}$. The selection function $\sigma$ is representable as a strongest stable set operator for threshold  $t=p/(p+q)$ if and only if there exists a measure $\mu$ such that, for any $\omega\in\Omega$ and $X\subseteq \Omega$, we have $\mu(\omega)> \frac{p}{q}\mu(X)$ if and only if $\omega\succ_{\sigma}X$. This amounts to the solvability of the system 
\begin{equation}\label{MATRIX2}
\begin{split}
\mathbf{M}_{\succeq}\cdot\vec{\mu} \geq \vec{0} \\
\mathbf{M}_{\succ}\cdot\vec{\mu} > \vec{0}
\end{split}
\end{equation}
Where
$$
\mathbf{M}_{\succeq}=
\begin{pmatrix}
(p\ind{A_{1}}- q\ind{B_{1}})^{T}\\
\vdots  \\
(p\ind{A_{k}}- q\ind{B_{k}} )^{T} 
\end{pmatrix}
\,\,\,\,\text{ and }\,\,\,\,
\mathbf{M}_{\succ}=
\begin{pmatrix}
(q\ind{A_{k+1}}- p\ind{B_{k+1}})^{T}\\
\vdots  \\
(q\ind{A_{k+m}}- p\ind{B_{k+m}} )^{T} 
\end{pmatrix}
$$
In this system, each row of the matrix $\mathbf{M}_{\succeq}$ captures (the constraint imposed by) an inequality of the form $[A_{i}\succeq_{\sigma} B_{i}$ and $A_i\not\succ_{\sigma}B_i]$: i.e., the $i$-th row of $\mathbf{M}_{\succeq}\cdot\vec{\mu}$ captures the inequality
\begin{align*}(p\ind{A_{i}}- q\ind{B_{i}})^{T} \vec{\mu} \,\,= \,\, p\cdot \sum_{\omega_{j}\in A_{i}} \mu_j - q\cdot \mu_{i} &\geq 0 \\
\text{or, equivalently, } \frac{p}{q}\cdot \sum_{\omega_{j}\in A_{i}} \mu_j &\geq \mu_{i} 
\end{align*}
where $B_{i}=\{\omega_{i}\}$, and similarly each row of $\mathbf{M}_{\succ}$ captures $q \cdot \mu_i - \sum_{\omega_{j}\in B} p\cdot \mu_{j}>0$  where $A_{i}=\{\omega_{i}\}$, corresponding to the constraint $ \mu_i > \frac{p}{q} \sum_{\omega_j\in B}\mu_j$. 
A solution vector $\vec{\mu}= (\mu_{1},...,\mu_{n})\in\mathbb{R}^{n} $ to the above system represents a (non-normalised) mass function which, after normalising by $||\vec{\mu}||$, yields the desired probability measure representing $\sigma$, where $\mu(\omega_{i}) := \mu_{i}/\sum_{j\leq n} \mu_j$.\footnote{This should be clear from the construction, but it is worth flagging one subtler case in the verification which depends on the fact that $t\geq 1/2$. First, for any $\omega\in\Omega$ and $X\subseteq\Omega$ with $\omega\not\in X$, if $\omega\succ_{\sigma} X$, then the matrix $\mathbf{M}_{\succ}$ evidently captures the constraint $\mu(\omega)> \frac{p}{q}\mu(X)$, which means $\mu(\omega)>\frac{t}{1-t}\mu(X)$. If, on the other hand, $\omega\not\succ_{\sigma} X$, we have $X\succeq_{\sigma} \omega$. Now there are two cases to consider: the first case is that $X\not\succ_{\sigma}\{\omega\}$, in which case $\mathbf{M}_{\succeq}$ directly captures the constraint $\mu(\omega)\leq \frac{p}{q}\mu(X)$, which means $\mu(\omega)\leq \frac{t}{1-t}\mu(X)$, as required. In the remaining case, we have $X\succ_{\sigma}\{\omega\}$, which means that $X=\{x\}$ for some $x\in\Omega$, in which case $\mathbf{M}_{\succ}$ already contains the constraint $\mu(\{x\})> \frac{p}{q}\mu(\omega)$. Now, because $t=\frac{p}{p+q}\geq 1/2$, we have that $\frac{p}{q}\geq 1$, and that ensures that  $\mu(X)> \frac{p}{q}\mu(\omega)$ already entails  $\frac{p}{q}\mu(X)\geq  \mu(\omega)$, i.e. $\frac{t}{1-t}\mu(X)\geq  \mu(\omega)$. In other words: because $t\geq 1/2$, the constraint $X\succ_{\sigma}\omega$, capturing the inequality $\mu(X)>\frac{t}{1-t}\mu(\omega)$, already entails the constraint $X\succeq_{\sigma}\omega$, corresponding to $\frac{t}{1-t}\mu(X)\geq \mu(\omega)$. But this entailment does not hold in general for $t<1/2$.}

Suppose towards a contradiction that the system \eqref{MATRIX2} does not admit a solution. The matrices $\mathbf{M}_{\succeq}$ and $\mathbf{M}_{\succ}$ are rational-valued, since they only contain entries in $\{-q,-p,0,p,q\}$. Then by Motzkin's Transposition Theorem, there exists infeasibility certificates $\alpha\in \mathbb{Q}^{k}$, $\beta\in \mathbb{Q}^{m}$ with non-negative rational coordinates such that
\begin{equation}\label{transpGEN}
\alpha^{T}\mathbf{M}_{\succeq} + \beta^{T}\mathbf{M}_{\succ}=\vec{0}
\end{equation}
and $\beta$ has at least one positive coordinate. Now, we can without loss of generality assume those are vectors in $\mathbb{N}$, by multiplying the entries by a common denominator. In other words, the linear combination of the rows of both matrices, where we multiply the $i$-th row of $M_{\succeq}$ (respectively, $M_{\succ}$)  by the natural number $\alpha_{i}$ (respectively, $\beta_{i}$) yields the $\vec{0}$ vector. 

The unfeasibility certificate gives us a $t$-balanced sequence of sets, which witnesses the failure of the \textsf{Scott}[$t$] scheme. As before, we let $mA$ to denote the sequence of sets consisting of the set $A$ repeated $m$ times. Now we claim that the following sequences or sets are $t$-balanced:
\begin{equation}\label{GENBAL}
(\alpha_{1}A_{1},\dots, \alpha_{k}A_{k}, \beta_{1}A_{k+1},\dots \beta_{m}A_{k+m})\equiv^{t}_{0} (\alpha_{1}B_{1},\dots, \alpha_{k}B_{k}, \beta_{1}B_{k+1},\dots \beta_{m}B_{k+m}).
\end{equation}
To see this, observe that each $\omega$ occurs on the dominant side of $\sum_{i\leq k}\alpha_i\ind{A_{i}}(\omega)$ many weak inequalities in this sequence. For each occurrence in a weak inequality $(i\leq k)$, its adjusted occurrence count is $p$. Similarly, $\omega$ occurs on the dominant side of $\sum^{k+p}_{i=k+1}\beta_i\ind{A_{i}}(\omega)$ many strict inequalities in this sequence. Each of these contributes $q$ towards its adjusted occurrence count. So its total adjusted left occurrence count is:
\begin{align*}
L(\omega) &= p\cdot \sum_{i\leq k}\alpha_i\ind{A_{i}}(\omega) + q\cdot \sum^{k+m}_{i=k+1}\beta_i\ind{A_{i}}(\omega) 
\\
&= \sum_{i\leq k}\alpha_i (p\cdot\ind{A_{i}}(\omega)) + \sum^{k+m}_{i=k+1}\beta_i (q\cdot\ind{A_{i}}(\omega))
\end{align*}
By the same reasoning, $R(\omega)=\sum_{i\leq k}\alpha_i (q\cdot\ind{B_{i}}(\omega)) + \sum^{k+m}_{i=k+1}\beta_i (p\cdot\ind{B_{i}}(\omega))$. Now, note that the $i$-th coordinate of $\alpha^{T}\mathbf{M}_{\succeq} + \beta^{T}\mathbf{M}_{\succ}$ is simply $L(\omega_{i})-R(\omega_{i})$. The $j$-th row of $\mathbf{M}_{\succeq}$ is multiplied by $\alpha_j$, and the $j$-th row of $\mathbf{M}_{\succ}$ by $\beta_j$. The $i$-th coordinate of $\alpha^{T}\mathbf{M}_{\succeq} + \beta^{T}\mathbf{M}_{\succ}$ then consists of the sum
\begin{align*}
& \sum_{\{j:\, \omega_{i}\in A_j \text{ and }A_{j}\succeq_{\sigma} B_{j}\}} \alpha_j p \phantom{-}+ 
\sum_{\{j:\, \omega_{i}\in A_j \text{ and }A_{j}\succ_{\sigma} B_{j}\}} \beta_{j}q \phantom{-}+
\sum_{\{j:\, \omega_{i}\in B_j \text{ and }A_{j}\succeq_{\sigma} B_{j}\}} -q \alpha_{j}  \phantom{-}+ 
\sum_{\{j:\, \omega_{i}\in B_j \text{ and }A_{j}\succ_{\sigma} B_{j}\}} -p \beta_{j} \\
&=\qquad\,\,\underbrace{\sum_{j\leq k}\alpha_j (p\cdot\ind{A_{j}}(\omega_{i}))  \phantom{-}+  
\sum^{k+m}_{j=k+1}\beta_j (q\cdot\ind{A_{j}}(\omega_i))}_{L(\omega_i)} \phantom{-}-
\underbrace{\big(\sum_{j\leq k}\alpha_j (q\cdot\ind{B_{j}}(\omega_i)) \phantom{-}+
\sum^{k+m}_{j=k+1}\beta_j (p\cdot\ind{B_{j}}(\omega))\big)}_{R(\omega_i)} 
\end{align*}
And so the $i$-th coordinate of $\alpha^{T}\mathbf{M}_{\succeq} + \beta^{T}\mathbf{M}_{\succ}$ is simply $L(\omega_i)-R(\omega_i)$. From \eqref{transpGEN}, since $\alpha^{T}\mathbf{M}_{\succeq} + \beta^{T}\mathbf{M}_{\succ}=\vec{0}$, we thus have that $L(\omega_{i})=R(\omega_{i})$ for every $\omega_{i}$. So the sequence in \eqref{GENBAL} is $t$-balanced.
By design of the matrices, we have that $A_{i}\succcurlyeq^{\ast}_{\sigma} B_{i}$ for all $A_i$, $B_i$ occurring in these sequences. Now we know our system contains at least some strict inequalities of the form $\omega_{i}\succ \emptyset$, so that $\mathbf{M}_{\succ}$ is non-empty. It follows from Motzkin's theorem that $\beta$ is not identically null: say, $\beta_{i}>0$. But this means that the $i$-th inequality $A_{i}\succ_{\sigma} B_{i}$ is strict. Without loss of generality, we can assume that $i=n$ (note that the notion of $t$-balancedness is not sensitive to the order of indices of pairs $(A_{i}, B_{i})$)\footnote{As long as the indexing of the $A_i$'s and the indexing of the $B_i$'s are not permuted independently of one another.}. So the sequences of sets in \eqref{GENBAL} constitutes a failure of the \textsf{Scott}[$t$] axiom. We thus conclude that system \eqref{MATRIX2} is solvable after all, which ensures the existence of a probability measure representing $\sigma$ for threshold $t$.
\end{proof}

The representation theorem is formulated for rational thresholds. But note that representability with respect to rational thresholds already captures representability \emph{simpliciter}.

\begin{prop}\label{QREP}
    Every representable selection function on a finite algebra is $t$-representable for some $t\in\mathbb{Q}$.
\end{prop}

\begin{proof}
Let $\Omega=\{\omega_{1},...,\omega_{n}\}$. Suppose $(\Omega, \mathcal{P}(\Omega), \sigma)$  is a selection structure that is representable for some threshold $t$. Let $\mu$ a measure such that $\sigma =\sigma_{\mu, t}$. Define $k:=\frac{t}{1-t}$. For every $\omega_{i}$, define $\ell_{i}:= \max\{\mu(X)\,|\, \omega_i\succ_{\sigma} X \}$ and $r_i := \min \{\mu(X)\,|\, X\succeq_{\sigma} \omega_i \}$. To say that $\mu$ represents $\sigma$ for threshold $q$ is equivalent to the statement that, for all $i\leq n$, we have $\ell_i < \frac{\mu(\omega_i)}{k} \leq r_i$, or equivalently $\frac{\mu(\omega_i)}{r_i} \leq k < \frac{\mu(\omega_i)}{\ell_i}$. In particular, $\max_{i\leq n}\frac{\mu(\omega_i)}{r_i} \leq k < \min_{i\leq n}\frac{\mu(\omega_i)}{\ell_i}$. So the interval $I=\big[\max_{i\leq n}\frac{\mu(\omega_i)}{r_i}, \min_{i\leq n}\frac{\mu(\omega_i)}{\ell_i}\big)$ is nonempty. Pick any $q\in I\cap\mathbb{Q}$. Then we have, $\ell_i < \frac{\mu(\omega_i)}{q} \leq r_i$ for all $i\leq n$. But this just means that the same measure $\mu$ also represents $\sigma$ with the rational threshold $t:=\frac{q}{q+1}$.
\end{proof}
This last proposition entails that no probabilistically stable belief revision policy can force the representing threshold to be irrational: any such belief revision plan is consistent with a rational stability threshold. In that sense, every probabilistically stable revision operator is captured by our representation theorem. Further, from the argument above we can immediately extract the following fact:

\begin{observation}
 Given a measure $\mu$ and a selection function $\sigma$ that $\mu$ represents, the range of thresholds that generate the exact same belief revision policy as $\sigma$ is given exactly by the interval $[\alpha, \beta)$ where 
$$
\alpha=\max_{\omega\in \Omega} \frac{\mu(\omega)}{\min \{\mu(X)\,|\, X\succeq_{\sigma} \omega_i \}} \qquad\text{ and }\qquad
\beta =\min_{\omega\in \Omega} \frac{\mu(\omega)}{\max \{\mu(X)\,|\,  \omega_i\succ_{\sigma} X \}}
$$
\end{observation}
This concludes our presentation of the general representation theorem. We now turn to some further applications.

\section{Applications and other directions}

\subsection{The logic of `at least $k$ times more likely than'} 

Consider a comparative probability structure consisting of a finite algebra of events $\mathfrak{A}$ and a binary relation $\succ$ on $\mathfrak{A}$. Given a rational number $k$, define an order $\succ$ to be \emph{$k$-representable} if and only if there exists a probability measure $\mu$ on $\mathfrak{A}$ such that for every $A,B\in\mathfrak{A}$, we have 
$$
A\succ B\qquad\text{if and only if} \qquad \mu(A)> k \cdot \mu(B) 
$$
It is a straightforward exercise to adapt the proof strategy from Theorem \ref{GENREPSEL} to prove a representation theorem for the $k$-representable comparative probability structures. We leave the details to another occasion, as this is not the main object of this paper: but it should be clear that the key axiom for representation is given by the obvious variant of the \textsf{Scott} scheme:
\begin{center}
\textsf{(Scott$[k]$)}\qquad\qquad If $(A_{i})_{i\leq n}\equiv^{k}_{0}(B_{i})_{i\leq n}$ and $A_{i}\succcurlyeq^{\ast}B_{i}$ for all $i<n$, then $A_{n} \not\succ B_{n}$.
\end{center}
where $A_{i}\succcurlyeq^{\ast}B_{i}$ means that either $A_{i}\succ B_{i}$ or $B_{i}\not\succ A_{i}$, and $(A_{i})_{i\leq n}\equiv^{k}_{0}(B_{i})_{i\leq n}$ denotes (abusing notation somewhat for better legibility) that the sequences are \emph{balanced for ratio $k$}, which means that they are balanced for threshold $t=k/(k+1)$ in the sense of Definition \ref{TBalanced}. In other words, it means that the corresponding balancedness equation $q \cdot  [\succ, A] + p \cdot [\succeq, A] =  p \cdot [\succ, B]+ q \cdot [\succeq, B]$ always holds for $p,q$ where $k=p/q$.\footnote{Here of course $\succeq$ is understood as $\not\prec$. For the case where $0<k\leq 1$, we need to adjust the formulation of balancedness slightly: the count $[\succeq, A]$ now counts occurrences in $A_i$ for all weak inequalities $B_i\not\succ A_i$, while $[\succ, A]$ counts only occurrences in $A_i$ in those cases where $A_{i}\succ B_{i}$ and $B_{i}\succ A_{i}$ both hold. This is because, when $k<1$, the constraint $\mu(A)> k\cdot \mu(B)$ no longer entails the constraint $k\cdot\mu(A)\geq \mu(B)$, but the converse entailment holds instead.}

This allows to spell out the exact conditions under which a primitive relation between events may or may not be interpreted as a ``$k$ times more likely than'' relation.\footnote{If we do not assume regularity (that event non-contradictory event has positive probability), we need to resort to a formulation which strengthens the antecedent to the  condition $(A_{i})_{i\leq n}\leq^{k}_{0}(B_{i})_{i\leq n}$, which only requires that each state's adjusted count is at least as high for the $B$-sequence of events than for the $A$-sequence. It should also be clear that, with that replacement, we obtain a stronger result: namely, a threshold-dependent version of Theorem \ref{strongrep}, which gives conditions for a pair of relations $(\succ, \succeq)$ to be jointly representable as :
\begin{align*}
   \text{if }A\succ B\qquad\text{then} \qquad \mu(A)> k \cdot \mu(B)  \\
   \text{if }A \succeq B\qquad \text{then} \qquad k\cdot \mu(A)\geq \mu(B)  
\end{align*}} This specific notion of representability is not the object of much discussion in the measurement theory literature. But it is a natural one: it can fruitfully be compared to previous work by \cite{FISH} which studies the `\emph{sufficiently more likely than}' relation understood in an additive threshold sense as $\mu(A)>\mu(B)+\epsilon$. The above can be seen as capturing the logic of `sufficiently more likely than' understood in a \emph{ratio} threshold sense $\mu(A)>k\cdot \mu(B)$. Lastly, since one can define having unconditional probability above a rational threshold $t$ in terms of a ratio threshold as $\mu(A)> \frac{t}{1-t} \mu(A^c)$, this approach allows to straighforwardly answer questions involving the axiomatisation of probabilistic threshold-based neighbourhood models, such as those studied by \cite{vanEijck}, as well as modal logics with modalities for `truth in at least $p\%$ of successors' in a Kripke frame.

From our present results we can extract more information about the notion of $k$-representability for various values of the ratio parameter $k$. Proposition \ref{IncreasingThresholds} established that whenever $t<q$, there exists a $t$-representable selection structure that is not $q$-representable. The reasoning deployed there can immediately be used to establish the following corollary about ratio representations in comparative probability orders. 
\begin{cor}
    Whenever $1\leq k <q$, there exists a comparative probability structure that is $k$-representable but not $q$-representable.  
\end{cor}
\begin{proof}
    As in the previous proof. For each choice of $n,m\in\mathbb{N}$, take $\Omega=\{b_{1},...,b_{m},c_{1},...,c_{n}\}$, and consider the following constraints: 
    \begin{align}
        \{ b_{i}\}   &\succ  \{b_{i+1}\}  \text{ for all }i\in\{1,...,m-1\}  \label{sigmaconstraintA2} \\
        \{b_{m-1}\}  &\succ \{c_{i}\}\text{ for all }i\in\{1,...,n\} \label{sigmaconstraintB2}  \\
           \{b_{1}\}  &\not\succ \{c_{1},...,c_{n}\}\label{sigmaconstraintC2}  
    \end{align}
    The same argument as in Proposition \ref{IncreasingThresholds} establishes the result.
\end{proof}

One last worry one may have about the notion of $t$-balancedness at the heart of this axiomatisation is that its very formulation looks inherently `quantitative', since it involves \emph{counting} occurrences of states differently depending on whether they occur in a strict or weak comparison between events. In particular, one might worry that this axiom outstrips the expressive capabilities of the usual modal logics of qualitative probability as those studied since at least \cite{SEG} and \cite{GAR}. The plain (modal) language of comparative probability usually includes no more than Boolean combinations of atomic propositions of the form $\alpha \succ \beta$, where $\alpha \succeq \beta$, where $\alpha, \beta$ are propositional formulas interpreted as events in a Boolean algebra. A crucial fact when axiomatising these logics of comparative probability is that the usual condition of two sequences of events being balanced, $(A_1,\dots, A_n)\equiv_0 (B_1,\dots, B_n)$, is in fact expressible as a purely Boolean statement, a fact independently observed by \cite{DOM} and \cite{SEG}. This is what allows for the original Scott axiom\textemdash the Finite Cancellation axiom scheme \text{{Q4} from \ref{Scott's Theorem}}\textemdash to be captured in the corresponding modal language. It is perhaps less clear whether the generalised Scott axiom \textsf{(Scott$[t]$)}, which relies on the condition of $t$-balancedness, can be expressed in such a language. 

A moment's reflection shows, however, that it can. Set $t=p/(p+q)$ the irreducible fraction representation of $t$. Say that, in a sequence of comparisons $(A_i \bowtie_i B_i)_{i\leq n}$ between events where $\bowtie_{i}\in\{\succ, \not\prec\}$, the $A$ sequence \emph{Pareto-dominates} the $B$ sequence if every coordinate $A_i$ dominates its $B_i$ counterpart (either $A_i \succ B_i$ or $A_i \not\prec B_i$) and at least one strictly dominates its counterpart ($A_i \succ B_i$). The function of the \textsf{(Scott$[t]$)} axiom is simply to rule out sequences of comparisons where (a) one side Pareto-dominates the other while (b) the sequences are $t$-balanced (the count of each occurrence of a state in $A_i$ or $B_i$ being appropriately multiplied by either $q$ or $p$, depending on whether $\bowtie_i$ is strict or weak). 

This can be expressed simply by having an axiom ruling out $t$-balancedness for every pair of sequences of events where one Pareto-dominates the other. The trick is that expressing the $t$-balancedness for a sequence of comparisons $(A_i \bowtie_i B_i)_{i\leq n}$ requires a different Boolean formulation depending on the weak or strict nature of each $\bowtie_i$ relation. A convenient way to specify this in the meta-language is to fix a function $f:\{1,\dots,n\}\to\{0,1\}$ which specifies whether the $i$-th inequality $A_{i}\bowtie_i B_i$ is strict $(f(i)=1)$ or weak $(f(i)=0)$. Appealing again to the notation $mA$ to represent a sequence of $m$ consecutive copies of the event $A$, define the $t$-balancedness statement for a configuration determined by $f$ to be:
\begin{align*}
\mathbb{B}_{t}(f,A_1,\dots, A_{n},B_{1},\dots, B_{n})&:= 
\big( l_{f(i)} A_{i}\big) \equiv_0  \big( r_{f(i)}  B_{i}\big) 
\end{align*}
where $$l_i := \begin{cases}
 q &\qquad\text{if } i=1\\ 
 p &\qquad \text{if } i=0
\end{cases}
\qquad \text{ and } \qquad r_i := \begin{cases}
 p &\qquad\text{if } i=1\\ 
 q &\qquad \text{if } i=0
\end{cases}
$$
Writing the functions $l_i$ and $r_i$ more explicitly, we can write
$$
\mathbb{B}_{t}(f,A_1,\dots, A_{n},B_{1},\dots, B_{n}):=\big( (qf_i + (1-f_i)p ) A_{i}\big) \equiv_0  \big( (pf_i + (1-f_i)q ) B_{i}\big) 
$$
Each instance of this scheme is a formula in the plain language of comparative probability. Given some fixed function $f$, integers $p$ and $q$, and two concrete sequences of formulas $(A_i)_{i\leq n}$ and $(B_{i})_{i\leq n}$ representing events, the expression on the right hand side is simply a (purely Boolean) statement expressing that certain two finite sequences of events are balanced in the ordinary sense (keeping in mind, for instance, that the expressions of the form $((pf_i + (1-f_i)q)A_{i})_{i\leq n}$ are meta-linguistic abbreviations for a concrete sequence of $A_i$'s). It can straightforwardly be verified that this is equivalent to the $t$-balancedness of a sequence of inequalities encoded by $f$.

We can then express the generalised Scott axiom for threshold $t$ by the following scheme: 

\begin{center}
$\textsf{(S[}t\textsf{])}_{n}\qquad \displaystyle\bigwedge_{\substack{f:\{1,\dots,n\}\to\{0,1\}\\ f(n)=1 }}  \Big(\big(\bigwedge_{i\leq n}(A_i\bowtie_{f(i)}B_i)\big)\to \neg \mathbb{B}_{t}(f,A_1,\dots, A_{n},B_{1},\dots, B_{n}) \Big)$
\end{center}
where
$$
\bowtie_i := \begin{cases}
 \succ &\qquad\text{if } i=1\\ 
 \not\prec &\qquad \text{if } i=0
\end{cases}
$$
Since each function $f:\{1,\dots,n\}\to\{0,1\}$ specifies whether 
the $i$-th inequality is strict $(f(i)=1)$ or weak $(f(i)=0)$, by taking the conjunction over all functions $f:\{1,\dots,n\}\to\{0,1\}$ with $f(n)=1$\textemdash which force at least one strict comparison \textemdash we are in effect quantifying over all sequences where the $A_i$'s Pareto-dominate the $B_i$'s. This axiom scheme thus states exactly what is needed to capture the generalised Scott axiom: for every pair of finite sequences $A=(A_i)$, $B=(B_i)$, if $A$ Pareto-dominates $B$, then the two are not $t$-balanced. This formulation opens a direct route to a complete axiomatisation of the corresponding modal logic of ``at least $k$ times more likely''.

\subsection{Logics of probabilistic stability} \label{Logics of probabilistic stability}

The representation results proved here can form the basis for a completeness theorem for a logic of probabilistic stability. To capture the logic of probabilistic stability in a formal language, there are various salient choices of language beyond the plain language of flat conditionals of the form $\varphi\nc\psi$ discussed above. One option is to consider a language $\mathcal{L}_{\mathsf{KBS}}$ of the form:
$$
\varphi,\,\psi ::= p  \,\,|\, \,  \varphi\wedge\psi \,\,|\, \,  \neg\varphi \,\,|\,\,\mathsf{B}^{\varphi}\psi \,\,|\,\,\mathsf{K}\varphi \,\,|\,\, \mathsf{S}^{\varphi}\psi
$$
Where $\mathsf{K}$, $\mathsf{B}$ and $\mathsf{S}$ are interpreted as operators for \emph{knowledge}, \emph{conditional belief}, and \emph{conditional strongest belief}, respectively. The models for the logic are given by representable selection structures $(\Omega,\mathcal{P}(\Omega), \sigma)$ characterised in Theorem \ref{GENREPSEL}. We define a valuation $\e{\cdot}:\mathcal{L}_{\mathsf{KBS}}\to \mathcal{P}(\Omega)$ as usual, with the key recursive clauses being
\begin{alignat*}{3}
&\mathfrak{M},\omega \vDash  \mathsf{B}^{\varphi}\psi \hspace{2em}&\Leftrightarrow  \hspace{2em}&\sigma(\e{\varphi})\subseteq\e{\psi}\\
&\mathfrak{M},\omega \vDash  \mathsf{S}^{\varphi}\psi  \hspace{2em}&\Leftrightarrow  \hspace{2em}&\sigma(\e{\varphi})=\e{\psi}
\end{alignat*}
In this language, some non-trivial validities include, for example: 
\begin{align*}
(\textsf{RM}) \qquad &(\mathsf{B}^{\varphi}\theta\wedge \neg\mathsf{B}^{\varphi}\neg\psi)\rightarrow \mathsf{B}^{\varphi\wedge\psi}\theta \\
(\textsf{S4}_{n}) \qquad &\Big(\bigwedge_{i\leq n} \mathsf{S}^{{\varphi\vee \psi_{i}}}\psi_{i} \Big)\rightarrow \mathsf{B}^{(\vee_{i\leq n} \psi_{i})\vee\varphi}\bigvee_{i\leq n} \psi_{i} \\
(\textsf{SM}) \qquad &\big(\sel^{\varphi\wedge\gamma}\psi \wedge \bel^{\varphi}\psi \big)\rightarrow \sel^{\varphi}\psi 
\end{align*}
where the first two capture the properties  $\textsf{(S3)}$ and $\textsf{(S4}_{n})$ from Theorem \ref{GENREPSEL}.

Another attractive option is to consider a more expressive logic with a \emph{typicality} operator $\nabla$ (see \cite{BOO}), which we can interpret on a selection function model $\mathfrak{M}=(\Omega,\mathcal{P}(\Omega),\sigma,\e{\cdot})$ as capturing exactly the selected states (the strongest stable event conditional on the input). We set
$$
\mathfrak{M},\omega\vDash\nabla\varphi \text{ if and only if }\omega\in \sigma(\e{\varphi})
$$
Then the $\bel$ and $\sel$ operators are definable: we can write $\bel^{\varphi}{\psi}\leftrightarrow \mathsf{K}(\nabla\varphi\rightarrow\psi)$ and $\sel^{\varphi}{\psi}\leftrightarrow \mathsf{K}(\nabla\varphi\leftrightarrow\psi)$. This is an expressive language that can capture various interesting properties over Leitgeb models. For instance, consider the formula
$$
\kn(\nabla\varphi \leftrightarrow \nabla\psi)
$$
which states that $\sigma(\e{\varphi}) = \sigma(\e{\psi})$: any hypothesis that is accepted after learning $\varphi$ is also accepted after learning $\psi$. Or, in other words: $\varphi$ and $\psi$ have the same non-monotonic consequences. 

Another useful fact about typicality operators\textemdash and one that may render the axiomatisation problem quite interesting\textemdash is that they can directly express \emph{iterations} of the selection function. As a consequence, we can describe, in quite some detail, various fine features of the probability measure generating the selection function $\sigma$. For example: given $A\in\mathfrak{A}$, we define the $\sigma$-depth of $A$ as $d(A):=\min\{n\in\mathbb{N}\,|\, \sigma^{n}(A)=A\}$. This gives us an approximate way to asses how concentrated the underlying probability measure is in $A$. Very roughly, an event of low $\sigma$-depth is one on which the measure is spread rather uniformly: a set of high $\sigma$-depth is one on which the measure is closer to being `big-stepped' (to use terminology from \cite{BEN}), i.e. with large probability gaps between individual atoms. Typicality operators can express the fact that an event $\e{\varphi}$ has depth $n$, through the formula
$$
\neg\kn(\nabla^{n-1}\varphi\rightarrow\nabla^{n}\varphi) \wedge \kn(\nabla^{n}\varphi\rightarrow\nabla^{n+1}\varphi) 
$$
What is the typicality logic of probabilistically stable revision (equivalently, of strongest-stable-set operators)? We leave this as a task for another occasion. More generally, we note that typicality logics have been chiefly studied for selection functions that are representable as order-minimisation operators: there remain many open questions about axiomatising more general classes of selection functions, such as the strongest stable set operators, which are not "trackable" by order-minimisation (for instance, they do not validate the idempotence axiom $\nabla\varphi\rightarrow\nabla\nabla\varphi$).

\subsection{Connection with simple voting games} 

\emph{Simple voting games} are (simple) structures with various interesting properties, studied in game theory and combinatorics \citep{SLI, TAY}. A simple voting game is a pair $(P, \mathcal{W})$, where $P$ usually represents a finite set of voters, and $\mathcal{W}\subseteq 2^{P}$ the set of winning coalitions, required to be closed under taking supersets. The game admits a quota representation whenever there is a quota $q\in\mathbb{R}$ and a weight function $m:P\rightarrow\mathbb{R}^{+}$ such that $$A\in\mathcal{W}\Leftrightarrow \sum_{a\in A}m(a)\geq q,$$
which means that there is a way to assign a weight to each player, in such a way that a coalition is winning exactly when the collective weight of the players in the coalition reaches (or surpasses) the required quota $q$.

Our representation problem bears a close connection to the theory of simple games. Consider a selection structure $(\Omega, \mathfrak{A})$ with  $\Omega=\{\omega_{1},...,\omega_{n}\}$. Let $\mathcal{D}_{i}:=\{X\subseteq\Omega\,|\,\omega_{i}\succ_{\sigma} X \}$, the collection of sets dominated by $\omega_{i}$. Each state $\omega_{i}$ generates a simple voting game $G_{i}=(\Omega, {\mathcal{D}_{i}}^{c})$, where the closure-under-superset condition is satisfied thanks to the (M1) property (see page \pageref{SLISTb}): the property ensures that if $X\in{\mathcal{D}_{i}}^{c}$, $X\subseteq Y$, then $Y\in {\mathcal{D}_{i}}^{c}$). Now suppose $\mu$ is a representation for the induced order $\succcurlyeq^{\ast}_{\sigma}$. Then we have, for each $\omega_{i}$, that 
\begin{center}
$X\in {\mathcal{D}_{i}}^{c}$ iff $\mu(X)\geq \mu(\omega_{i})$.
\end{center}
In other words, if $\mu$ is a probabilistic representation for $\succcurlyeq^{\ast}_{\sigma}$, it is also weight representation of each game $G_{i}$. More specifically, $\mu$ \emph{simultaneously} represents all games $G_{i}$, where the quota for each $G_{i}$ is $\mu(\omega_{i})$. Conversely, each such simultaneous quota representation of the associated system of games  $\{G_{i}\}_{i\leq n}$ gives rise to a probability distribution representing the order $\succcurlyeq^{\ast}_{\sigma}$ (it suffices to normalise each weight by the weight of the grand coalition $\Omega$). Thus, finding necessary and sufficient conditions for $\succcurlyeq^{\ast}_{\sigma}$ to be representable is equivalent to finding the exact conditions for the collection of games $\{G_{i}\}_{i\leq n}$ to be simultaneously representable in this sense (note that an important aspect of simultaneous representations is that the quotas themselves depend on the weight function).

We can give this problem a possible game-theoretic interpretation in terms of what we could call `coordinated blocking games'. We first identify each state $\omega_{i}\in\Omega$ with a player. Each game $G_{i}$ tells us which coalitions can block any decision that player $\omega_{i}$ supports. If the collection of games $\{G_{i}\}_{i\leq n}$ is simultaneously representable as in the above, we can say that the voting system admits a coherent weight representation: one can attribute weights to all players in a uniform manner, in such a way that a coalition is winning in $G_i$ exactly if its cumulative weight is above the weight of the $i$-th player\footnote{We can illustrate this with the following scenario. Imagine that you are a historian researching the voting protocols of an ancient civilization. You know that each province sent a delegate to vote in a national council, but you have no explicit information about how exactly the outcomes of the votes were determined. The only information available is records of some results of the votes which specify which coalitions were able to block which delegate (that is, you have a collection of games $\{G_{i}\}_{i\leq n}$ like the above). The working hypothesis is that the vote of each delegate was accorded a fixed \emph{weight}\textemdash proportional, say, to the population of the delegate's province. In order to check if this hypothesis is at least consistent with your data, you must check if  the games $\{G_{i}\}_{i\leq n}$ are simultaneously representable.}.

Recall two properties of the $\succ_{\sigma}$ order that we highlighted above (p. \pageref{SLISTb}):
\begin{itemize}
\item[(M1)] If $\omega\succ_{\sigma} X$ and $X\supseteq Y$, then $\omega\succ_{\sigma} Y$.
\item[(Sc)] If $(A_{i})_{i\leq m}$ and $(B_{i})_{i\leq m}$ are balanced sequences and $(\omega_{i})_{i\leq m}$ a sequence of states, then \\ ${\forall i\leq m,\, \omega_{i}\succ_{\sigma} A_{i}} \text{ entails }\exists i\leq m,\,\omega_{i}\succ_{\sigma} B_{i}$. 
\end{itemize}

In the context of a simple games, the (Sc) property entails that one cannot transform a collection $A_{1},....,A_{n}$ of winning coalitions into a collection of losing coalitions by a sequence of pairwise exchanges of players from one coalition to another. 
We know from classical results on simple games (see \cite{TAY}) that axioms (M1) and (Sc) suffice for each individual game $\{G_{i}\}_{i\leq n}$ to be weight-representable. This follows by a standard type of hyperplane separation argument characteristic to comparative probability: we identify each coalition with its characteristic function, and the (Sc) axiom guarantees that there is a hyperplane separating the winning from the losing coalitions. The normal vector to the separating hyperplane determines the desired weight function. In the case of representable selection functions, the full (\textsf{Scott}) axiom on $\succcurlyeq^{\ast}_{\sigma}$ suffices to ensure that there is a way of constructing representations of all the games $G_{i}$ that are consistent with each other; e.g., that no separating hyperplane for $G_{i}$ lumps together some winning and losing coalitions from another game $G_{j}$.

We can also interpret strongest-stable-set operators in the context of voting games. Consider a weighted voting game $(P, W)$ with weight representation $m:P\rightarrow \mathbb{R}^{+}$. For each player $p\in P$, the \emph{projected} game $G_{p} = (P\setminus \{p\}, W_{p})$ is given by $W_{p}:= \{X\subseteq P\setminus \{p\}\,|\, \sum_{q\in X} m(q) \geq m(p) \}$: i.e., the winning coalitions in $G_p$ are exactly those that can outweight player $p$. The weight function $m$ gives rise to a selection function $\sigma:\mathcal{P}(P)\rightarrow\mathcal{P}(P)$ (via the probability function obtained by normalising $m$). Given any subset of players $A\subseteq P$, the function $\sigma$ outputs the minimal $X\subseteq A$ such that $A\setminus X$ is a losing coalition in every reduced game $G_{p}$ (with $p\in X$). The set $\sigma(A)$ represents the minimal coalition of players such that each individual player $p\in\sigma(A)$ can block (or `veto') the coalition $A\setminus\sigma(A)$. Suppose, for instance, that we play the following game. First, we restrict attention to a subset $A$ of players. Then we play a \emph{champion game} on $A$: a champion $p$ is picked at random from some pre-selected subset $X\subseteq A$.  The champion then votes against the entire opposition $A\setminus X$. Different choices of possible champion sets $X$ yield different chances of winning against the opposition. A \emph{decisive} team is a subset of players $X\subseteq A$ such that, no matter who is chosen from $X$ as a champion, the opposition $A\setminus X$ loses against the champion. We can think of $\sigma(A)$ as the minimal decisive team in the champion game on $A$. 

We can also give such games, and their associated selection functions, a slightly different interpretation. One way to identify the concentration of power in a voting system is to ask for the minimal coalition that can force the outcome of any vote that its members take part in. Suppose that a group of agents $P$ engages in voting by forming various committees: some issues get decided by a plenary vote involving the `grand committee' $P$, while others get decided by a vote among members of a restricted committee $Y\subset P$. Say that a coalition of players $C$ is a \emph{stably decisive coalition} for a set of players $X$ if the members of $C$ can force the decision in every committee $Y\subseteq X$ that contains some members of $C$ (for instance: $C$ holds the majority in every committee that includes at least some members of $C$). That is: if we restrict voting to players in $X$, not only can $C$ decide the outcome of any `plenary' vote involving all players in $X$, but they also decide the outcome of any subcommittee vote, provided the subcommittee overlaps with at least some members of $C$. 

Suppose committees vote by (super)majority: for each committee $X\subseteq P$, each vote requires at least $100t\%$ of the committee's votes to pass, for some fixed threshold $t\in[1/2,1)$ (where $t=1/2$ corresponds to plain majority voting). The $t$-representable selection functions are then exactly those that, given a subset of players $X\subseteq P$, return the \emph{smallest} stably decisive coalition for $X$. These are exactly the functions captured by our representation theorem.

\begin{ex}
The College Council features one representative from each of its departments. Some decisions are settled by vote of the entire Council; other decisions are decided by vote among subcommittees\textemdash such as specific Schools, or ad-hoc committees created for decisions which involve particular collections of departments. Every vote, whether involving the entire Council or merely some sub-committee of its representatives, is decided by simple majority voting. Each representative's voting weight corresponds to the size of their department: their vote counts for $n$ votes, where $n$ is the number of faculty members in their department. The players and their voting weights are given by the following table:

\begin{center}
\normalfont
\begin{tabular}{|cc|c|}
\hline
 \emph{Department} &  & \emph{Size (voting weight)} \\
\hline
 Computer Science &  \textsf{c} & 30\\
  Mathematics  &  \textsf{m} & 25\\
   English &  \textsf{e} &  10\\
  History  &  \textsf{h} &  4\\
  Anthropology &  \textsf{a} &  3\\
  Drama &  \textsf{d} &  2\\
  Fine Arts &  \textsf{f} &  2\\
 \hline
\end{tabular}
\end{center}

The College Council players are $P:=\{\textsf{c},\textsf{m},\textsf{e},\textsf{h},\textsf{a},\textsf{d},\textsf{f}\}$. Let $\sigma$ the selection function which returns, given any group of players $X\subseteq P$, the smallest stably decisive coalition for $X$ under simple majority voting. For simplicity, we will call this set $\sigma(X)$ the \emph{core} of $X$. Here are the cores associated with some groups of players: 
\begin{center}
\normalfont
\begin{tabular}{|c|l| c|}
\hline
\emph{Group} $X$ & \emph{Members}  & \emph{Smallest stably decisive coalition} $\sigma(X)$  \\
\hline
Full College Council & $P:=\{\textsf{c},\textsf{m},\textsf{e},\textsf{h},\textsf{a},\textsf{d},\textsf{f}\}$ & $\{\textsf{c},\textsf{m}\}$\\
School of Formal Sciences  & $F:=\{\textsf{c}, \textsf{m}\}$ & \{\textsf{c}\} \\ 
School of Arts   & $A:=\{\textsf{d}, \textsf{f}\}$ & \{\textsf{d}, \textsf{f}\} \\ 
School of Humanities & $H:=\{\textsf{e}, \textsf{h}, \textsf{a}\}$ & \{\textsf{e}\} \\ 
Arts \& Humanities Council  & $A\cup H=\{\textsf{e}, \textsf{h}, \textsf{a},\textsf{d}, \textsf{f}\}$ & $\{\textsf{e}, \textsf{h}, \textsf{a},\textsf{d}, \textsf{f}\} $\\ 
\hline
\end{tabular}
\end{center}

The School of Formal Sciences\textemdash the coalition of Computer Science and Mathematics\textemdash forms the core of the College Council: any of their representatives represents a majority in every committee that includes them. Note that Mathematics is part of the core for the whole Council although  they are not part of the core among the School of Formal Sciences. The Arts \& Humanities Council has no non-trivial stably decisive coalition beyond the grand coalition (the whole Arts \& Humanities Council itself).  It is also worth observing that the core for the Arts \& Humanities Council is \emph{not} the union of the cores for the School of Arts and the School of Humanities considered separately.\footnote{This structure of course reflects the failure of the \textsf{Or}-rule for probabilistically stable revision.}  
\end{ex}

In this example, the selection function $\sigma:2^P\to 2^P$ is of course numerically representable, since it was generated by a particular choice of numerical weights: as such, it meets the conditions of Theorem \ref{REPSEL}. A more interesting case where Theorems \ref{REPSEL} and \ref{GENREPSEL} can be deployed is going in the opposite direction, when asking about the numerical representability of a given selection function. Given a selection function $\sigma$ capturing all the stably decisive coalitions in some collective decision procedure, and some particular threshold $t$, can we represent this power structure as emerging from $t$-supermajority voting? That is, does there exist any assignment of weights that would realise the same arrangement of decisive coalitions under $t$-supermajority voting? 

\begin{ex}
Consider the selection function $\sigma$ capturing the stably decisive coalitions in the College Council vote. Could this exact same power structure as above be recreated under a more demanding voting rule? Suppose, for example, that each vote requires a 2/3 supermajority to pass. Is there some assignment of weights to the various representatives that will generate the same selection function? Theorem \ref{GENREPSEL} supplies the requisite conditions to answer that question. The answer is negative: this can be seen by noting that the following is a pair of $2/3$-balanced sequences that violates the \textsf{(Scott$[t]$)} axiom for $t=2/3$:

$$
\begin{array}{cccc}
\big(\mathsf{h}, & \mathsf{a}, & \mathsf{a}, & \{\mathsf{d, f}\} \big)\\
\,\sdom & \sdom& \sdom &  \sdomeq  \\
\big(\mathsf{a}, & \mathsf{d}, & \mathsf{f}, & \phantom{\mathsf{a, }} \mathsf{h}\phantom{\mathsf{a}} \big)\\
\end{array}
$$
A brief inspection reveals that every player in $P$ has the same adjusted count in both sequences, which means that the two sequences are indeed balanced for threshold $2/3$.\footnote{
Recall that 2/3-balancedness requires counting an occurrence of an element as \emph{two} occurrences, whenever it occurs on the dominated side of a strict inequality $\omega \succ_{\sigma} X$, or on the dominating side of a weak inequality $X \succeq_{\sigma}\omega$.} As a result, the selection function $\sigma$ above cannot arise for a threshold of $t=2/3$: for this threshold, there is no assignment of voting weights that would generate the same selection function. The structure of stably decisive coalitions of the College Council cannot arise under a 2/3 supermajority voting system. 
\end{ex}

\subsection{Stable revision as the simultaneous representability of Lockean belief sets}

The above analysis in terms of simultaneous representations of multiple games makes another connection quite salient: one that is quite obvious, but nonetheless worthy of note. Given a representable selection function $\sigma$ and a state $\omega_{i}$, consider the set $D^{c}_{i}$ as defined above, i.e., the set of sets that are not dominated by $\omega_i$. In formal epistemology parlance, we can say that each set $D^{c}_{i}$ constitutes a \emph{Lockean} set: a set of propositions that are above a fixed threshold, where the threshold is determined by the weight of $\omega_{i}$. Our numerical representability conditions amount to the condition that not only each $D^{c}_{i}$ is represented as a Lockean set for some probability function, but also that all the Lockean sets $D^{c}_{i}$ are \emph{simultaneously} represented by a single probability function, with each weight being given by the probability of the corresponding state $\omega_{i}$.  

\subsection{Stacking interpretation} 

Given a list of items with masses, we want to stack them safely: we want to avoid any item being damaged\textemdash or, to fix our terminology, \emph{squished}\textemdash by the mass of all items above it. The higher an object's mass $\mu(\omega)$, the more resistant it is to being squished. A stacking is safe only if each item $\omega_{i}$ has higher mass than the sum of all masses of objects above it. When stacking objects, we can imagine that adding layers is cost-free, but places in each layer are not. We thus want the smallest possible safe bottom layer. A representable selection function $\sigma$ picks, given a collection of items $X$, the first layer $\sigma(X)$ in the finest safe stacking of the objects in $X$ (the lowermost one).

\subsection{Choice functions and revealed preference} 

There is a direct formal correspondence between, on the one hand, selection function semantics from conditional and non-monotonic logics and, on the other, the theory of social choice functions and revealed preference theory \citep{Chernoff1954, Sen1971, AizermanMalishevski}. This is a correspondence which was noted and investigated by a number of authors \citep{Lindstrom1991, Rott1993, ROT, Lehmann2001, Bonanno2009, Collins2}. Making the connection requires only a slight perspective switch on our formalism. First interpret the set $\Omega$ as consisting of some basic alternatives, which constitute the objects of choice (some goods, or actions, or candidates, etc.). Then note that each selection function $\sigma:\mathcal{P}(\Omega)\rightarrow\mathcal{P}(\Omega)$ constitutes what is usually called a \emph{choice function} in revealed preference theory: a function which satisfies $\sigma(X)\subseteq X$ and $\sigma(X)\neq\emptyset$ for $X\neq\emptyset$.\footnote{Some authors call this a `choice correspondence', reserving the term `choice function' for singleton-valued functions.} The usual interpretation of the choice function is that $\sigma(X)$ captures the alternatives that an agent or a group of agents deem \emph{acceptable}, given a choice between the elements in $X$. According to one interpretation, the objects of choice/comparison are really the alternatives themselves: the agent(s) is then essentially indifferent between any two acceptable alternatives. Another interpretation sees $\sigma(X)$ as a set of alternatives chosen by an agent from menu $X$, where the object of choice is the \emph{set itself} \citep{BrandtHarrenstein, PetersProtopapas}.

The choice-theoretic approach to revealed preference theory begins with questions about the relationship between the choice function of an agent (or group thereof) and their preferences. What can we infer about agents' preferences from knowing their choice behaviour, as captured by a choice function? Which choice functions can be \emph{rationalized} by a preference ordering over alternatives, or over bundles thereof?

Suppose we now interpret strongest-stable-set operators $\sigma_{\mu,t}$ as choice functions, rather than probabilistically stable revision plans. 
What choice-theoretic properties do they have? It is worthwile to examine which common desiderata for choice functions are satisfied by the representable selection functions. Some of the most prominent properties of choice functions are the following:
\begin{align*}
   (\alpha) \qquad  &\text{if $X\subseteq Y$ then $X\cap \sigma(Y)\subseteq \sigma(X)$} \\
    (\beta) \qquad  &\text{if $X\subseteq Y$ and $x,y\in \sigma(X)$, then $x\in\sigma(Y)$ entails $y\in\sigma(Y)$} \\
    (\gamma) \qquad & \sigma (\cap_{i\in I}X_{i}) \subseteq \sigma (\cup_{i\in I}X_{i}) \\
      (\text{Aizerman})  \qquad &\text{If $Y\subseteq X$ and $\sigma(X)\subseteq Y$ then $\sigma(Y)\subseteq \sigma(X)$}
\end{align*}
We can easily observe the following: 
\begin{observation}
The $(\beta)$ and Aizerman axioms are valid for representable selection functions. Neither $(\alpha)$ nor $(\gamma)$ are valid for representable selection functions (for any threshold value). \label{PROPalphabetaAizerman}
\end{observation}
\begin{proof}
Suppose $\sigma$ is a representable selection function with $\sigma=\sigma_{\mu,t}$ for some probability measure $\mu$ and threshold $t\geq 1/2$. We show that both $(\beta)$ and (Aizerman) hold.

For the validity of $(\beta)$, let $X\subseteq Y$ and $x,y\in \sigma(X)$. Now assume $x\in \sigma(Y)$. We show $y\in \sigma(Y)$. Since $x\in \sigma(Y)\cap X$, the latter set is non-empty. We show that $\sigma(Y)\cap X$ is in fact $(\mu_X,t)$-stable. By the $(\mu_Y, t)$- stability of  $\sigma(Y)$, for every $\omega\in \sigma(Y)$, we have $\mu(\omega)>\frac{t}{1-t}\mu(Y\setminus \sigma (Y))$. So for every $\omega\in \sigma(Y)\cap X$, we have $\mu(\omega)>\frac{t}{1-t}\mu(Y\setminus \sigma (Y))\geq \frac{t}{1-t}\mu(X\setminus \sigma (Y))$, where the latter inequality follows from the fact that $X\subseteq Y$. But this simply means that $\mu(\omega\,|\, (X\setminus (\sigma(Y)\cap X))\cup\{\omega\})>t$, so that $\sigma(Y)\cap X$ is in fact $(\mu_X,t)$-stable. Since $\sigma(X)$ is the strongest $(\mu_X,t)$-stable proposition, we therefore have that $\sigma(X)\subseteq \sigma(Y)\cap X$. Since we already know $y\in \sigma(X)$, we conclude $y\in\sigma(Y)$, as desired.

For the validity of the Aizerman axiom,  suppose $Y\subseteq X$ and $\sigma(X)\subseteq Y$. This means that $\mu(\omega\,|\, (Y\setminus \sigma(X))\cup \{\omega\}  )>t$ for all $\omega\in \sigma(X)$. Since $\sigma(X)\subseteq Y$, this also means $\mu_{Y}(\omega\,|\, (Y\setminus \sigma(X))\cup \{\omega\}  )>t$: both sides of the conditional probability expression are included in $Y$, and conditioning preserves all ratios of states consistent with the second argument. So $\sigma(X)$ is also $(\mu_{Y},t)$-stable, which automatically entails that it includes the strongest $(\mu_{Y},t)$-stable set: that is, $\sigma(Y)\subseteq \sigma(X)$.

For the failure of $(\alpha)$, fix an arbitrary threshold $t\geq 1/2$. Let $\Omega =\{a,b,c\}$. Let $k_t:=t/(1-t)$ and pick some $\epsilon $ with $0<\epsilon <k_t$ (without loss of generality, we can take $\epsilon=1/2$ since $k_t\geq 1$). Define the weight function $m:\Omega\to \mathbb{R}$ by  $m(a) = 1$, $m(b) = k_t+\epsilon$  and $m(c) = k_t(k_t+\epsilon)+\epsilon$. Note that we have 
 \begin{align}
     m(b) &= k_t \cdot m(a) + \epsilon \label{eq1} \\ 
      m(c) &= k_t \cdot m(b) + \epsilon \label{eq2} \\ 
      m(c) &\leq k_t \cdot [m(a)+m(b)]  \label{eq3} 
 \end{align}
Then the normalised measure $\mu(x):= m(x)/\sum_{y\in\Omega}m(y)$ is a probability measure such that $\sigma_{\mu,t}$ fails ($\alpha$). By \eqref{eq1} and \eqref{eq2}, the set $\{b,c\}$ is $(\mu,t)$-stable, because $\mu(a),\mu(b)>k_t \mu(c)$. By \eqref{eq3}, the set $\{c\}$ is not $(\mu,t)$-stable, and a fortiori neither if $\{b\}$ (since $\mu(c)>\mu(b)$). So we have $\sigma(\Omega)=\{b,c\}$. Note also that $\sigma(\{b,c\})=\{c\}$ because of \eqref{eq2}. So, if we let $X:=\{b,c\}$, we have $X\subseteq \Omega$, while $\{b,c\}=X\cap \sigma(\Omega)\not\subseteq \,\sigma (X)=\{c\}$.

The failure of $(\gamma)$ is also shown by this construction. Let $\Omega$ and $\mu$ as above. Let $X=\{a,b\}$ and $Y=\{a,c\}$. Then $\sigma(X\cap Y) = \sigma (\{a\})=\{a\}\not\subseteq \sigma(X\cup Y) = \sigma(\Omega)=\{b,c\}$. 

\end{proof}

The observations above allow for a helfpful comparison between belief change operators. \cite{Collins2} observes that, for finite state spaces, AGM revision operators can be characterised semantically as those that satisfy the axioms $(\alpha)$ and $(\beta)$, while non-probabilistic \emph{imaging} (introduced by \cite{Lewis1976}) is characterised by $(\alpha)$ and the Aizerman axiom. Probabilistically stable revision shares some features with AGM revision\textemdash notably, it satisfies $(\beta)$ and the property of rational montonicity \textemdash but it also shares with imaging the satisfation of the Aizerman axiom. Yet it rejects the $(\alpha)$ property shared by both AGM revision and imaging.

One crucial reason why these constraints on choice functions are significant is because they allow to identify the so-called  \emph{rationalisable} choice functions. Say that a choice function is \emph{rationalisable} whenever there exists some complete and transitive binary relation $R$ on $\Omega$ such that, for every $X\in\mathcal{P}(\Omega)$, we have $\sigma(X)=\{x\in X\,|\, \forall y\in X, xRy\}$: i.e., given a menu $X$ to select from, $\sigma$ picks out the weakly dominating options in $X$. If one interprets completeness and transitivity as rationality constraints on a preference relation, the rationalisable choice functions are those that select the `most preferred' options for some undelrying preference relation. \cite{Sen1971} proved that a selection function is rationalisable exactly if it satisfies $(\alpha)$ and $(\gamma)$. One can of course read the previously mentioned result by van Benthem (Proposition \ref{MINIMOPPROP}) along the same lines: it already follows from that result that representable selection functions cannot be represented as selecting, from each set, its undominated elements for some underlying asymmetric binary relation on options.

Further, \cite{AizermanMalishevski} establish that the functions which satisfy $(\alpha)$ and the Aizerman axiom are exactly those are rationalisable for a \emph{set} of preference orders (which are each required to be complete, transitive, and have a unique maximum for each set): the selection function then collects all the maximal elements for each preference order. Given that they fail $(\alpha)$, the choice functions corresponding to strongest stable set operators cannot be interpreted in this way either. 

Is there any other sense in which we may characterise the strongest stable set operators as being rationalisable via a preference relation? More recent literature on revealed preference studies the selection functions that are representable not by a preference relation on the basic options, but in terms of a preference relation $R\subseteq \mathcal{P}(\Omega)\times \mathcal{P}(\Omega)$ on \emph{sets} of options themselves.

\begin{mydef}[Set-rationalisability \citep{BrandtHarrenstein}]

A selection function $\sigma:\mathcal{P}(\Omega)\to \mathcal{P}(\Omega)$ is \emph{set-rationalisable}  if there exists a binary relation $R$ on $\mathcal{P}(\Omega)$ such that for all $X,Z$:
$$
\sigma(Z)=X \qquad \text{if and only if} \qquad \text{there is no $Y\subseteq Z$ such that }[YRX \text { and not } XRY]
$$
\end{mydef}
Given a relation $R$ defined on sets of alternatives, we say that $Y$ dominates $X$ whenever $Y R  X$ but not $X R Y$. A selection function is \emph{set-rationalizable} if there exists a binary relations on sets such that the function selects, from each set $Z$, exactly the (unique) undominated subset among those in $Z$. \cite{PetersProtopapas} identify the following properties of choice functions as the most relevant in connection with set-rationalizability, where one takes a set of options to be the object of preferences and choice: 

\begin{mydef}[Further properties of choice functions]
\phantom{}
\begin{itemize}
    \item \emph{\textsf{Weak Axiom of Revealed Preference} \textsf{(WARP)}}: \\
    The revealed set-preference relation $R_{\sigma}$ is asymmetric, \\
    with $R_{\sigma}\subseteq \mathcal{P}(\Omega)\times \mathcal{P}(\Omega)$ defined as: $X R_{\sigma}Y$ if and only if $\exists Z\in\mathcal{P}(\Omega)$ such that $Y\subseteq Z$ and $\sigma(Z)=X$.

   \item \emph{\textsf{Weak Irrelevance of Independent Alternatives} \textsf{(WIIA)}}: \\
       If $X\subseteq Y$ and $\sigma(Y)\cap X \neq \emptyset$ then $\sigma(X)\subseteq \sigma(Y)$
     \item \emph{\textsf{Restricted Irrelevance of Independent Alternatives} \textsf{(RIIA)}}: \\
     If $X\subseteq Y$ and $\sigma(Y)\subseteq X$ then $\sigma(X)=\sigma(Y)$
\end{itemize}
\end{mydef}

The first and last property characterise set-rationalisable choice functions:

\begin{theorem}[\cite{PetersProtopapas, BrandtHarrenstein}]
Given a selection function $\sigma:\mathcal{P}(\Omega)\to \mathcal{P}(\Omega)$, the following are equivalent: 
\begin{itemize}
\item $\sigma$ satisfies \textsf{(WARP)};
\item $\sigma$ satisfies  \textsf{(RIIA)}:
\item $\sigma$ is \emph{set-rationalizable}.
\end{itemize}
\end{theorem}

It is straightforward to see that representable selection functions do no meet these conditions. 

\begin{prop}
Representable selection functions satisfy \textsf{(WIIA)}. The axiom \textsf{(RIIA)} (equivalently, \textsf{(WARP)}) is not valid for representable selection functions.
\end{prop}
\begin{proof}
We have already shown the validity of \textsf{(WIIA)} in our argument for $(\beta)$. Suppose $\sigma$ is a representable selection function with $\sigma=\sigma_{\mu,t}$ for some probability measure $\mu$ and threshold $t\geq 1/2$. We showed in our proof of $(\beta)$ that $X\subseteq Y$ and $\sigma(Y)\cap X\neq \emptyset$ entail $\sigma(X)\subseteq \sigma(Y)\cap X$, hence $\sigma(X)\subseteq \sigma(Y)$. For the failure of \textsf{(RIIA)}, let $t\geq1/2$ arbitrary threshold. The counterexample for $(\alpha)$ in Observation \ref{PROPalphabetaAizerman} is also a counterexample to \textsf{(RIIA)}. Take again the very same probability space $(\Omega, \mathcal{P}(\Omega), \mu)$ as in the counterexample for $(\alpha)$. Then $\sigma(\Omega)=\{b,c\}$ but $\sigma(\sigma(\Omega))=\sigma(\{b,c\})=\{c\}$. So we have $X\subseteq \Omega$ and $\sigma(\Omega)\subseteq X$ but $\sigma(X)\neq \sigma(\Omega)$. 
\end{proof}

As a consequence, representable selection functions are not set-rationalizable. As the counterexample illustrates, a quick way to see that \textsf{(RIIA)} and \textsf{(WARP)} fail is simply to observe that strongest-stable-set operators are not idempotent: we may well have $\sigma(\sigma(X))\neq \sigma(X)$. Indeed, highlighting the failure of idempotence in explaining the failure of \textsf{(WARP)} is apt: \cite{PetersProtopapas} show that, for choice functions satisfying \textsf{(WIIA)}, idempotence is equivalent to \textsf{(WARP)}.\footnote{On the belief-revision side, the failure of idempotence means that, even after a learning a proposition that was already believed, the agent's doxastic state changes: this illustrates how stable revision distinguishes propositions that that are \emph{received as evidence and learned with certainty}, as opposed to those that are merely \emph{defeasibly believed, given the evidence}.} 

Thus strongest stable set operators do not admit a revealed preference interpretation in the usual sense of choosing the most preferred options, either in terms of preferences over alternatives or over subsets thereof. This is perhaps a further point in favour of the view that these operators\textemdash and thus, probabilistically stable revision\textemdash are less `qualitative'. 

On the other hand, strongest stable set operators do satisfy a weak version of independence of irrelevant alternatives, and one may hope that they may nonetheless admit some interpretation in terms of revealed preference. And this is the case: however, that interpretation crucially relies on a qualitative notion of stability, rather than picking the most preferred alternatives. It is clear that we can recover a selection function in terms of a preference relation: indeed, the leading idea of our main representation theorem is that we can uniquely recover a probabilistically stable revision plan from the dominance relation $\omega\succ_{\sigma}X$ generated by a selection function $\sigma$. Now, we can generalise this dominance relation slightly (along the lines already sketched in Section \ref{FISHPAR}). Say that $A$ is strongly preferred to $B$, written $A\gg_{\sigma}B$, if $\sigma(A\cup B)\subseteq A\setminus B$: i.e., given the choice among elements in $A\cup B$, the selected options are all in $A$ and none are in $B$.\footnote{This evidently generalises the $\succ_{\sigma}$ relation: $\omega\succ_{\sigma}X$ is equivalent to $\{\omega\}\gg_{\sigma} X$, but $A\gg_{\sigma}B$ can also hold when $A$ is not a singleton.} Now, say $S\subseteq X$ is \emph{dominance-stable with respect to $X$} if $A\gg_{\sigma}B$ holds for every $A\subseteq S$ and $B\subseteq X\setminus S$. That is: after restricting attention to objects in the set $X$, every subset of $S$ is strongly preferred to any subset outside of $S$. Then representable selection functions are exactly those that satisfy:

\begin{itemize}
    \item[] (Dominance-Stability) For every $X$, $\sigma(X)$ is the smallest dominance-stable set with respect to $X$.\footnote{A very similar notion of stability emerges in the microeconomics literature: in their study of set-rationalisable choice functions,  \cite{BrandtHarrenstein} investigate a notion of stability which is similar to Dominance-Stability. Brandt and Harrenstein call a set of alternatives $X$ \emph{$\sigma$-stable} in $A$ if it is \emph{internally stable}, meaning that $X$ is a fixpoint of $\sigma$, and \emph{externally stable}, meaning that $\omega\not\in \sigma(X\cup \{\omega\})$ for all $\omega\in A\setminus X$. Note that, for representable selection functions, although internal stability may fail, every set $\sigma(A)$ is always $\sigma$-stable in $A$: this follows from Dominance-Stability.}
 \item[] (Representability) There exists a numerical utility assignment $u:\Omega\to \mathbb{R}^{\geq 0}$ such that the $\gg_\sigma$ relation corresponds to the cumulative utility of $A$ being sufficiently larger than the cumulative utility of $B$. Here \emph{sufficiently larger than} depends on the threshold parameter $t$, meaning: $$A\gg_{\sigma}B \qquad \text{if and only if}\qquad\sum_{a\in A} \mu(a)> \frac{t}{1-t} \sum_{b\in B}\mu(b)$$
\end{itemize}

In a sense, then, the stability-representable selection functions (probabilistically stable revision plans) arise from the conjunction of two components: the requirement of Dominance-Stability, entirely qualitative in spirit and expressible purely in terms of the strong preference relation $\gg_{\sigma}$; and the more `quantitative' requirement that the strong preference relation be numerically representable in terms of (additive and non-negative) utilities. 

If we interpret representable selection functions as capturing the choice behaviour of an agent (rather than interpreting them as belief revision plans), what kind of choice behaviour do they capture? The property of dominance stability captures a fairly intuitive, if strong, constraint on choice: given a range of options, \emph{any} nonempty subset of the selected options is strongly preferred to \emph{any} subset of non-selected options. 
One can conceive of it as a \emph{trade-invariance} property: the choosing agent would reject any offer to trade a collection of selected options against any collection of non-selected ones. In particular, they would refuse to trade any selected alternative for any collection of non-selected alternatives: for any selected alternative $\omega\in\sigma(X)$ and collection of rejected options $Y\subseteq X\setminus \sigma(X)$, we have $\sigma(Y\cup \{\omega\})=\{\omega\}$, so the agent would refuse trading $\{\omega\}$ for $Y$. The representable selection functions thus capture the choice behaviour of agents who (i) have additive utilities over sets of options, (ii) always select a dominance-stable subset of options, and (iii) always select the smallest dominance-stable subset that is available. 

Are there many decision-theoretic scenarios or situations of economic interest in which such a choice procedure arises naturally? We will not weigh on that question here. But it is of interest to note what our investigation of probabilistically stable revision plans, reinterpreted as choice functions, reveals: representable selection functions admit a certain type of revealed-preference interpretation, which captures the choice behaviour of a cautious agent who selects minimal dominance-stable sets of alternatives: and from that perspective, our representation theorem furnishes the structural conditions on choice functions that characterise exactly the choices of such agents. 

\subsection{Revision policies and imprecise probability?}

There is another bridging role that strongest-stable-set operators can play in the representation of uncertain inference. Each set of probability measures corresponding to a given strongest-stable-set operator\textemdash or, equivalently, a probabilistically stable revision operator, encoding a `belief state with contingency plan'\textemdash is a convex set of probability distributions (though not closed in general). Consider the family of gambles $\{f_{(\omega, X)}\,|\, \omega\not\in X\}$, where 
$$f_{(\omega,X)}(x)=\begin{cases}
1 &\text{ if } x=\omega \\
-t/(1-t)&\text{ if } x\in X \\
0 &\text{ otherwise} 
\end{cases}$$
Given a probabilistically stable revision operator $\sigma$, its set of representing measures $\Delta(\sigma):=\{\mu\in\Delta\,|\, \sigma_{\mu}=\sigma\}$ contains exactly those measures $\mu$ which satisfy $\mathbb{E}_{\mu}[f_{(\omega,X)}]>0$ iff $\omega\succ_{\sigma}X$. By contrast, under the stability rule, regions on the probability simplex corresponding to the distributions that agree only on \emph{unconditional} beliefs are not necessarily convex. Since Bayesian authors often advocate the convexity requirement for credal sets \citep{LEV}, this may indicate that such a Bayesian would be more inclined to see the \emph{full} conditional belief structure generated by the stability rule\textemdash as given by a representable selection function $\sigma$ or a non-monotonic consequence relation \textemdash as a legitimate `qualitative' representation of an agent's belief state (rather than simply taking the raw belief set of the agent). This is not to suggest that there is any compelling imprecise probability semantics for probabilistically stable belief, or that there are any deeper connections between imprecise probability and Leitgeb's account of belief: after all, linear probability constraints giving rise to a convex set are easy to come by, and representors for probabilistically stable revision operators are only a very special case. But if one takes the convexity requirement seriously, Leitgeb's stability rule and its associated revision operator may be good point to start a conversation between (imprecise) Bayesians and 'qualitative' reasoners. 

\subsection{Qualitative revision and the definitional complexity of probabilistic acceptance rules}

We have just explored some alternative settings, beyond the context of belief revision, where strongest-stable-set operators may arise, and where our representation theorem thereby finds an application. As we conclude our investigation, let us return to the subject of belief revision and belief-credence principles from which these results arose. We close this paper on some general considerations that our results raise about the nature of the tracking problem and the relationship between `qualitative' and probabilistic reasoning.

The tracking problem is the problem of providing a belief revision method that tracks Bayesian conditioning. One could worry that a `solution' to the tracking problem may come too cheap. Recall Lin and Kelly's warning about acceptance-induced belief revision plans:

\begin{quote}
It is easy to achieve perfect tracking: just define [\emph{the revised belief set to be $\alpha(\mu_{E})$, the set induced by applying the acceptance rule to one's posterior}]. To avoid triviality, one must specify what would count as a propositional approach to belief revision that does not essentially peek at probabilities to decide what to do. \cite[p. 963]{LIK}
\end{quote}

Lin and Kelly's worry can perhaps be understood as follows. By definition of the tracking problem, every tracking belief revision plan will be of the form $E\mapsto \alpha(\mu_{E})$ for some acceptance rule $\alpha$. In a sense, then, every solution to the tracking problem will amount to specifying an acceptance-induced belief revision policy. But characterising a revision operator explicitly through a probabilistic acceptance rule\textemdash through its openly probabilistic semantics, as it were\textemdash is not a satisfactory solution to the problem: indeed, that would trivialise the tracking problem altogether. We cannot accept any old specification of a tracking belief revision policy. What makes the tracking problem non-trivial is that we look for a way to track probabilistic reasoning with a `qualitative' belief revision policy. Accordingly, for a satisfactory solution to the tracking problem, we must require not only  that the acceptance rule itself be independently motivated as a principle for rational acceptance but also, crucially, that the resulting revision operator admit a self-standing, `qualitative' characterization that does not amount to `peeking at the probabilities'. 

The latter criterion is as intuitive as it is unclear. What makes a revision policy qualitative? What constitutes an illicit `peak at the probabilities' when characterising a belief revision plan? Our study of probabilistically stable revision illustrates how elusive a sharp answer to that question is. On the one hand, our characterisation of stable revision plans in Theorem \ref{GENREPSEL} is undoubtedly at least as `qualitative' as canonical axiomatisations of qualitative probability orders from \cite{KPS} and \cite{SCO}. On the other, the fact that stable revision cannot be characterised in preferential terms (in terms of taking the preferred elements of a plausibility order on basic outcomes) suggests that it is intuitively less qualitative than AGM belief revision and preferential systems of non-monotonic logic, which admit canonical order-based representations. 

There are two distinct questions we can pursue when asking about the extent to which a belief revision plan is qualitative. One is simply whether it admits a `probability-free' characterisation: in essence, whether it admits non-probabilistic semantics. This question of course suffers from the difficulty that it is unclear what makes a semantics sufficiently non-probabilistic to count as qualitative. There certainly is a sense in which, say, the partial order semantics for Lin and Kelly's Shoham-driven belief revision plans\footnote{As discussed in Section \ref{The Ramsey test and tau-models}: recall that the proposal was to declare $\omega$ as at least as plausibe as $v$ whenever the odds ratio $\mu(\omega)/\mu(v)$ reaches a certain threshold.} can be considered purely qualitative: preferential reasoning is arguably a paradigmatic case of qualitative non-monotonic reasoning.  By contrast, our axiomatisation of probabilistically stable revision involves Scott-style `cancellation' axioms of a very probabilistic flavour. The conditional belief structure given by a probabilistically stable revision plan thus admits combinatorial patterns of reasoning\textemdash such as those involving balancedness\textemdash that are characteristic of probabilistic inference. One can be tempted to conclude that inference involving such conditional beliefs encodes patterns of reasoning that ``cry for'' a probabilistic interpretation: they are plausibly read as implicitly encoding numerical probabilistic constraints. That is perhaps one one sense in which these revision operators ``need'' to peak at the probabilities. 

Yet it isn't clear whether the distinction at play \textemdash between those revision operators that require a probabilistic interpretation and those that do not\textemdash can be sharpened and substantiated. A general criterion for qualitativeness is not forthcoming, in part due to the inherent fuzziness of the target. This reflects a broader conceptual difficulty in isolating the nature of `qualitative' reasoning, one which has been observed in the study of probability logics: it is hard to say what it means, for a belief revision operator or any other inference system, to \emph{require} a probabilistic semantics. When does one \emph{need} to peak at probabilities? To illustrate the difficulty: as has been observed by \cite{APAL}, even paradigmatically quantitative probability logics involving explicit addition of probability terms can be given arguably `qualitative' semantics in terms of concatenation of stings, and thus may not \emph{require} a probabilistic interpretation, however natural such an interpretation may be.\footnote{Suppose we have a probability logic in which we allow Boolean reasoning about comparisons of probability terms, where the terms can involve explicit addition of probabilities. That is, suppose we allow boolean reasoning about judgments of the form $\sum_{i\in I} P(\alpha_i) < \sum_{i\in J} P(\beta_i)$ with $\alpha_i$, $\beta_i$ boolean formulas representing events in a (finite) probability space. One can prove a completeness result for this logic by interpreting each probability term as an unary string, with addition interpreted as concatenation, and the less-than relation $<$ as string containment: see \cite[p. 37]{APAL}. There, the authors propose that a more meaningful distinction in the context of probabilistic reasoning, which also tracks considerations of computational complexity, is not between qualitative and quantitative reasoning, but between purely additive probabilistic reasoning and reasoning implicitly involving \emph{multiplicative} constraints.} By these standards, probabilistically stable revision qualifies even more easily for `qualitative' status. It follows from our discussion that probabilistically stable revision can be characterised using only first-order quantification over events and a `dominance' relation $\omega \succ X$ (more on this below). Our axioms indeed ensure that the domination relation \emph{can} be interpreted as `sufficiently more probable than' (for a given threshold), but one could just as well treat the axioms for $\succ$ as primitive constraints on qualitative plausibility judgments, much in the spirit of the measurement-theoretic approach to probability. The mere numerical representability of the selection function/dominance order does not force a probabilistic interpretation thereof, any more than the possibility to represent a partial order on states as an odds-ratio order forces a probabilistic interpretation of Shoham-driven revision.\footnote{Indeed, note that the dominance order $\succ$ involved in probabilistic stability captures essentially the same kind of probability judgments than those involved in the belief revision policies of \cite{LIK}. The only difference is that, while the latter relies on odds-ratio $\mu(\omega)/\mu(v)$ between basic outcomes only, the former additionally keeps track of odds ratios of the form $\mu(\omega)/\mu(X)$ between a basic outcome and a general event (we have $\omega\succ X$ whenever $\mu(\omega)/\mu(X)$ exceeds a threshold $\frac{t}{1-t}$, for threshold $t$).}  All this highlights the special status of probabilistically stable revision as a halfway station, sharing features of both qualitative belief revision and reasoning with comparative probability orders. 

This leads to another, more tractable question one can ask about the dependence of a belief revision policy on probabilities. Rather than asking whether a belief revision policy admits a `qualitative' semantics at all (whatever that might ultimately mean), we can ask how much probabilistic structure one's belief revision policy is sensitive to \emph{under its given acceptance rule semantics}. 

Our modest proposal is that progress can be made on that question by classifying acceptance rules and their associated revision operators by their definitional complexity in a sufficiently expressive probability logic. A rough outline of how this can be done is as follows. Suppose we have a first-order language to talk about probability spaces, in which we can quantify over events in the algebra, and with enough resources to contain a modest amount of real arithmetic. For concreteness, take the first-order language of boolean algebras $\mathcal{L}_{BA}$ with a function symbol $\mathsf{P}$ to be interpreted as a probability measure on the algebra. Let the signature for $\mathcal{L}_{BA}$ contain the boolean operations of complement ($\cdot\,^{\bot}$), meet ($\sqcap$) and join ($\sqcup$), and a constant $\bot$ for the bottom element of the algebra (for notational simplicity, include also the boolean inclusion relation $\sqsubseteq$).  Additionally, we take function symbols for basic arithmetical operations to be interpreted over the ordered real field $(\mathbb{R}, \cdot, +, \leq, 0,1)$. Aside from ordinary $\mathcal{L}_{BA}$ formulas, we also allow formulas that are recursively built from basic expressions of the form $p_{1}\big(\mathsf{P}(\beta_{1}),\dots,\mathsf{P}(\beta_{n})\big) \leq p_{2}\big(\mathsf{P}(\beta_{i}),\dots,\mathsf{P}(\beta_{k})\big) $, where $p_{i}$'s are polynomial expressions in the signature $(\cdot, +, 0,1)$ (with, say, constants for rationals) and the $\beta_i$'s are terms in  $\mathcal{L}_{BA}$ standing for events (this is similar to the systems investigated by \cite{FHM} and \cite{SPE}). Call the resulting language $\mathcal{L}$.

We can interpret formulas in $\mathcal{L}$ over structures of the form $(\mathbb{B}, \mu)$ where $\mathbb{B}$ is a (finite) boolean algebra and $\mu$ a probability measure on it. Formulas in $\mathcal{L}_{BA}$ are evaluated in $\mathbb{B}$, and expressions of the form $p_{1}\big(\mathsf{P}(\beta_{1}),\dots,\mathsf{P}(\beta_{n})\big) \leq p_{2}\big(\mathsf{P}(\beta_{i}),\dots,\mathsf{P}(\beta_{k})\big) $ hold if the resulting inequalities are true once each $\mathsf{P}(\beta)$ has been replaced by its intepretation $\mu(\e{\beta})$ (here $\e{\beta}\in\mathbb{B}$ is the interpretation of the Boolean term $\beta$). As an example:  $(\mathbb{B}, \mu)\vDash \forall x \big((x\neq \bot)\rightarrow \mathsf{P}(x)>\mathsf{0}\big) $ means that the measure $\mathsf{P}$ is regular: for all $A\in\mathbb{B}$, if $A$ is not the bottom element $\bot_{\mathbb{B}}$ in $\mathbb{B}$, then $\mu(A)>0$.
 
Acceptance rules can then be classified by the definitional complexity of the set of accepted elements in the underlying algebra. Here are some examples.

\begin{ex}[The Lockean rule]
The Lockean rule $\lambda$ can be uniformly captured by an atomic formula:
$$
B^{\lambda}(x) := \mathsf{P}(x)\geq \mathsf{q}
$$
where we take the parameter $\mathsf{q}$ as a constant symbol for $q\in\mathbb{Q}$.
\end{ex}

\begin{ex}[The stability rule]
The stability rule $\tau$ can be captured by the following formula:
$$
B^{\tau}(x) :=\exists y ( y \sqsubseteq x \wedge \textsf{S}(y))
$$
where $\textsf{S}(y)$ is a stability predicate which can be defined as
$$
\textsf{S}(y) := \forall z (y\wedge z \neq\bot \to \mathsf{P}(y\wedge z)> \mathsf{q}\cdot \mathsf{P}(z)) 
$$
or, equivalently, as $\textsf{S}(y)\Leftrightarrow\forall z \big((z\neq \bot \wedge z\sqsubseteq y) \rightarrow (1-\mathsf{q}) \mathsf{P}(z)> \mathsf{q}(1-\mathsf{P}(y)\big)$. The property of being an accepted proposition is a $\exists\forall$-condition. Similarly, the property of being the strongest stable proposition is given by 
\begin{align*}
\textsf{SS}(x) &:=\textsf{S} (x) \wedge \forall y \big(\textsf{S} (y) \rightarrow x\sqsubseteq  y\big) \\
& \Leftrightarrow \forall w \forall y \exists z\Big( \psi(x,w)\wedge  \big(\psi(y,z) \to x\sqsubseteq y \big)\Big).
\end{align*}
where $\psi(y,z)=\big(y\wedge z \neq \bot \to \mathsf{P}(y\wedge z)> \mathsf{q}\cdot \mathsf{P}(z)\big)$. This amounts to a  $\forall\exists$-condition.
\end{ex}

\begin{ex}[Lehrer's rule]
\cite{LEH} proposes the following acceptance rule: accept all propositions that are highly probable and whose probability cannot be lowered by learning an even more probable proposition. We can capture this acceptance rule by the following formula: 
$$
B^L(x):= \mathsf{P}(x)\geq \mathsf{q} \,\wedge \neg \exists y \big( \mathsf{P}(x)\leq \mathsf{P}(y) \wedge \mathsf{P}(x\wedge y)< \mathsf{P}(x)\cdot \mathsf{P}(y)\big)
$$
\end{ex}

We can proceed in the same way to analyze induced belief revision policies. Given an acceptance rule $\alpha$, say that a formula $\varphi(x,y)$ defines the belief revision policy induced by $\alpha$ if and only if, in every model, it captures exactly the function $A\mapsto \alpha(\mu_{A})$, understood as the set of pairs $(A,B)$ such that $B$ is the strongest accepted proposition upon learning $A$. That is: for every model $(\mathbb{B}, \mu)$, and any $A,B\in \mathbb{B}$ we have 
\begin{center}
$(\mathbb{B}, \mu)\models \varphi(A,B)$ if and only if 
$\alpha(\mu_{A})=B $
\end{center}
For instance, we can give an explicit definition of probabilistically stable revision. 
\begin{ex}[Probabilistically stable revision]
Consider the following $\mathcal{L}$-formulas:
\begin{align*}
 \mathsf{S}(x,y)&:=  x\sqsubseteq y \wedge \forall z \big((z\neq \bot\wedge z\sqsubseteq x)\to (1-\mathsf{q})\cdot\mathsf{P}(z)> \mathsf{q}\cdot \mathsf{P}(y\sqcap x^{\bot})\big)\\
 \varphi(x,y) &:=  \mathsf{S}(x,y) \wedge \forall z (\mathsf{S}(z,y)\to x\sqsubseteq z)  \\
  B^{\tau}(x,y) &:=  \exists z ( \mathsf{S}(z,y)\wedge z\sqsubseteq x)
\end{align*}
The formula $ \mathsf{S}(x,y)$ defines the relation ``$x$ is probabilistically stable conditional on $y$'', while $ \varphi(x,y)$ defines the belief revision policy induced by the stability rule, i.e. probabilistically stable revision for threshold $q$. The formula  $B^{\tau}(x,y)$ defines the relation ``$x$ is believed conditional on $y$'' under the stability rule.
\end{ex}

Classifying acceptance rules and revision policies in terms of their definitional complexity allows to bring to the fore important logical differences between acceptance rules. First, it helps identify two sources of complexity for an acceptance rule or revision policy: one is the quantificational structure of its defining condition, while the other is the algebraic complexity of probabilistic condition involved (captured by quantifier-free formulas).

For example, in the cases listed above, the Lockean rule is captured by a atomic formula, while the stability rule is given by a $\Sigma_2$ condition ($\exists\forall$). This is one sense in which these explicit definitions syntactically capture the `higher-order' nature of the stability rule: as opposed to several other acceptance rules suggested in the literature, the acceptance of a hypothesis $H$ does not merely depend on a `\emph{local}' property\textemdash e.g. the mere value of $\mu(H)$\textemdash but takes into account a more global property of the probability space that $H$ lives in, that of admitting a probabilistically stable subset. 

We can also see that the stability rule nonetheless only requires quantification over formulas that do not involve multiplication of probability terms of the form $\mathsf{P}(\beta_1)\cdot \mathsf{P}(\beta_2)$: in particular, we can say that stable belief is captured by a $\Sigma_2(\mathcal{L}_{\text{linear}})$ formula: an $\exists\forall$ condition whose quantifier-free part encodes \emph{linear} constraints on probability terms. And indeed, it can be seen from the defining formulas above that \emph{model-checking} for any instance of $\mathsf{S}(x)$ is a linear problem: given a fixed probability space (model) $(\mathbb{B}, \mu)$, checking whether the stability property $\mathsf{S}(A)$ holds of any particular element $A\in \mathbb{B}$  amounts to verifying a conjunction of \emph{linear} inequalities $\bigwedge_{\bot\neq B\sqsubseteq A} \,\mu(B)> \frac{q}{1-q} \mu (A^c)$.\footnote{It is of course sufficient to check these inequalities for algebraic \emph{atoms} below $B$.} Similarly,  given a fixed probability space, determining the probabilistically stable revision plan amounts to solving a linear problem. By contrast, in the case of Lehrer's rule, the belief condition $B^L(x)$ is given by a universal quantifier followed by \emph{multiplicative} constraints between probability terms. There, the conditions $\mathsf{P}(x\wedge y)< \mathsf{P}(x)\cdot \mathsf{P}(y)$ cannot be captured by inequalities that are linear in probability terms. We may say that Lehrer's rule is captured by a $\Pi_{1}(\mathcal{L}_{\text{conf}})$ condition, where $\mathcal{L}_\text{conf}$ consists of the quantifier-free formulas specifying \emph{confirmation constraints} $\mathsf{P}(x\,|\,y)\geq P(x)$. 

These examples illustrate how an explicitly logical approach to classifying acceptance rules offers a framework within which to investigate the question of how much probabilistic structure an acceptance rule or revision policy depends on. How much information about the entire measure do we need to have access to in order to determine, for instance, whether a given hypothesis is accepted? Given an acceptance rule, what kind of constraints on probability distributions can one encode by (conditional) belief statements (are they linear, semi-algebraic, exponential)? Further, these considerations go hand-in-hand with considerations of complexity: how complex is inference involving the corresponding notions of acceptance? One can of course relate the definitional complexity of an acceptance rule or belief revision policy to the computational complexity of checking whether the condition obtains in a given model. Further still, once one knows in which fragment of a probability logic one can define statements of (conditional) belief, one can leverage results about the complexity of decision problems for various probability logics to address questions of computational complexity for the corresponding logics of (conditional) belief. In this manner we might, in some cases, establish a ready connection between the definitional complexity of an acceptance rule and the complexity of decision problems for the resulting doxastic logic or logic of belief revision.

Much recent work in probability logic is relevant to these questions. \cite{APAL} investigate the boundary between qualitative and quantitative reasoning in the setting of probability logics, establishing a robust complexity boundary between additive probabilistic reasoning (corresponding to linear inequality constraints) and multiplicative reasoning (corresponding to more general polynomial constraints). Further, \cite{MOS} establishes that even the simple logic of probabilistic confirmation, which makes an appearance in the definition of Lehrer's rule above, has an $\exists\mathbb{R}$-complete decision problem (which corresponds to the existential theory of the reals; it is as hard as the `quantitative' logic of inequalities between polynomial probability terms). \cite{SPE, SPE2}, on the other hand, develops in more depth the theory of probability logics with quantifiers over events, focusing on results about expressivity and complexity that may well be fruitfully employed for our purposes.

It should of course be mentioned that, for more sophisticated rules, we may need a language more complex than comparisons of polynomial expressions over probability terms (and quantification over a boolean algebra). This is so if, for example, one relies on a notion of acceptance that relies on information-theoretic entropy. Here is one family of acceptance rules where this might occur:

\begin{ex}[Voronoi rules: acceptance rules induced by statistical distance measures]
To each consistent proposition $X$ in an algebra of propositions $\mathfrak{A}$ assign a \emph{representative measure} $\rho^{X}$ where $X$ is the strongest accepted proposition: let $\alpha(\rho^X)=X$. Next, fix a favourite statistical distance metric or divergence $D:\Delta_{\mathfrak{A}}\times \Delta_{\mathfrak{A}} \to [0,\infty]$ between probability measures: e.g., symmetrised Kullback-Leibler divergence, total variation distance, or Heilinger distance. Define the acceptance zones $[X]:=\{\mu\in\Delta_{\mathfrak{A}}\,|\,\alpha(\mu)=X\}$ as the cells of a Voronoi diagram centered around the representative probability measures $\{\rho^{X}\,|\, \emptyset\neq X\in\mathfrak{A}\}$. In other words: given a measure $\mu$, determine which of the canonical representatives $\rho^X$ ($X\in\mathfrak{A}$) the measure $\mu$ is closest to, in terms of the distance $D$; then $X$ is the agent's strongest accepted proposition. 

One instance of this construction is the following. Define the canonical representative of $X$ as the maximum entropy measure $\mu\in\Delta_{\mathfrak{A}}$ such that $\mu(X)\geq t$, where $t$ is a fixed threshold: 
$$
\rho^{X}:= \text{\normalfont {argmax}}_{\{\mu\in\Delta_{\mathfrak{A}} \,|\,\mu(X)\geq t\}}\, H(\mu) 
$$
where $H(\mu)= -\sum_{\omega\in \Omega}\mu(\omega)\log\mu(\omega)$ is the information-theoretic \emph{entropy} of $\mu$. The measure $\rho^X$ can be thought of as the most uninformative or `equivocal' measure among all those that assign measure at least $t$ to $X$. Then the acceptance rule is given by:
$$
\alpha(\mu)= \text{\normalfont{argmin}}_{X\in\mathfrak{A}} \,D(\rho^X,\mu) 
$$
where $D$ is some particular measure of statistical distance or divergence. A natural choice is to let $D$ denote the Kullback-Leibler divergence (equivalently, relative entropy):$$
D(\mu, \nu) = \sum_{\omega\in \Omega} \mu(\omega) \log \frac{\mu(\omega}{\nu(\omega)}. 
$$
The agent's belief set is then $X$ exactly if the representative measure $\rho^X$ is the closest one to their credence function under Kullback-Leibler divergence (among all the representative measures).\footnote{Kullback-Leibler divergence is a Bregman divergence, so Voronoi diagrams induced by it are an instance of what is known as Bregman-Voronoi diagrams studied in computational geometry: see \cite{VOR, VOR2}. When one uses a non-symmetric divergence like Kullback-Leibler, it makes a difference to the Voronoi regions whether one defines the nearest representative as the one minimising $D(\rho^X,\mu)$ or the one minimising $D(\mu,\rho^X)$. Also, when $\mu$ is at an equal distance between representative points\textemdash lies on \emph{bisector} between Voronoi regions\textemdash a tie-breaking rule might need to be added. We will not dwell on different ways to achieve this, since our aim is less to motivate this particular rule than to raise a logical point about definability for the family of rules induced by statistical distance.}

In general, depending on how complex the definition of canonical measures $\rho^X$ and the statistical distance $D$ are, the definition of the acceptance rule can outstrip the expressive capabilities of a simple logic of polynomial probability terms. For acceptance rules contructied in this way, the strongest accepted proposition could be defined by a formula of the form
$$
B^{\alpha}(x):= \forall y \,\varphi(x,y)
$$
where $\varphi(x,y)$ captures the property that the given measure $\mu$ is at least as close to the representative measure for $x$ as it is to the representative measure for $y$. For the latter property to be definable, there must be a formula $\varphi(x,y)$ such that for any model $(\mathbb{B}, \mu)$, we have $(\mathbb{B}, \mu)\models  \varphi(A,B)$ for $A,B\in\mathbb{B}$ if and only if $D(\mu, \rho^{A})\geq D(\mu, \rho^B)$. For non-linear distance/divergence measures this will in general not be definable by quantification over events and polynomial probability terms: a more expressive probability logic will then be needed, such as one allowing exponentiation in the construction of terms and/or quantification over real numbers. 
\end{ex}

Some rules proposed in the literature are given definitions which  use quantification over sets of events: thus another way in which capturing an acceptance principle may require more expressivity is that it may require second-order quantification.

\begin{ex}[Douven's rule]

 \cite{DOU2} proposes the following rule: accept $H$ if and only if $\mu(H)>t$ and $H$ does not belong to any \emph{probabilistically self-undermining} set of propositions (for some fixed threshold $t\in(0.5,1]$). A set of propositions $\mathcal{G}\subseteq \mathfrak{A}$ is probabilistically self-undermining for a threshold $t$ if for all $X\in\mathcal{G}$, $\mu(X)>t$ and there is some $n\in\mathbb{N}$ such that $\mu(X|\bigcap^{k}_{i=1} Y_{i})\leq t$ for any collection $\{Y_{i}\}_{i\leq k}$ of at least $n$ propositions in $\mathcal{G}$.
\end{ex}

The upshot of the examples discussed here is simply that one can get more insight into the behaviour of an acceptance rule\textemdash which features of a probability space it is sensitive to, how complex a logic of belief (revision) it generates, and so on\textemdash by looking at its explicit definitional complexity in a probability logic. On the philosophical side, such results might even shed light on the vexed question of whether, and to what extent, all-or-nothing beliefs help simplify reasoning \citep{STAFF}.\footnote{One would be right to retain some skepticism about a flat-footed appeal to worst-case complexity to make a point about cognitive economy; nonetheless, a logical approach like the one sketched here can offer some suggesting results in that direction.} At any rate, there is good reason to expect the framework of probability logics to provide a fruitful perspective on the acceptance rules and belief revision policies proposed in the literature. We leave further work in that direction for another occasion. 

\section{Conclusion} 

A solution to the tracking problem based on Leitgeb's stability rule for belief requires a qualitative charactarisation of probabilistically stable revision: one must axiomatise the belief revision plans induced by Leitgeb's stability rule for belief. Our main theorem (Theorem \ref{GENREPSEL}) does just this: it gives necessary and sufficient conditions for a selection function to be representable as a probabilistically stable revision operator. As such, it yields selection function semantics for the logic of probabilistically stable belief revision. 

The key to this representation theorem is a connection to the theory of comparative probability orders. Probabilistically stable revision turns out to be essentially characterised by Rational Monotonicity, a very weak form of the \textsf{Or} rule (case reasoning), and a Scott-style cancellation axiom in the style of comparative probability theory. 

Beyond solving the characterisation problem for probabilistically stable revision, we have obtained some results on the way that are of independent interest for the theory of comparative probability orders. First, we gave a representation result (Proposition \ref{strongrep}) giving necessary and sufficient conditions for the \emph{joint weak representability} of two relations as strict and, respectively, non-strict probability comparisons\textemdash which also helped settle a question about the most general conditions for \emph{strong agreement} of a probability measure with a comparative probability order. We also investigated how the characterisation of probabilistically stable revision changes as a function of the stability threshold. Our main representation theorem gives a parametric characterisation of probabilistically stable revision for \emph{any} choice of rational threshold. The essential component there is the generalisation of the Scott-style finite cancellation axiom based on a more general, threshold-dependent notion of balancedness for sequences of events. Importantly, this result directly leads to a general axiomatization of the `$q$ times more probable' relation for any rational number $q$.  

We close with the hope that, whether or not one agrees with Leitgeb on the significance of the stability rule specifically for the analysis of \emph{belief}, the present work demonstrates that the notion of probabilistic stability is worthy of independent logical interest. First, as we observed, probabilistically stable revision plans represent the agent's belief state at an interesting `medium' level of granularity, exhibiting both features that are typical of paradigmatically `qualitative' belief revision operators as well as validating reasoning patterns of a distinctly probabilistic flavour. The study of probabilistically stable revision thus serves as a fertile testbed for sharpening intuitions as to what constitutes a qualitative belief revision method, and for exploring the borderlands between logical belief revision and probabilistic reasoning. Second, probabilistic stability is a concept as fruitful as it is elementary and, beyond its attested philosophical applications (e.g. in \cite{SKY} and \cite{LEIBOOK}), we have seen that it emerges in a variety of other contexts, such as voting games and the theory of revealed preference, where our representation theorem finds natural applications. If nothing else, our investigation of probabilistically stable revision, the results we have obtained in the process, and the conceptual questions they raise, serve to illustrate the productive interplay between logical and measurement-theoretic methods in the study of belief dynamics.

\bibliographystyle{plainnat}
\addcontentsline{toc}{section}{Bibliography} 
\bibliography{bib2}

\end{document}